\documentclass[aps,pra,reprint,superscriptaddress,preprintnumbers]{revtex4-2}
    \usepackage[utf8]{inputenc}
    \usepackage[english]{babel}
    \newtheorem{theorem}{Theorem}

    \newtheorem{lemma}{Lemma}
    \usepackage{float}
    \newtheorem{definition}{Definition}
    \newtheorem{remark}{Remark}
    \usepackage{setspace}
    \usepackage{booktabs}
    \usepackage{multirow}
    \usepackage{amssymb}
    \usepackage{qcircuit}
    \usepackage{lipsum}
    \usepackage{tabularx}
    \usepackage{natbib}
    \usepackage{graphicx}
    \usepackage{rotating}
    \usepackage{float}
\usepackage{tikz}
\usepackage{pgfplots}
\usepackage{placeins}
\usepackage[ruled,vlined]{algorithm2e} 

\pgfplotsset{compat=newest} 

\newtheorem{proposition}[theorem]{Proposition}

    \usepackage{makecell}
    \setcellgapes{2.3pt}
    \usepackage[figurename=Fig., justification=RaggedRight]{caption}
    \usepackage{subcaption} 
    \newenvironment{proof}{\textit{Proof.}}{\hfill$\square$}
    \usepackage{physics}

    \usepackage{amsmath}
    \usepackage{xcolor}
    \usepackage{hyperref}

    \usepackage{adjustbox,expl3,etoolbox}
    \letcs\replicate{prg_replicate:nn}

\usepackage{cleveref}    

\makeatletter
\newcommand{\sectionnotoc}[1]{%
  \begingroup
  \def\addcontentsline##1##2##3{}%
  \section{#1}%
  \endgroup
}

\newcommand{\starsectionnotoc}[1]{%
  \begingroup
  \def\addcontentsline##1##2##3{}%
  \section*{#1}%
  \endgroup
}

\newcommand{\subsectionnotoc}[1]{%
  \begingroup
  \def\addcontentsline##1##2##3{}%
  \subsection{#1}%
  \endgroup
}
\makeatother

\begin{document} 
\title{End-to-End PDE-Based Quantum Algorithms for Multi-Asset Option Pricing under Local and Stochastic Volatility}

\author{Nikita Guseynov}
\email{guseynov.nm@gmail.com}
\affiliation{Global College, Shanghai Jiao Tong University, Shanghai 200240, China.}

\author{Nana Liu}
\affiliation{Global College, Shanghai Jiao Tong University, Shanghai 200240, China.}
\affiliation{Institute of Natural Sciences, School of Mathematical Sciences, Shanghai Jiao Tong University, Shanghai 200240, China}
\affiliation{Ministry of Education Key Laboratory in Scientific and Engineering Computing, Shanghai Jiao Tong University, Shanghai 200240, China}

\author{Chi Seng Pun}
\affiliation{SPMS, Nanyang Technological University, 21 Nanyang Link, Singapore 637371}

\author{Tushar Vaidya}
\email{tushar.vaidya@ntu.edu.sg}
\affiliation{SPMS, Nanyang Technological University, 21 Nanyang Link, Singapore 637371}

\begin{abstract}
Multi-asset option pricing under local- and stochastic-volatility models leads naturally to high-dimensional parabolic PDEs. We develop an end-to-end quantum PDE framework for European option pricing under local-volatility Black--Scholes and Heston models. The framework takes classical contract and model data as input and returns classical estimates of selected option values. We solve the pricing PDEs after finite-difference discretization on spatial grids. For $N=2^n$ grid points per spatial direction and $d$ assets, the end-to-end gate complexity for single-point recovery, counted in elementary CNOT gates and one-qubit Pauli-axis rotations, has leading grid-size dependence $\widetilde{O}(d^2 N^{2+d/2})$ for local-volatility Black--Scholes and $\widetilde{O}(d^2 N^{d+2})$ for Heston. Relative to grid-based finite-difference baselines, these scalings correspond to polynomial improvement factors $N^{d/2}$ and $N^d$, respectively. These estimates translate to Clifford+T resources via standard compilation. We complement the complexity analysis with numerical benchmarks against standard classical methods. In the Heston setting, the framework recovers option prices across strikes together with the associated implied-volatility smile/skew. Overall, this work provides a complete end-to-end quantum pricing pipeline with explicit resource accounting and theoretical performance guarantees.

\end{abstract}
\maketitle

\sectionnotoc{Introduction}

Quantum computing \cite{FEYNMAN} has emerged as a promising paradigm for accelerating computational tasks that are difficult for classical computers. Its potential is especially visible in mathematical and scientific computing, where many problems reduce to high-dimensional linear algebra \cite{hhl}, Hamiltonian simulation \cite{PhysRevLett.114.090502}, differential equations \cite{Berry2017}, and numerical integration \cite{10.1098/rspa.2015.0301}. In recent years, quantum algorithms have been actively explored for partial differential equations (PDEs) \cite{Childs2021HighPrecisionPDE,Berry2017,Liu2021DissipativeNonlinearDE,PRL2024,PRA2023,analog,cao2023quantum,Jin2025IllPosedSchrodingerization,Jin2024BoundaryInterfacePDE,guseynov2025quantumframeworksimulatinglinear,guseynov2024efficientPDE}, with applications ranging from physics \cite{costa2019quantum,gaitan2020finding} and engineering \cite{jin2023time,linden2022quantum,sato2024hamiltonian} to finance \cite{Stamatopoulos2020OptionPricing,gonzalez2023efficient}. In computational finance, quantum methods have been proposed for pricing \cite{Tang2021CDOQuantum,CarreraVazquez2021StatePreparation,Stamatopoulos2020OptionPricing,Martin2021PricingIBM,Rebentrost2018QuantumFinance}, hedging \cite{Cherrat2023QuantumDeepHedging,Stamatopoulos2022QuantumMarketRisk}, calibration \cite{Fan2007OptionPricingQEA}, and risk management \cite{Woerner2019QuantumRisk,Egger2021CreditRisk}.

In this paper, we focus on option pricing as a concrete and practically relevant benchmark for assessing the potential advantages of quantum algorithms, and we take the PDE-based pricing approach.

In computational finance, option prices are commonly computed either by pathwise Monte Carlo simulation based on the underlying risk-neutral SDE dynamics, or by solving the associated pricing PDE \cite{duffy2006fdm,glasserman2003monte,hirsa2024computational}. The Monte Carlo route is widely used, but it is not always the most natural formulation when the contract is strongly tied to terminal and boundary conditions, high dimensionality, early-exercise features, or the recovery of a price surface rather than a single expected payoff. Therefore, PDE-based pricing remains a promising direction. It represents the no-arbitrage pricing problem directly through a deterministic evolution equation, with the payoff imposed as terminal data and the financial constraints encoded by boundary conditions \cite{wilmott1995mathematics}. 

In this work, we develop an end-to-end quantum PDE framework for option pricing, covering local-volatility Black--Scholes \cite{BlackScholes1973} and stochastic-volatility Heston models \cite{heston1993closed}. A schematic overview of the quantum pricing workflow is shown in Fig.~\ref{fig:quantum_scheme_without_postselection}. Starting from the pricing PDE, we apply finite-difference methods to discretize the spatial variables. This turns the continuous pricing equation into a sparse finite-dimensional ODE system, where the sampled payoff at maturity becomes the initial vector for the backward pricing evolution.  We then prepare this payoff vector as an initial quantum state, evolve the discretized pricing dynamics quantumly, and recover selected option values from the final state.

\begin{figure*}[t!]
    \centering
    \includegraphics[width=0.99\linewidth]{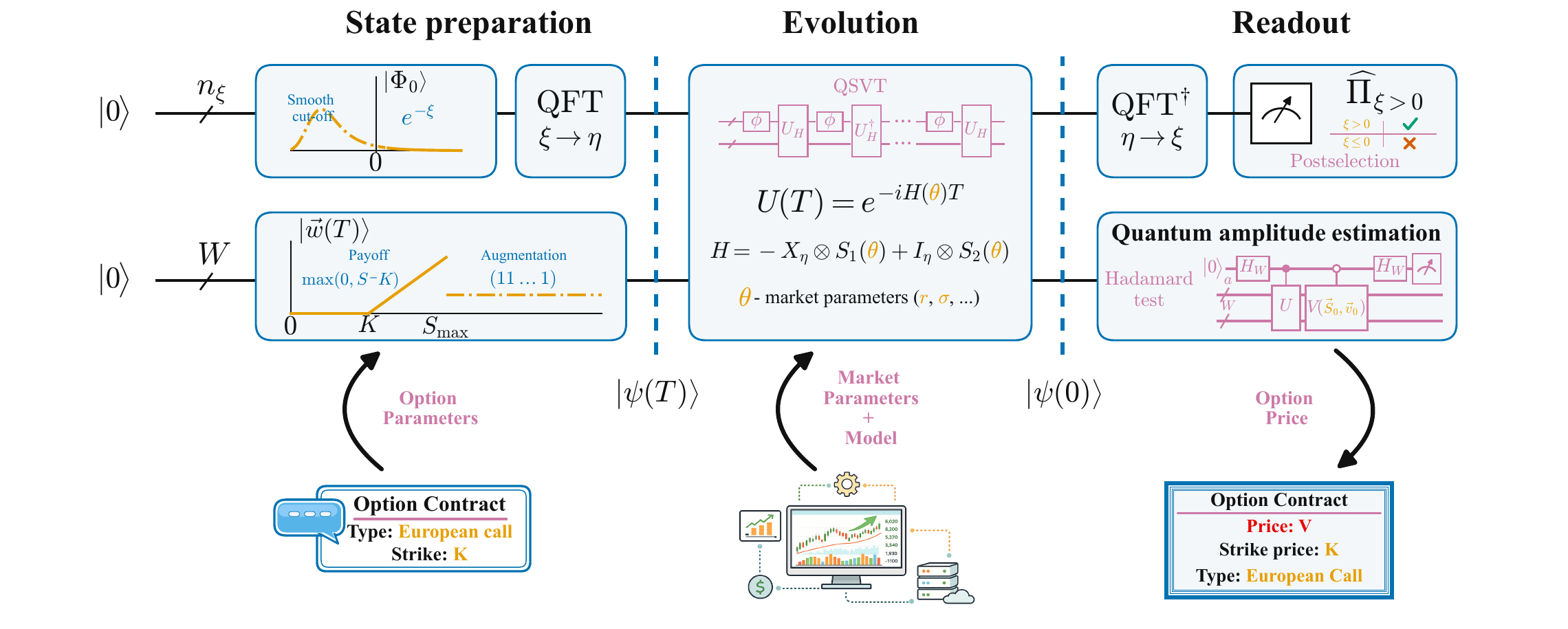}
    \caption{\textbf{Quantum workflow for solving the Black–Scholes \eqref{eq:background-local-vol-bs} and Heston models \eqref{eq:background-heston-pde}.} The $W$-qubit main register encodes option values over the stock-price grid (and volatility for the Heston PDE); the auxiliary $n_\xi$-qubit register implements the Schrödingerisation technique, mapping the dynamics to a unitary Schrödinger evolution. Both registers are initialized in the $\ket{0}$ state. The state preparation stage prepares (i) the smooth initialization $\ket{\Phi_0}$, a $C^\infty$ profile that coincides with $e^{-\xi}$ for $\xi>0$, (ii) the payoff state (augmented by a constant component) $\{\vec V(T),1\}^\top$. The evolution stage includes Hamiltonian evolution under the Schrödinger system induced by the ODE system \eqref{eq:background-discrete-ode}, inducing $V(T)\!\to\! V(0)$. The postselection stage includes a projective measurement of the auxiliary register onto the subspace $\xi>0$ (i.e., measure bit strings $> N_\xi/2$). After these steps, the $W$-qubit main register contains the normalized solution to \eqref{eq:background-discrete-ode}, namely the normalized option prices $\vec V(0)$, which can be read out by a method based on quantum amplitude estimation \cite{doi:10.1137/1.9781611977561.ch12}}.
\label{fig:quantum_scheme_without_postselection}
\end{figure*}

After these steps, the input to the overall algorithm is classical contract and model data, while the output is a classical estimate of selected option values. The resource analysis is carried out for the complete pricing pipeline, including quantum state preparation of the payoff, quantum evolution of the discretized dynamics, normalization recovery, and readout. Under this accounting, the leading grid-size dependence improves relative to classical finite-difference PDE solvers. Expressed in terms of the grid size $N=2^n$, the improvement is polynomial in $N$: the multi-asset Black--Scholes pipeline improves the leading dependence by a factor $N^{d/2}=2^{nd/2}$, while the multi-asset Heston pipeline improves it by a factor $N^d=2^{nd}$, up to suppressed polynomial and polylogarithmic factors. Thus, the end-to-end pipeline exhibits a polynomial quantum advantage in the grid size $N$, with improvement factors $N^{d/2}$ for multi-asset Black--Scholes and $N^d$ for multi-asset Heston.

We further validate the framework numerically. The quantum PDE workflow is compared against analytical (or semi-analytical \cite{heston1993closed,mathworks_financial_instruments_toolbox_r2025a}), finite-difference, and matrix-exponential baselines (we clarify the classical numerical methods used in Appendix~\ref{appendix:numerical methods}). In the Heston case, the workflow produces option prices across strikes and recovers the associated non-flat implied-volatility smile/skew. Thus, the numerical results demonstrate that the workflow is able to recover selected option prices together with the market-relevant implied-volatility smile/skew produced by the Heston stochastic-volatility model.

\bigskip\medskip
\sectionnotoc{Background}
\label{sec:background}

The central object considered in this paper is an option: a derivative contract that gives its holder the right, but not the obligation, to buy or sell an underlying asset at a predetermined price $K$, called the strike, at a specified time $T$, called the maturity. A call option gives the right to buy the asset, whereas a put option gives the right to sell it.

The most basic examples are vanilla call and put options. In their European form, the payoff is evaluated at maturity and depends only on the terminal asset price $S_T$. A call is valuable when $S_T>K$ and pays the excess $S_T-K$; a put is analogous, but with the roles of $S_T$ and $K$ reversed:
\begin{equation}
\label{eq:background-call-put}
f_{\rm call}(S_T)=(S_T-K)_+,
\qquad
f_{\rm put}(S_T)=(K-S_T)_+,
\end{equation}
where $(x)_+:=\max(x,0)$. These contracts are called vanilla because they have the standard single-asset terminal payoff structure. Options with more complicated payoffs, for example contracts depending on several assets, barriers, averages, spreads, or best/worst-of features, are usually called exotic options. Such contracts may still have European exercise style, but their payoff is no longer vanilla \cite{shreve2004stochastic, wilmott2006paul}.

After specifying the payoff, the next ingredient is a pricing model. In the PDE approach, the option value is written as a function of the current state and time, and the payoff is imposed as the terminal condition at maturity. In the one-asset Black--Scholes setting, the state variable is the asset price $S$, and the option price is denoted by $V(S,t)$. With constant risk-free rate $r$ and no dividends, the local-volatility Black--Scholes PDE \cite{BlackScholes1973} is
\begin{equation}
\label{eq:background-local-vol-bs}
\begin{aligned}
\partial_t V
+\frac{1}{2}\sigma(S)^2S^2\partial_{SS}V 
+rS\partial_SV-rV&=0,\\
V(S,T)&=f(S).
\end{aligned}
\end{equation}
Here $f(S)$ is the terminal payoff. The model is called local-volatility because the volatility is allowed to depend on the current asset price, $\sigma=\sigma(S)$. If $\sigma$ is independent of $S$, namely $\sigma(S)=\sigma_0$, Eq.~\eqref{eq:background-local-vol-bs} reduces to the classical constant-volatility Black--Scholes model \cite{Merton1973}. In that basic case, the terminal asset price is lognormal, and vanilla European call and put prices admit the well-known closed-form Black--Scholes formula. Hence these prices can be computed directly from an analytical expression, without solving the PDE numerically. For this reason, the constant-volatility model is mainly useful here as a simple and transparent validation benchmark.

A key limitation of the constant-volatility model is visible through implied volatility. If market prices were exactly described by a single constant volatility, the implied volatility would be flat across strikes. In practice, observed implied volatilities usually form a smile or skew, as illustrated schematically in Fig.~\ref{fig:volsmileskew}. This departure from flatness can be modeled in several ways. The first generation models were based on the concept of local volatility \cite{dupire1994pricing,derman2016volsmile}. Here the model keeps a one-factor
Markovian description but replaces the constant volatility by a deterministic
function of time and spot, $\sigma_{\mathrm{loc}}(t,S)$. Such models generally
lead to variable-coefficient Black--Scholes PDEs and are therefore natural
benchmarks for numerical PDE methods, but they do not usually admit closed-form
Black--Scholes formulas.

\begin{figure}[ht!]
    \centering
    \includegraphics[width=0.99\linewidth]{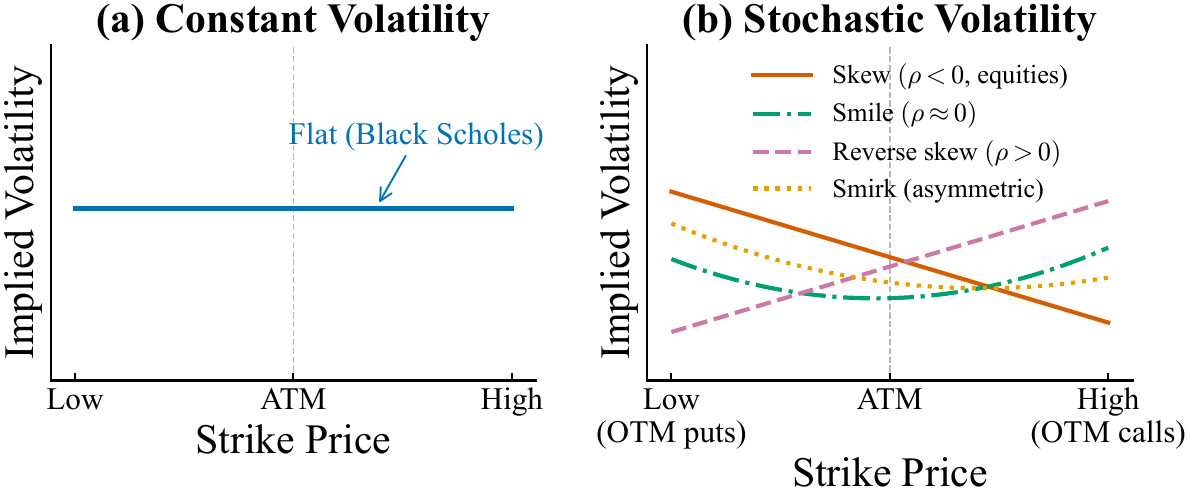}
    \caption{Constant-volatility Black--Scholes corresponds to a flat implied-volatility curve, whereas stochastic volatility models such as Heston can generate a smile or skew.}
\label{fig:volsmileskew}
\end{figure}
The Heston model \cite{heston1993closed} instead introduces a stochastic volatility framework, in which the instantaneous variance \(v\) is modeled as an additional state variable of the system. This is useful because market volatility fluctuates over time, tends to revert toward typical levels, and is often correlated with asset-price movements.
In the Heston model, the skew is structural: $\rho<0$ couples negative spot moves with positive variance shocks. This leverage effect makes implied volatility higher
at low strikes than at high strikes, producing the familiar equity skew. The option value is therefore written as $V(S,v,t)$, and the Heston pricing PDE is
\begin{equation}
\label{eq:background-heston-pde}
\begin{aligned}
&\partial_t V
+\frac{1}{2}vS^2\partial_{SS}V
+\rho\sigma_v vS\partial_{Sv}V+\frac{1}{2}\sigma_v^2v\partial_{vv}V \\
&\quad
+rS\partial_SV
+\kappa(\theta-v)\partial_vV
-rV=0,\\
&V(S,v,T)=f(S).
\end{aligned}
\end{equation}
Here $\kappa$ is the mean-reversion rate of the variance, $\theta$ is its long-run mean, $\sigma_v$ is the volatility of variance, and $\rho$ is the spot--variance correlation. Compared with the local-volatility Black--Scholes PDE, the Heston PDE has an additional variance direction and a mixed derivative term $\partial_{Sv}V$. This makes the numerical problem more demanding, but it also gives a stochastic-volatility description that is closer to observed market behavior.

To connect the pricing PDEs with a quantum algorithm, we first discretize the spatial variables using finite differences while keeping time continuous. In practice, the function $V(S,t)$ is replaced by its values on a finite grid. For a one-dimensional grid with $N$ points, this gives the option-price vector
\begin{equation}
\label{eq:background-price-vector}
\begin{aligned}
\vec V(t)
&=
\bigl(V_0(t),V_1(t),\ldots,V_{N-1}(t)\bigr)^{\top},\\
V_j(t)
&\approx V(S_j,t),
\qquad j=0,\ldots,N-1.
\end{aligned}
\end{equation}
After this spatial discretization, the pricing PDE becomes a finite-dimensional linear ODE system,
\begin{equation}
\label{eq:background-discrete-ode}
\begin{aligned}
\frac{d\vec V(t)}{dt}&=A\vec V(t)+\vec b,\\
\vec V(T)&=\vec f,
\end{aligned}
\end{equation}
where $\vec f$ is the sampled payoff at maturity, $A$ is the finite-difference pricing operator, and $\vec b$ contains possible boundary-condition contributions.

In the multi-asset setting, compatibility of the payoff with the finite-dimensional ODE form in Eq.~\eqref{eq:background-discrete-ode} is tied to the boundary treatment. The quantum evolution framework used here supports standard time-independent boundary conditions of Dirichlet, Neumann, Robin, and periodic type. However, for some multi-asset payoffs, the mathematically natural boundary data on faces   are not given by a simple local boundary condition. Instead, they require solving lower-dimensional auxiliary pricing PDEs on boundary faces, edges, and corners. Incorporating this hierarchy of boundary PDEs would require an additional iterative PDE-solving layer, which is outside the boundary model used in the present end-to-end analysis. Nevertheless, we stress that for many payoffs one can impose quantum-compatible approximate boundary closures, at the cost of an additional modeling or discretization error. For the forthcoming complexity results, we therefore focus on a payoff that is directly compatible with the construction without auxiliary boundary solves: the Worst-of Call,
\begin{equation}
\label{eq:background-worstof-call}
f_{\mathrm{wo}}(S)
=
\left(\min_{1\le i\le d}S_i-K\right)_+ .
\end{equation}
We discuss the boundary-compatibility issue in more detail in Appendix~\ref{sec:bs-multiasset-discretization}.

The quantum representation is obtained by choosing the grid size to be a power of two, $N=2^n$. An $n$-qubit quantum state has exactly $2^n$ computational-basis amplitudes, so the $N$ grid values can be encoded into the amplitudes of a quantum state:
\begin{equation}
\label{eq:background-price-state}
\ket{V(t)}
=
\frac{1}{\|\vec V(t)\|_2}
\sum_{j=0}^{2^n-1}V_j(t)\ket{j}.
\end{equation}
The normalization factor is necessary because quantum states must have unit norm. Thus the quantum state stores the shape of the discretized option-price vector, while the overall scale $\|\vec V(t)\|_2$ must be tracked or recovered separately when amplitudes are converted back into option prices.

For multi-asset pricing problems, the same amplitude-encoding idea is applied to each spatial coordinate. If there are $s$ spatial variables and $N=2^n$ grid points per variable, then the full grid has
\begin{equation}
\label{eq:background-grid-size}
N^s=2^{sn}
\end{equation}
points and can be encoded using $W=sn$ system qubits. The two main multidimensional models considered in this paper are the multi-asset Black--Scholes PDE and the multi-asset Heston PDE.

For a $d$-asset Black--Scholes model, the spatial variables are the asset prices
$S=(S_1,\ldots,S_d)$, so $W=dn$. In its correlated local-volatility form, the pricing PDE is
\cite{guillaume2019multidimensional}
\begin{equation}
\label{eq:background-multi-bs-pde}
\begin{aligned}
&\partial_t V
+\sum_{i=1}^{d} rS_i\partial_{S_i}V
+\frac{1}{2}\sum_{i=1}^{d}
\sigma_i(S_i)^2S_i^2\partial_{S_iS_i}V \\
&\quad
+\sum_{1\le i<j\le d}
\rho_{ij}\sigma_i(S_i)\sigma_j(S_j)
S_iS_j\partial_{S_iS_j}V
-rV=0,\\
&V(S,T)=f(S).
\end{aligned}
\end{equation}
Here $f(S)$ is the terminal payoff, and the correlation matrix
$(\rho_{ij})_{i,j=1}^d$ is positive semidefinite with $\rho_{ii}=1$.
The finite-difference reduction of Eq.~\eqref{eq:background-multi-bs-pde}, including the mixed-derivative stencil, boundary conditions, and homogeneous ODE embedding, is given in Appendix~\ref{sec:bs-multiasset-discretization}.

For a $d$-asset Heston model, each asset price is coupled to a variance factor, so the state variables are
$(S,v)=(S_1,\ldots,S_d,v_1,\ldots,v_d)$ and $W=2dn$. Using compact notation for the Brownian correlation blocks
\[
\rho^{SS}_{ij}:=\rho_{S_i,S_j},\qquad
\rho^{Sv}_{ij}:=\rho_{S_i,v_j},\qquad
\rho^{vv}_{ij}:=\rho_{v_i,v_j},
\]
the multi-asset Heston pricing PDE \cite{chan2024financial} has the form
\begin{equation}
\label{eq:background-multi-heston-pde}
\begin{aligned}
&\partial_t V
+\sum_{i=1}^{d} rS_i\partial_{S_i}V
+\sum_{i=1}^{d}\kappa_i(\theta_i-v_i)\partial_{v_i}V
-rV \\
&\quad
+\frac{1}{2}\sum_{i,j=1}^{d}
\rho^{SS}_{ij}S_iS_j\sqrt{v_iv_j}\,
\partial_{S_iS_j}V \\
&\quad
+\sum_{i,j=1}^{d}
\rho^{Sv}_{ij}\sigma_jS_i\sqrt{v_iv_j}\,
\partial_{S_iv_j}V \\
&\quad
+\frac{1}{2}\sum_{i,j=1}^{d}
\rho^{vv}_{ij}\sigma_i\sigma_j\sqrt{v_iv_j}\,
\partial_{v_iv_j}V=0,\\
&V(S,v,T)=f(S).
\end{aligned}
\end{equation}
The corresponding finite-difference system, boundary treatment, and homogeneous ODE form used by the quantum algorithm are described in Appendix~\ref{sec: ODE heston multi}.

\begin{table*}[t!]
\centering
\caption{\textbf{Dominant end-to-end gate-cost scalings for single-point output in the multi-asset Black--Scholes and Heston pipelines.}
Gate costs are reported in terms of CNOT gates and one-qubit Pauli-axis rotations. Here $N=2^n$ is the number of grid points per spatial axis and $d$ is the number of assets. The payoff is the Worst-of Call in Eq.~\eqref{eq:background-worstof-call}. State-preparation and evolution costs add, while the single-point readout overhead is multiplicative.}
\label{tab:scaling_bs_heston_multi}
\renewcommand{\arraystretch}{1.3}
\setlength{\tabcolsep}{3.5pt}
\begin{adjustbox}{width=\textwidth}
\begin{tabular}{l c c c c c c c}
\toprule
\textbf{Model} & & \textbf{State prep.} & & \textbf{Evolution} & & & \textbf{Readout} \\
\midrule
Local vol.\ Black--Scholes
&
$\displaystyle \Bigl($
&
$\displaystyle d^{3}n+d^{2}n\log n$
&
{\Large $+$}
&
$\displaystyle dT\,2^{2n}\!\left[dQ_{\sigma}n\log n+n_{\xi}\log n_{\xi}+d^{4} n\right]$
&
$\displaystyle \Bigr)$
&
{\Large $\times$}
&
$\displaystyle 2^{nd/2}$
\\[1.4ex]
Heston
&
$\displaystyle \Bigl($
&
$\displaystyle d^{3}n+d^{2}n\log n$
&
{\Large $+$}
&
$\displaystyle dT\,2^{2n}\!\left[dn\log n+n_{\xi}\log n_{\xi}+d^{4} n\right]$
&
$\displaystyle \Bigr)$
&
{\Large $\times$}
&
$\displaystyle 2^{nd}$
\\
\bottomrule
\end{tabular}
\end{adjustbox}
\end{table*}

\bigskip
\sectionnotoc{Results}

\label{sec:main-results}

Quantum algorithms used for numerical tasks commonly have three conceptual stages: (i) state preparation, (ii) quantum evolution  and (iii) readout. In the first stage, classical data are encoded into a quantum state. In the second stage, a quantum operation implements the desired computational dynamics. In the third stage, measurements convert the final quantum state back into classical information. For financial applications, this end-to-end classical interface is essential: contracts, strikes, rates, model parameters, quotes, and currency-denominated prices are all classical objects. Therefore, \textit{a practically meaningful quantum pricing algorithm should start from classical contract/model data and return classical option values.}
Our algorithm follows this standard structure: 
\begin{enumerate}
    \item The \textbf{state-preparation} stage prepares both the option payoff and the smooth Schr\"odingerisation cut-off state; the construction uses block-encodings of coordinate operators together with quantum singular value transformation (QSVT) \cite{gilyen2019quantum} implementing polynomial transformations \cite{guseynov2024explicit_Quantum_state_prep}.
    
    \item The \textbf{evolution stage} maps the non-unitary finite-difference pricing dynamics into a unitary Schr\"odinger system via Schr\"odingerisation \cite{PRL2024,PRA2023,analog,cao2023quantum,PRS2024,jin2024quantumsimulationfokkerplanckequation,JIN2025114138,jin2025quantumpreconditioningmethodlinear,doi:10.1137/23M1563451,JinLiuMa2024MaxwellSchrodingerisation,jin2025schrodingerizationbasedquantumalgorithms,schrodingerisation_optimal_queries}, and then applies block-encoding framework \cite{guseynov2025quantumframeworksimulatinglinear,guseynov2024efficientPDE} and QSVT to implement the corresponding Hamiltonian evolution.
    
    \item The \textbf{readout stage} uses a Hadamard-test construction together with amplitude estimation to recover a selected option value as a classical number \cite{doi:10.1137/1.9781611977561.ch12}.
\end{enumerate}

We now state the two main complexity results for the quantum pricing framework. The detailed costs from Theorems~\ref{thm:main-bs-multi} and~\ref{thm:main-heston-multi} are summarized in Table~\ref{tab:scaling_bs_heston_multi}.

\begin{theorem}[Local vol. Black--Scholes]
\label{thm:main-bs-multi}
Consider the $d$-asset local-volatility Black--Scholes pricing PDE in Eq.~\eqref{eq:background-multi-bs-pde}. Let each asset direction be discretized with $N=2^n$ grid points using second-order finite differences, and take the payoff to be the Worst-of Call payoff in Eq.~\eqref{eq:background-worstof-call}, or another payoff whose boundary treatment is compatible with the finite-dimensional ODE representation in Eq.~\eqref{eq:background-discrete-ode}. Suppose the local volatility functions are approximated by polynomials of degree at most $Q_\sigma$. Then the end-to-end gate cost for estimating one value $V(0,\vec S_0)$ is
\begin{equation}
\label{eq:main-bs-cost}
\begin{aligned}
\mathcal C_{\mathrm q}^{\mathrm{BS}}
=\mathcal{O}&
\Bigl(
d^{3}n+d^{2}n\log n
+dT\,2^{2n}
\bigl[
dQ_{\sigma}n\log n
\\&\qquad+n_{\xi}\log n_{\xi}
+d^{4}n
\bigr]
\Bigr)2^{nd/2},
\end{aligned}
\end{equation}
up to factors polylogarithmic in the remaining precision parameters. The resulting leading grid-size dependence is
\begin{equation}
\label{eq:main-bs-leading}
\widetilde{\mathcal C_{\mathrm q}^{\mathrm{BS}}}=
\widetilde{\mathcal O}\!\left(d^22^{n(2+d/2)}\right)  =
\widetilde{\mathcal O}\!\left(d^2N^{2+d/2}\right),
\end{equation}
where we omitted $poly(n)$ factors for simplicity. Here $n_\xi$ is the number of qubits in the auxiliary Schr\"odingerisation register. For additive Schr\"odingerisation error $\epsilon_{\mathrm{Schr}}$ following the result \cite{schrodingerisation_optimal_queries}, we take
\begin{equation}
\label{eq:main-bs-nxi-scaling}
n_\xi=\mathcal O\!\left(
\log\log\frac{1}{\epsilon_{\mathrm{Schr}}}
\right).
\end{equation}

\end{theorem}

\begin{proof}
Appendix~\ref{sec:bs-multiasset-discretization} reduces the multi-asset Black--Scholes PDE to the homogeneous ODE system used by the quantum algorithm. The payoff and cut-off preparation costs are derived in Appendix~\ref{section:state prep}. The Schr\"odingerisation and Hamiltonian-simulation costs follow from Appendix~\ref{section:evoltuion non conservative}, in particular Proposition~\ref{theorem:QuantumHamiltonianSimulation}; the factor $2^{2n}$ is induced by the second-order finite-difference generator. Finally, Appendix~\ref{subsec:readout_complexity} shows that single-point readout on an $N^d$-point grid contributes the factor $2^{nd/2}$. Adding the state-preparation and evolution costs and multiplying by the readout overhead yields Eq.~\eqref{eq:main-bs-cost}.
\end{proof}

\begin{theorem}[Heston]
\label{thm:main-heston-multi}
Consider the $d$-asset Heston pricing PDE in Eq.~\eqref{eq:background-multi-heston-pde}. The state variables are the $d$ asset coordinates and the $d$ variance coordinates, hence the spatial dimension is $2d$. Let each spatial direction, including both asset and variance directions, be discretized with $N=2^n$ grid points using second-order finite differences, and take the payoff to be the Worst-of Call payoff in Eq.~\eqref{eq:background-worstof-call}, or another payoff whose boundary treatment is compatible with the finite-dimensional ODE representation in Eq.~\eqref{eq:background-discrete-ode}. Suppose the payoff and smooth cut-off states are prepared using Appendix~\ref{section:state prep}, with success probabilities boosted to $\mathcal O(1)$. Then the end-to-end gate cost for estimating one value $V(0,\vec S_0,\vec v_0)$ is
\begin{equation}
\label{eq:main-heston-cost}
\begin{aligned}
\mathcal C_{\mathrm q}^{\mathrm{Heston}}
=\mathcal{O}
\Bigl(
d^{3}n&+d^{2}n\log n+dT\,2^{2n}
\bigl[
dn\log n
\\&+n_{\xi}\log n_{\xi}+d^{4}n
\bigr]
\Bigr)2^{nd},
\end{aligned}
\end{equation}
up to factors polylogarithmic in the precision parameters. The resulting leading grid-size dependence is
\begin{equation}
\label{eq:main-heston-leading}
\begin{aligned}
\widetilde{\mathcal C_{\mathrm q}^{\mathrm{Heston}}}=
\widetilde{\mathcal O}\!\left(d^22^{n(d+2)}\right) =
\widetilde{\mathcal O}\!\left(d^2N^{d+2}\right),
\end{aligned}
\end{equation}
where we omitted $poly(n)$ factors for simplicity. Here $n_\xi$ is again the number of qubits in the auxiliary Schr\"odingerisation register. Following the result \cite{schrodingerisation_optimal_queries}, we take
\begin{equation}
\label{eq:main-heston-nxi-scaling}
n_\xi=\mathcal O\!\left(
\log\log\frac{1}{\epsilon_{\mathrm{Schr}}}
\right).
\end{equation}
\end{theorem}

\begin{proof}
Appendix~\ref{sec: ODE heston multi} details the finite-difference reduction of the multi-asset Heston PDE to the homogeneous ODE system. Appendix~\ref{section:state prep} lays out the payoff and cut-off state-preparation costs. Appendix~\ref{section:evoltuion non conservative}, together with Proposition~\ref{theorem:QuantumHamiltonianSimulation}, gives the Schr\"odingerisation and Hamiltonian-simulation cost. Since the Heston grid has $2d$ spatial directions, Appendix~\ref{subsec:readout_complexity} produces the single-point readout factor $2^{nd}$ for an $N^{2d}$-point grid. Combining these contributions gives Eq.~\eqref{eq:main-heston-cost}.
\end{proof}

\begin{figure*}[ht!]
    \centering
    \includegraphics[width=1\linewidth]{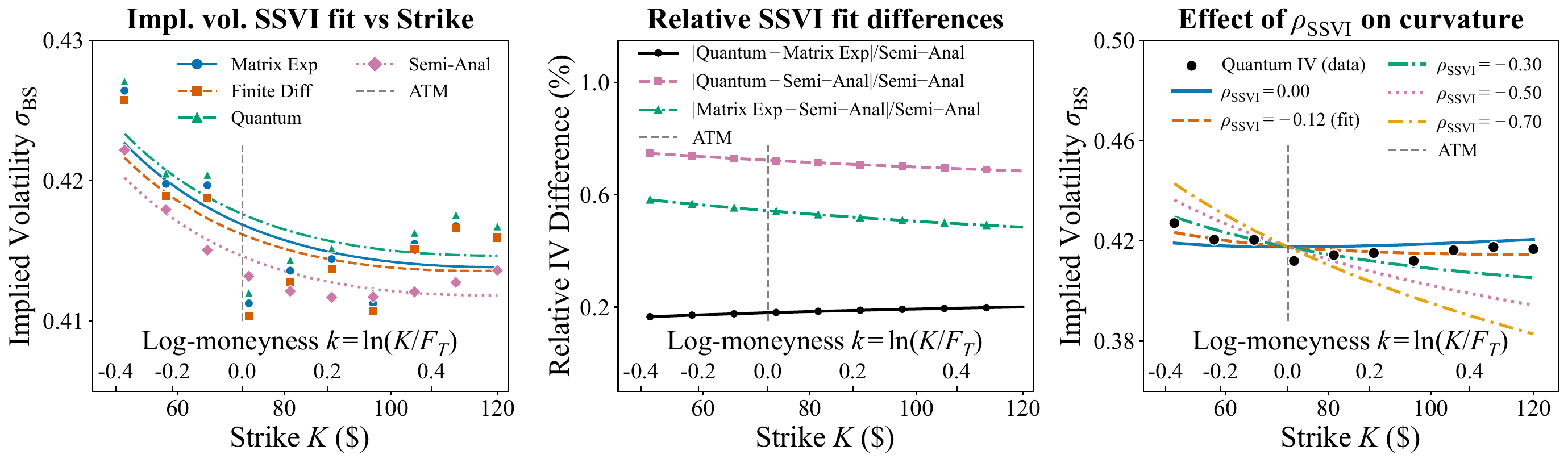}
    \caption{
    \textbf{Fixed-maturity Heston implied volatility smile shown in strike and forward log-moneyness coordinates ($k=\ln(K/F_T)$).} 
    \textbf{Left:} SSVI fits from four solvers. 
    \textbf{Centre:} Relative discrepancies against the semi-analytical benchmark. 
    \textbf{Right:} Sensitivity to the SSVI correlation parameter $\rho_{\mathrm{SSVI}}$. The best fitted $\rho_{\mathrm{SSVI}}$ produces a mild skew that preserves the arbitrage-free nature of the curve (Algorithm~\ref{alg:ssvi}). Simulation details are clarified in Appendix~\ref{subsec: Heston numerics}.}
    \label{fig:ssvi-three-panels}
\end{figure*}

The gate counts above are expressed in terms of CNOT gates and one-qubit Pauli-axis rotations. This is the natural logical-level description for our explicit circuits, since such gates are standard elementary operations in quantum-circuit models and are directly meaningful before specifying a particular error correction code. If one further assumes an error-corrected architecture, it is often useful to translate these counts into Clifford+$T$ resources, because Clifford operations are typically fault-tolerant under stabilizer-code constructions while non-Clifford $T$ gates carry the dominant synthesis overhead. The following Lemma records the corresponding compilation overhead for the rotation gates appearing in our estimates.

\begin{lemma}[Compilation to Clifford+$T$~\cite{10.5555/3179330.3179331}]
\label{thm:clifford-t-compilation}
Let $\epsilon>0$. Let $U$ be a quantum circuit composed of CNOT gates and
$\mathcal C$ one-qubit Pauli-axis rotations $R_{\alpha}(\theta)$, with
$\alpha\in\{x,y,z\}$. Then there exists a Clifford+$T$ circuit
$\widetilde U$ such that
\begin{equation}
\label{eq:clifford-t-approx-error}
\|U-\widetilde U\|_\mathrm{op}\leq \epsilon ,
\end{equation}
where $\|\cdot\|_\mathrm{op}$ denotes the operator norm.
The CNOT gates are already Clifford gates, and the $\mathcal C$ one-qubit
rotations can be replaced by Clifford+$T$ approximations with total $T$-count
\begin{equation}
\label{eq:clifford-t-total-count}
T_{\mathrm{count}}\leq 4\mathcal C\log_2\!\left(\frac{\mathcal C}{\epsilon}\right)+\mathcal O\!\left(\mathcal C\log\log\frac{\mathcal C}{\epsilon}\right).
\end{equation}
\end{lemma}

\subsectionnotoc{Numerical validation}

We now validate the quantum pricing framework numerically. The main test is the Heston implied-volatility smile: in contrast to the flat curve of constant-volatility Black--Scholes, stochastic-volatility models generate strike-dependent smiles and skews  \cite{derman2016volsmile}, as illustrated in Fig.~\ref{fig:volsmileskew}. Recovering this structure is a crucial market-level reality and the quantum framework preserves the financial behavior of the Heston model.

Fig.~\ref{fig:ssvi-three-panels} shows the resulting fixed-maturity Heston smile. The left panel compares the smiles obtained from four solvers, the centre panel reports relative discrepancies against the semi-analytical benchmark, and the right panel shows the dependence of the fitted smile on the best-fit SSVI skew parameter from Algorithm~\ref{alg:ssvi}. Across the displayed strike range, the quantum framework agrees with the classical and semi-analytical references and reproduces the non-flat smile/skew structure associated with stochastic volatility models.

The four solvers compared in Fig.~\ref{fig:ssvi-three-panels} are the quantum framework developed in this paper, the semi-analytical Heston benchmark based on the characteristic-function representation \cite{heston1993closed}, and two classical numerical baselines: a matrix-exponential baseline and an implicit finite-difference baseline. The matrix-exponential method \cite{Higham2008FunctionsOfMatrices} applies the exponential action of the semi-discrete ODE system in Eq.~\eqref{eq:background-discrete-ode}. The finite-difference method is the standard grid-based PDE solver (implicit Euler method) that evolves the discretized pricing equation backward in time \cite{duffy2006fdm}. Details of the classical baselines and numerical implementation are given in Appendix~\ref{appendix:numerical methods}.

We also test the same workflow on the one-dimensional Black--Scholes model in Appendix~\ref{sec:bs1d-numerical-test}, where the quantum framework is compared with the analytical Black--Scholes formula, finite differences, and matrix exponentiation. This gives a closed-form benchmark for the pipeline before the stochastic-volatility smile experiment.

Together, these numerical validations complement the complexity results by showing that the framework produces the expected financial outputs. It matches the analytical Black--Scholes benchmark in the closed-form setting and recovers the Heston implied-volatility smile in the stochastic-volatility setting, supporting the end-to-end consistency of the proposed quantum pricing workflow. To the best of our knowledge, this is the first demonstration of a quantum PDE-based pricing workflow that recovers an implied-volatility smile/skew from quantum-computed option prices.

\bigskip
\subsectionnotoc{Quantum advantage}

For the grid-based PDE baselines described in Appendix~\ref{appendix:numerical methods}, the leading classical arithmetic cost for a pricing PDE with $s$ spatial variables and $N=2^n$ grid points per variable is
\begin{equation}
\label{eq:main-classical-W}
\mathcal C_{\mathrm{cl}}(s,N)
=
\widetilde{\mathcal O}\!\left(s^2N^{s+2}\right).
\end{equation}
Here $N^s$ is the number of spatial grid degrees of freedom. The additional factor $N^2$ comes from the second-order finite-difference discretization: for finite-difference time stepping it appears through the number of time steps required at grid resolution $\Delta x=\Theta(1/N)$, while for matrix-exponential-action methods it appears through the generator norm $\|A\|_{\max}=\Theta(\Delta x^{-2})=\Theta(N^2)$.

\begin{remark}[Disclaimer on end-to-end costs]
\label{rem:universal-disclaimer-ft-topology}
The resource estimates in this paper are \emph{logical-level} (algorithmic) scalings and should be
interpreted with care. First, quantum gate counts and circuit depth are not directly comparable to classical
FLOP-based operation counts, since one logical gate is not a single physical step and may be slower and more
costly to implement; thus the quantum--classical comparisons made here are at the level of asymptotic scaling
rather than like-for-like wall-clock cost. Second, we do not include (i) the overhead of quantum error-correction codes, including the number of physical qubits required per logical qubit and the additional depth required for fault-tolerant logical operations, nor (ii) hardware constraints such as limited qubit connectivity and chip topology, which can increase depth through routing (SWAP) and scheduling overhead.
\end{remark}

For the $d$-asset Black--Scholes PDE, the number of spatial variables is $s=d$, so $W=dn$. Therefore, Eq.~\eqref{eq:main-classical-W} gives the leading scaling
\begin{equation}
\label{eq:main-bs-classical}
\mathcal C_{\mathrm{cl}}^{\mathrm{BS}}=
\mathcal O\!\left(d^2 2^{n(d+2)}\right)=
\mathcal O\!\left(d^2 N^{d+2}\right),
\end{equation}
as derived in Appendix~\ref{subsec:multiasset BS scaling}. Combining this with Theorem~\ref{thm:main-bs-multi} gives
\begin{equation}
\label{eq:main-bs-gain}
\frac{
\mathcal C_{\mathrm{cl}}^{\mathrm{BS}}
}{\widetilde{
\mathcal C_{\mathrm q}^{\mathrm{BS}}}
}
=
\widetilde{\Omega}\!\left(2^{nd/2}\right)=
\widetilde{\Omega}\!\left(N^{d/2}\right),
\end{equation}
where we omitted $poly(n)$ factors for simplicity. Thus, for multi-asset Black--Scholes, the quantum pipeline gives an improvement factor $N^{d/2}$ for single-point option-price recovery.

For the $d$-asset Heston PDE, the state contains both asset and variance variables, so $s=2d$ and $W=2dn$. Therefore,
\begin{equation}
\label{eq:main-heston-classical}
\mathcal C_{\mathrm{cl}}^{\mathrm{Heston}}=
\mathcal O\!\left(d^2 2^{n(2d+2)}\right) =
\mathcal O\!\left(d^2 N^{2d+2}\right),
\end{equation}
as derived in Appendix~\ref{subsec: Heston multi}. Combining this with Theorem~\ref{thm:main-heston-multi} gives
\begin{equation}
\label{eq:main-heston-gain}
\frac{
\mathcal C_{\mathrm{cl}}^{\mathrm{Heston}}
}{\widetilde{
\mathcal C_{\mathrm q}^{\mathrm{Heston}}}
}
=
\widetilde{\Omega}\!\left(2^{nd}\right) =
\widetilde{\Omega}\!\left(N^{d}\right).
\end{equation}
Thus, for multi-asset Heston, the quantum pipeline gives an improvement factor $N^d$ for single-point option-price recovery.

Consequently, within the grid-based PDE setting considered here, the quantum framework gives a strict asymptotic improvement for every $d\geq 1$ and every nontrivial grid resolution $N>1$. For Black--Scholes, the improvement factor is $N^{d/2}$; for Heston, the improvement factor is $N^d$. These factors increase exponentially with the number of assets, so the advantage becomes more pronounced in higher-dimensional pricing problems.

This scaling should be viewed in the context of the curse of dimensionality. In a full-grid PDE method, adding spatial variables increases the number of grid degrees of freedom exponentially: $N^d$ for $d$-asset Black--Scholes and $N^{2d}$ for $d$-asset Heston. The proposed quantum framework reduces this grid-volume dependence to its square root for single-point output, namely from $2^W$ to $2^{W/2}$. Thus the framework provides a dimension-enhanced quantum advantage for the end-to-end single-point pricing task, while it does not eliminate the curse of dimensionality: the resulting costs still grow exponentially with the number of assets.

\bigskip\medskip
\sectionnotoc{Related work}\label{section:related-literature}

\textbf{PDE-centric quantum methods.} The state-of-the-art quantum PDE-centric methods for option pricing are methodologically close to our work: they use the pricing PDE as the computational object and implement the induced dynamics through Hamiltonian-simulation-type primitives or variational solvers. Compared with those methods, our contribution is a fully explicit end-to-end pipeline, with performance guarantees from classical input data to classical option-price output.

A closely related work \cite{alghassi2022variational} proposes a variational quantum algorithm based on
imaginary-time evolution to solve the Feynman--Kac PDE induced by a multidimensional SDE system. Notably for quantitative-finance applications, the authors discuss the non-unitary normalization issue
arising when the PDE evolution preserves probability mass (an $\ell_1$-type structure) rather than an
$\ell_2$ norm, and introduce a proxy normalization scheme to keep the solution approximately normalized. At the same time, the approach remains variational, so its practical performance depends on the choice, initialization, and expressivity of the ansatz.

A strong and closely related PDE-based quantum contribution is the finite-difference formulation of multi-asset Black--Scholes pricing in \cite{KuboFDMmultiasset2022}. The work is valuable because it treats a genuinely multidimensional pricing PDE and, importantly, makes explicit a central obstacle for quantum PDE methods in finance: the desired option value may be encoded in a single amplitude, so readout might potentially remove the apparent speedup with respect to the PDE dimension. Another strong aspect of the paper is that it proposes a concrete way to mitigate this issue, by reformulating the price extraction as an expectation at an intermediate time rather than as direct single-amplitude readout. However, this readout strategy is closely tied to the closed-form multivariate lognormal transition law of the Black--Scholes model. As a result, for local-volatility or stochastic-local-volatility models, where such a transition law is not available, the same mechanism does not apply directly and would require a different readout construction.

A further useful perspective is provided in \cite{guseynov2025quantumalgorithmlocalvolatilityoption}, which directly targets the readout bottleneck in quantum PDE-based option pricing. The paper argues that solving the forward Kolmogorov equation can turn price recovery into an expected-payoff estimation problem, rather than a single-amplitude readout problem, and proposes a swap-test-based strategy for extracting the option value. This gives a constructive route around one of the main obstacles in quantum PDE pricing, although the readout analysis is payoff dependent.

Another relevant direction is the variational quantum-simulation framework for stochastic differential equations in \cite{kubo2021variational}. The paper presents a flexible NISQ-oriented approach applicable to Black--Scholes-type diffusions and, in principle, to stochastic-volatility models such as Heston, by evolving a parametrized ansatz state $\ket{\psi(\theta)}$ according to McLachlan's variational principle. This provides a useful route for SDE-based option pricing. At the same time, the method remains variational, so its accuracy depends on the ansatz and conditioning of the parameter dynamics. Moreover, the finite-difference diffusion generator contains $1/\Delta x^2$ coefficients, so for $N=2^n$ grid points its norm scales as $\mathcal{O}(2^{2n})$, which may lead to many variational time-marching steps.

\textbf{SDE-centric quantum Monte Carlo methods.}
Most SDE-centric quantum approaches to derivative pricing aim to accelerate Monte Carlo estimators by combining coherent path simulation with quantum amplitude estimation (QAE). Under suitable path-loading and payoff-evaluation assumptions, this replaces the classical $\mathcal{O}(\varepsilon^{-2})$ sampling dependence for additive error $\varepsilon$ by a quantum $\widetilde{\mathcal{O}}(\varepsilon^{-1})$ scaling, where $\varepsilon$ is absolute error in option price. For comparison, consider a grid-based discretization of a PDE whose highest spatial derivatives are of second order. Let $N=2^n$ denote the number of spatial grid points, in one dimension or per spatial dimension in the multidimensional case. If the spatial discretization is performed using a finite-difference scheme with second-order accuracy, then the spatial discretization error scales as $\mathcal{O}(N^{-2})$ \cite{LeVeque2007FDM}. Therefore, when the target accuracy is informally matched to the spatial grid error, one may identify
\begin{equation}
    \varepsilon \sim N^{-2}=2^{-2n}.
\end{equation}
Under this parametrization, the QAE sampling factor, which scales as $\widetilde{\mathcal{O}}(\varepsilon^{-1})$, becomes $\widetilde{\mathcal{O}}(2^{2n})$.

The quantum multilevel Monte Carlo framework in \cite{An2021quantumaccelerated} gives rigorous quadratic-speedup guarantees for SDE-based pricing, including Black--Scholes, local-volatility models, and Greeks. These results provide an important theoretical foundation, although they heavily rely on quantum oracle access to the SDE increments and payoff, whose concrete implementation may be nontrivial for a given financial model. The approach in \cite{scriba2025montecarlooptionpricingquantum} moves toward a more explicit path-simulation construction and discusses stochastic-volatility-type settings, thereby reducing the reliance on abstract sampling oracles. Circuit-level constructions for local-volatility pricing are further developed in \cite{Kaneko2022localvol}; in particular, the use of pseudo-random numbers avoids storing one random-number register per time step, with the tradeoff that random numbers are generated sequentially and the circuit depth grows with their number. For stochastic volatility, \cite{Wang2024optionpricingunder} provides a detailed end-to-end resource estimate for single-asset Heston Asian and barrier options using Euler-type SDE discretizations with QAE. This gives a concrete and practically informative benchmark. Remarkably, the paper also demonstrates an important practical limitation of pathwise quantum Monte Carlo: the substantial memory overhead. For instance, the weak-Euler construction with $256$ time steps already requires approximately $2.2\times 10^4$ logical qubits.

In the multi-asset setting, \cite{chen2025quantummontecarloalgorithm} gives a rigorous quantum Monte Carlo algorithm for high-dimensional correlated Black--Scholes pricing with continuous piecewise-affine payoffs, overcoming the full-grid curse of dimensionality in that model. The construction exploits the multivariate lognormal transition law that is available in closed form. Extending the same strategy to local-volatility or stochastic-volatility models would require a different path-loading mechanism. A particularly complete analysis in the SDE setting is given in \cite{herman2026quantumspeedupsderivativepricing}, where input loading, path evolution, coherent arithmetic, payoff evaluation, and readout are all included in the resource analysis. The paper proves end-to-end quadratic improvements in the precision dependence for a restricted Heston model, where the model excludes variance-variance cross-correlations and variance processes correlated with multiple equity processes. The pathwise implementation also uses registers for discretized random inputs and monitoring dates yielding high demand for logical qubits.

A remarkable \textbf{classical} result based on multilevel Picard methods overcomes the curse of dimensionality for certain semilinear Black--Scholes PDEs with default risk \cite{JentzenMultipicard}. This relies on the known explicit form of the underlying Black--Scholes/GBM SDE solution, which provides the required sampling representation. For local-volatility or stochastic-volatility models, such a closed-form path representation is generally unavailable, so the same mechanism does not directly transfer to the setting considered here.

\medskip\bigskip

\sectionnotoc{Discussion}

We have developed an end-to-end quantum PDE framework for option pricing under multi-asset Black--Scholes and Heston-type models. The framework starts from classical model and contract data, prepares the payoff and auxiliary Schr\"odingerisation states, implements the non-unitary pricing dynamics through a unitary embedding, and recovers selected option values as classical outputs. The resulting pipeline comes with explicit performance guarantees for the main algorithmic components: state preparation, Hamiltonian evolution, normalization recovery, and amplitude-estimation-based readout.

A central caveat is the readout step. In the backward-PDE formulation used here, the desired option price at a specified market state is encoded in one selected amplitude of the final quantum state. Recovering this single amplitude requires a readout overhead proportional to the square root of the grid volume. This is the main reason why the end-to-end complexity retains the factors $2^{nd/2}$ for Black--Scholes and $2^{nd}$ for Heston. A promising direction is to combine the present  pricing engine with forward Kolmogorov-type formulations, as in \cite{guseynov2025quantumalgorithmlocalvolatilityoption}. In such approaches, the option value can be expressed through an expected payoff against an evolved density, which may lead to more favorable readout mechanisms than direct single-amplitude recovery.

A second caveat concerns payoff coverage. The present analysis focuses on European-style payoffs that can be prepared efficiently and whose PDE boundary treatment fits the assumptions of the quantum evolution framework. However, many practically important contracts require additional structure: early-exercise products involve iterative or variational PDE solves, while path-dependent contracts require augmenting the state variables to encode accumulated quantities such as running averages, barriers, or realized variance. These cases are not covered by the current complexity theorems.

At the same time, path dependence is compatible with the general PDE-based framework: it can be incorporated by enlarging the state space and applying the same discretization, state-preparation, evolution, and readout pipeline to the augmented PDE system. Extending the present guarantees to such contracts, and in particular to path-dependent and early-exercise derivatives, is a natural target for future work.

\bigskip\medskip
\starsectionnotoc{Acknowledgments}
The authors thank Xiajie Huang for fruitful discussions. NG acknowledges funding from NSFC grant W2442002. NL acknowledges the support of the NSFC grants No. 12341104 and No. 12471411, the Shanghai Pilot Program for Basic Research, the Science and Technology Commission of Shanghai Municipality (STCSM) grant no. 24LZ1401200
(21JC1402900), the Shanghai Jiao Tong University 2030 Initiative, and the Fundamental Research Funds for the
Central Universities. 

\bigskip\medskip
\starsectionnotoc{Data Availability Statement}\label{sec:data availability}

The code used to generate the numerical results and figures in this work, including the implementations for the Black--Scholes and Heston simulations, the semi-analytical validation routines, and the implied-volatility / SSVI post-processing workflows is available in a GitLab repository \cite{guseynov_gitlab_placeholder}. No external experimental datasets were used in this study; all numerical data were generated from the models and parameters specified in the main text and Supplementary Material.

\bigskip\medskip
\starsectionnotoc{Declarations}
The authors have no competing interests to declare.

\clearpage

\let\oldaddcontentsline\addcontentsline
\renewcommand{\addcontentsline}[3]{}
\bibliography{references}
\let\addcontentsline\oldaddcontentsline

\clearpage
\appendix
\makeatletter
\renewcommand{\p@subsection}{\thesection.}
\renewcommand{\p@subsubsection}{\thesection.\arabic{subsection}.}
\makeatother
\onecolumngrid

\begin{center}
    \LARGE\textbf{Appendix}
\end{center}

\twocolumngrid
\tableofcontents
\onecolumngrid
\section{Black--Scholes}\label{section:Black-Scholes}

The formulation of the Black--Scholes model in the early 1970s represents a seminal development in modern mathematical finance. Building on It\^{o}'s stochastic calculus \cite{ito1944109}, Black and Scholes \cite{BlackScholes1973} and Merton \cite{Merton1973} introduced a replicating-portfolio argument that links the price of a contingent claim to a self-financing trading strategy written on the underlying asset. This approach yields a unique arbitrage-free valuation rule under a risk-neutral measure, resulting in a linear parabolic partial differential equation (PDE) for the option price. The model became the standard reference due to its internal consistency, clear probabilistic structure, and the availability of closed-form solutions for European options. Following these developments, the derivatives industry expanded rapidly in both sophistication and the breadth of traded products.

Despite well-documented empirical limitations, such as the assumption of constant volatility, the Black--Scholes framework remains the industry benchmark. As noted by \cite{vaidya2020learning}, market practitioners have effectively adopted it as a quotation convention rather than a strictly predictive model. Traders observe the market value of a derivative in dollar terms and invert the Black--Scholes formula to derive the ``implied volatility." This practice allows for the standardization of prices across different assets and maturities, cementing the model's status as the backbone of the derivatives market.

\subsection{Black--Scholes PDE}\label{sec:discretization}
Let $V(S,t)$ denote the value of a European-style derivative written on an underlying with spot price $S$ at time $t\in[0,T]$.
Under the risk-neutral measure, assuming a constant risk-free rate $r$ and zero continuous dividend yield ($q=0$), the Black--Scholes partial differential equation (PDE) reads
\begin{equation}
\label{eq:BS-PDE}
\pdv{V}{t}
+ \frac{1}{2}\,\sigma(S)^2 S^2 \pdv[2]{V}{S}
+ r S \pdv{V}{S}
- r V = 0,
\qquad S>0,\; 0\le t<T.
\end{equation}

In what follows we investigate a local-volatility specification in which the instantaneous volatility depends on the spot, $\sigma=\sigma(S)$, and is assumed piecewise continuous in $S$ on $(0,\infty)$. The notation is summarized in Table~\ref{table:notation-bs}.

\begin{table}[h!]
\centering
\caption{Notation used in the Black--Scholes formulation.}
\label{table:notation-bs}
\begin{tabular}{@{}ll@{}}
\toprule
\textbf{Symbol} & \textbf{Meaning} \\
\midrule
$S$ & Underlying asset price (spot) \\
$t$ & Calendar time, with maturity at $t=T$ \\
$T$ & Option maturity date \\
$K$ & Strike price \\
$V(S,t)$ & Arbitrage-free option value (generic) \\
$V^{\mathrm{call}}(S,t)$, $V^{\mathrm{put}}(S,t)$ & Call/put prices solving \eqref{eq:BS-PDE} with \eqref{eq:call/put-terminal} \\
$r$ & Constant risk-free interest rate \\
$\sigma(S)$ & Local volatility of the underlying \\
\bottomrule
\end{tabular}
\end{table}

For a European option, the pricing PDE is posed backward in time on $[0,T)$ and requires a terminal (time) boundary at maturity $t=T$.
This boundary corresponds to the payoff at maturity, which determines the value of $V$ for every state $(S,T)$.
Specifically, for strike $K$,
\begin{equation}
\label{eq:call/put-terminal}
V^{\mathrm{call}}(S,T)=\max(0,S-K), 
\qquad 
V^{\mathrm{put}}(S,T)=\max(0,K-S).
\end{equation}
Eq.~\eqref{eq:call/put-terminal} encodes the cash flows of the contract at maturity and is independent of the volatility model or interest-rate dynamics assumed between $t$ and $T$.
The terminal condition provides the time boundary needed for well-posedness of the problem together with the spatial conditions in $S$.
In classical finite-difference schemes, it is imposed directly at the last time layer, initializing the backward time-march from $t=T$ to $t=0$.

To model the Black--Scholes PDE on a finite asset domain $S\in[0,S_{\max}]$, we impose boundary
conditions that mirror the no-arbitrage limits of the option value. At the lower boundary the underlying is worthless,
so we set Dirichlet data
\begin{equation}
\label{eq:bc-S0-Dirichlet}
\lim_{S\to 0^+} V^{\mathrm{call}}(S,t)=0,
\qquad
\lim_{S\to 0^+} V^{\mathrm{put}}(S,t)=K\,e^{-r(T-t)}.
\end{equation}
As $S\to+\infty$ the call replicates one share minus a discounted bond while the put vanishes:
\begin{equation}
\label{eq:bc-infty-asymptotic}
\lim_{S\to+\infty} V^{\mathrm{call}}(S,t)=S-K\,e^{-r(T-t)},
\qquad
\lim_{S\to+\infty} V^{\mathrm{put}}(S,t)=0.
\end{equation}
On a truncated grid this is enforced at $S=S_{\max}$ either in Dirichlet form
\begin{equation}
\label{eq:bc-Smax-Dirichlet}
V^{\mathrm{call}}(S_{\max},t)=S_{\max}-K\,e^{-r(T-t)},
\qquad
V^{\mathrm{put}}(S_{\max},t)=0,
\end{equation}
or in Neumann (slope) form \cite{jeong2015accuracy,LyuParknonuniformPDE}, using the asymptotic deltas,
\begin{equation}
\label{eq:bc-Smax-Neumann}
\partial_S V^{\mathrm{call}}(S_{\max},t)=1,
\qquad
\partial_S V^{\mathrm{put}}(S_{\max},t)=0.
\end{equation}
Both choices are consistent with the continuous limits \eqref{eq:bc-infty-asymptotic}. In practice one chooses $S_{\max}$ well
above moneyness (e.g., $S_{\max}\in[4K,10K]$ depending on $T$ and $\sigma(S)$).

To solve the Black--Scholes PDE numerically, we first discretize the spatial (stock price) variable $S$ while keeping time $t$ continuous (method of lines). Let $S_k=\dfrac{S_{\max}}{N_s-1}\,k,\; k=0,1,\dots,N_s-1$ be a uniform grid and define the grid spacing
\begin{equation}
\label{eq:DeltaS-def}
\Delta S \;=\; S_{k+1}-S_k \;=\; \frac{S_{\max}}{N_s-1}.
\end{equation}
Form the state vector as an ordered list
\begin{equation}
\label{eq:state-vector-curly}
\vec V(t)\;=\;\{\,V(S_0,t),\,V(S_1,t),\,\dots,\,V(S_{N_s-1},t)\,\}^{\top},
\end{equation}
where $^{\top}$ denotes transposition (so that \eqref{eq:state-vector-curly} is a column vector). Write $V_k(t)\equiv V(S_k,t)$ for brevity. From now on we consider call options only; put prices can be computed similarly.

Next, we approximate spatial derivatives on the interior nodes $k=1,\dots,N_s-2$ with central, second-order accurate formulas:
\begin{align}
\label{eq:D1-central}
(\partial_S V)_k &\approx \frac{V_{k+1}-V_{k-1}}{2\,\Delta S}+\mathcal{O}(\Delta S^2),\\
\label{eq:D2-central}
(\partial_{SS}^2 V)_k &\approx \frac{V_{k-1}-2V_k+V_{k+1}}{(\Delta S)^2}+\mathcal{O}(\Delta S^2).
\end{align}
We pick central differences because the Black--Scholes operator is diffusion-dominated and symmetric in space, so central stencils yield low bias with second-order accuracy; other popular choices are summarized in Table~\ref{table:famous-schemes}. We also stick to second-order accuracy for simplicity, since our paper focuses on quantum methods rather than high-order spatial discretizations.

\renewcommand{\arraystretch}{1.8}
\setlength{\tabcolsep}{8pt}
\newcolumntype{Y}{>{\raggedright\arraybackslash}X}

\begin{table}[t!]
\centering
\begin{tabularx}{\textwidth}{|c|c|c|Y|}
\hline
\textbf{Type} & \textbf{Der. order} & \textbf{Acc. order} & \multicolumn{1}{c|}{\textbf{Formula at } $x_i$} \\
\hline
Backward & First & First &
$\displaystyle \frac{\partial u}{\partial x}\Big|_{x_i} \approx
\frac{u_i - u_{i-1}}{\Delta x} + \mathcal{O}(\Delta x)$
\\
\hline
Forward & First & Second &
$\displaystyle \frac{\partial u}{\partial x}\Big|_{x_i} \approx
\frac{-u_{i+2} + 4u_{i+1} - 3u_i}{2\,\Delta x}
+ \mathcal{O}((\Delta x)^2)$
\\
\hline
\textbf{Central} & \textbf{First} & \textbf{Second} &
$\displaystyle \frac{\partial u}{\partial x}\Big|_{x_i} \approx
\frac{u_{i+1} - u_{i-1}}{2\,\Delta x}
+ \mathcal{O}((\Delta x)^2)$
\\
\hline
Backward & Second & First &
$\displaystyle \frac{\partial^2 u}{\partial x^2}\Big|_{x_i} \approx
\frac{u_i - 2u_{i-1} + u_{i-2}}{(\Delta x)^2}
+ \mathcal{O}(\Delta x)$
\\
\hline
\textbf{Central} & \textbf{Second} & \textbf{Second} &
$\displaystyle \frac{\partial^2 u}{\partial x^2}\Big|_{x_i} \approx
\frac{u_{i-1} - 2u_i + u_{i+1}}{(\Delta x)^2}
+ \mathcal{O}((\Delta x)^2)$
\\
\hline
Forward & Second & Second &
$\displaystyle \frac{\partial^2 u}{\partial x^2}\Big|_{x_i} \approx
\frac{2u_i - 5u_{i+1} + 4u_{i+2} - u_{i+3}}{(\Delta x)^2}
+ \mathcal{O}((\Delta x)^2)$
\\
\hline
\end{tabularx}
\caption{Representative finite-difference stencils for first/second spatial derivatives.}
\label{table:famous-schemes}
\end{table}

Substituting \eqref{eq:D1-central}–\eqref{eq:D2-central} into \eqref{eq:BS-PDE} yields the bulk semi-discrete equations, for $k=1,\dots,N_s-1$,
\begin{align}
\label{eq:bulk-semidiscrete}
\frac{d V_k}{dt}
&= r V_k
- \frac{1}{2}\,\sigma(S_k)^2 S_k^2 \,\frac{V_{k-1}-2V_k+V_{k+1}}{(\Delta S)^2}
- r S_k \,\frac{V_{k+1}-V_{k-1}}{2\,\Delta S}.
\end{align}
The right boundary condition is enforced via a ghost point $V_{N_s}(t)$ placed one step outside the computational domain. This lets us keep a centered second-order stencil for the Neumann slope at $S=S_{\max}$:
\begin{align}
\label{eq:BC-right-centered}
\frac{V_{N_s}(t)-V_{N_s-2}(t)}{2\,\Delta S} = 1
\quad\Longrightarrow\quad
V_{N_s}(t) = V_{N_s-2}(t) + 2\,\Delta S.
\end{align}
We adopt this Neumann view because the quantum framework \cite{guseynov2025quantumframeworksimulatinglinear,guseynov2024efficientPDE} selected in this paper does not support time-dependent boundary data. On the left boundary, the call’s value vanishes and is time-independent as well, $V_0(t)=0$. Consequently, the semi-discrete ordinary differential equation (ODE) system for the bulk nodes $k=1,\dots,N_s-1$ together with the two spatial boundary conditions can be written compactly as
\begin{align}
\label{eq:BS-semidiscrete-system}
\left\{
\begin{array}{ll}
\displaystyle \frac{d V_k}{dt}
= r V_k
- \frac{1}{2}\,\sigma(S_k)^2 S_k^2 \,\frac{V_{k-1}-2V_k+V_{k+1}}{(\Delta S)^2}
- r S_k \,\frac{V_{k+1}-V_{k-1}}{2\,\Delta S},
& k=1,\dots,N_s-1, \\[10pt]
\displaystyle V_0(t)=0,
 &\text{Dirichlet at }  S=0, \\[6pt]
\displaystyle \frac{V_{N_s}(t)-V_{N_s-2}(t)}{2\,\Delta S} = 1,
& \text{Neumann at } S=S_{\max}.
\end{array}
\right.
\end{align}
Here $\Delta S=S_{k+1}-S_k=S_{\max}/(N_s-1)$. Define the interior state vector, consistent with \eqref{eq:state-vector-curly}, as
\begin{align}
\label{eq:state-vector-interior}
\vec V(t) \equiv \{V_1(t),\dots,V_{N_s-1}(t)\}^{\top},
\end{align}
excluding the fixed Dirichlet node $V_0(t)=0$ and eliminating the ghost $V_{N_s}(t)$ by the Neumann relation \eqref{eq:BC-right-centered}. Then the ODE system \eqref{eq:BS-semidiscrete-system} can be written in affine form
\begin{align}
\label{eq:semidiscrete-affine}
\frac{d\vec V}{dt} = A\vec V + \vec b,
\qquad t\in[0,T),
\end{align}
with terminal data inherited from the payoff at $t=T$,
\begin{align}
\label{eq:terminal-data-vector}
\vec V(T) = \bigl(\max(0,S_1-K),\dots,\max(0,S_{N_s-1}-K)\bigr)^{\top}.
\end{align}
The matrix $A\in\mathbb{R}^{(N_s-1)\times(N_s-1)}$ is tridiagonal (i.e. has sparsity $\mathfrak s =3$); the vector $\vec b\in\mathbb{R}^{N_s-1}$ is zero everywhere except at the last interior node $k=N_s-1$, where the ghost elimination $V_{N_s}=V_{N_s-2}+2\Delta S$ adds a constant contribution. Substituting $V_{k+1}=V_{k-1}+2\Delta S$ at $k=N_s-1$ into \eqref{eq:BS-semidiscrete-system} yields
\begin{align}
\label{eq:b-last}
b_k = 0 \quad \text{for } k< N_s-1,
\qquad
b_{N_s-1} = -\frac{\sigma(S_{N_s-1})^2 S_{N_s-1}^2}{\Delta S} - r\,S_{N_s-1}.
\end{align}
Even though the continuous Black--Scholes PDE \eqref{eq:BS-PDE} is homogeneous, the semi-discrete ODE acquires a nonzero source $\vec b$ supported at the right boundary. To enable the quantum-friendly homogeneous form of the semi-discrete dynamics, following \cite{jin2024schr}, we augment the affine system \eqref{eq:semidiscrete-affine} by a single constant scalar state. Define $h(t)\equiv 1$, enforced dynamically by $\dot h(t)=0$ with $h(T)=1$, and introduce the augmented state
\begin{align}
\label{eq:aug-state}
\vec w(t)\equiv
\begin{bmatrix}
\vec V(t)\\[2pt] h(t)
\end{bmatrix}
\in\mathbb{R}^{N_s},
\end{align}
and the block $N_s\times N_s$ matrix formed by concatenation
\begin{align}
\label{eq:S-matrix}
S\equiv
\begin{bmatrix}
A & \vec b\\
\vec 0^{\top} & 0
\end{bmatrix}
\in\mathbb{R}^{N_s\times N_s}.
\end{align}
Then the affine ODE becomes a homogeneous linear system of size increased by one:
\begin{align}
\label{eq:homog-ode}
\frac{d\,\vec w}{dt}=S\,\vec w,
\qquad
\vec w(T)=
\begin{bmatrix}
\vec V(T)\\[2pt] 1
\end{bmatrix},
\end{align}
where the terminal vector $\vec V(T)$ is the payoff sampling from \eqref{eq:terminal-data-vector}. Thus the original affine semi-discrete Black--Scholes dynamics is embedded into a homogeneous system with only one additional grid point.

\begin{figure}[h!]
    \centering

    \begin{minipage}[t]{0.48\linewidth}
        \centering
        \includegraphics[width=\linewidth]{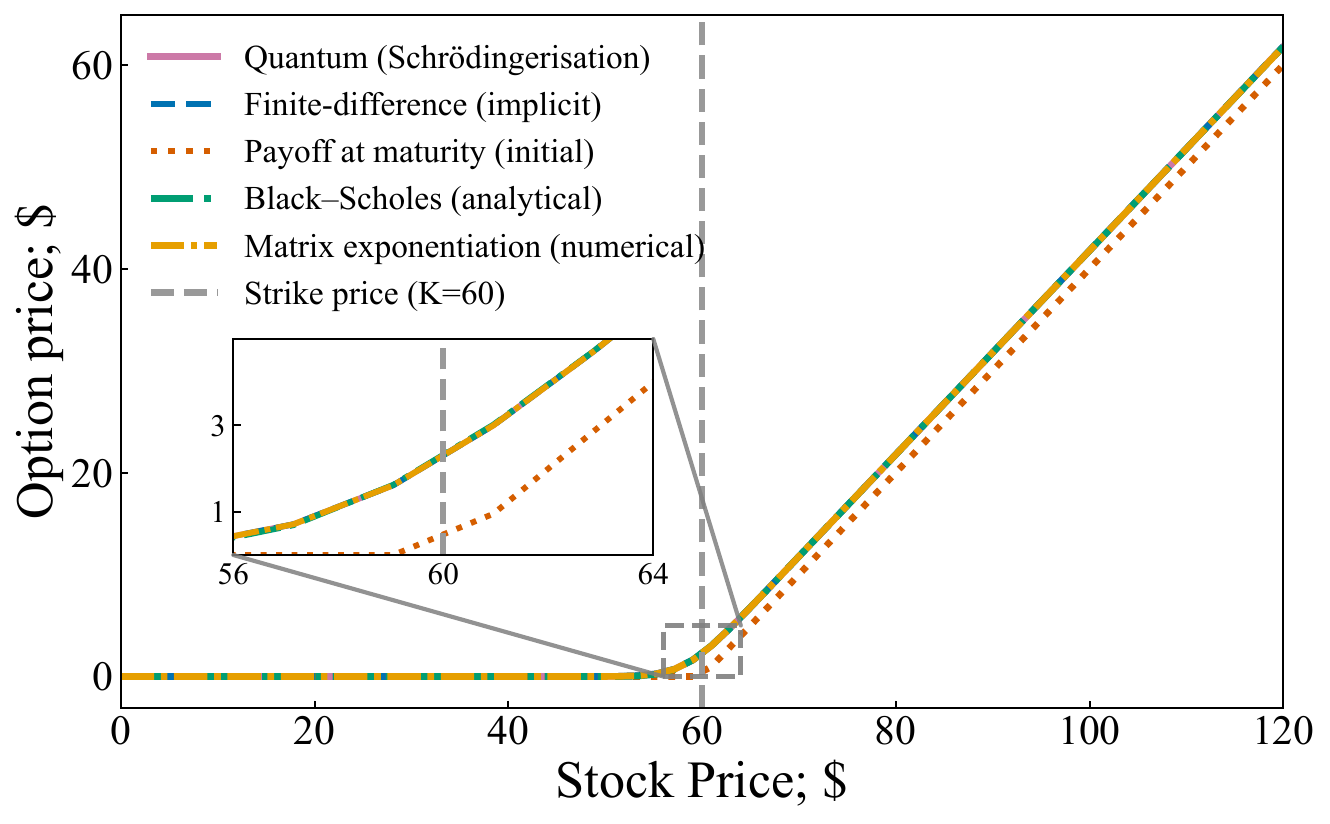}\\[2pt]
        \textbf{(A)}
    \end{minipage}\hfill
    \begin{minipage}[t]{0.48\linewidth}
        \centering
        \includegraphics[width=\linewidth]{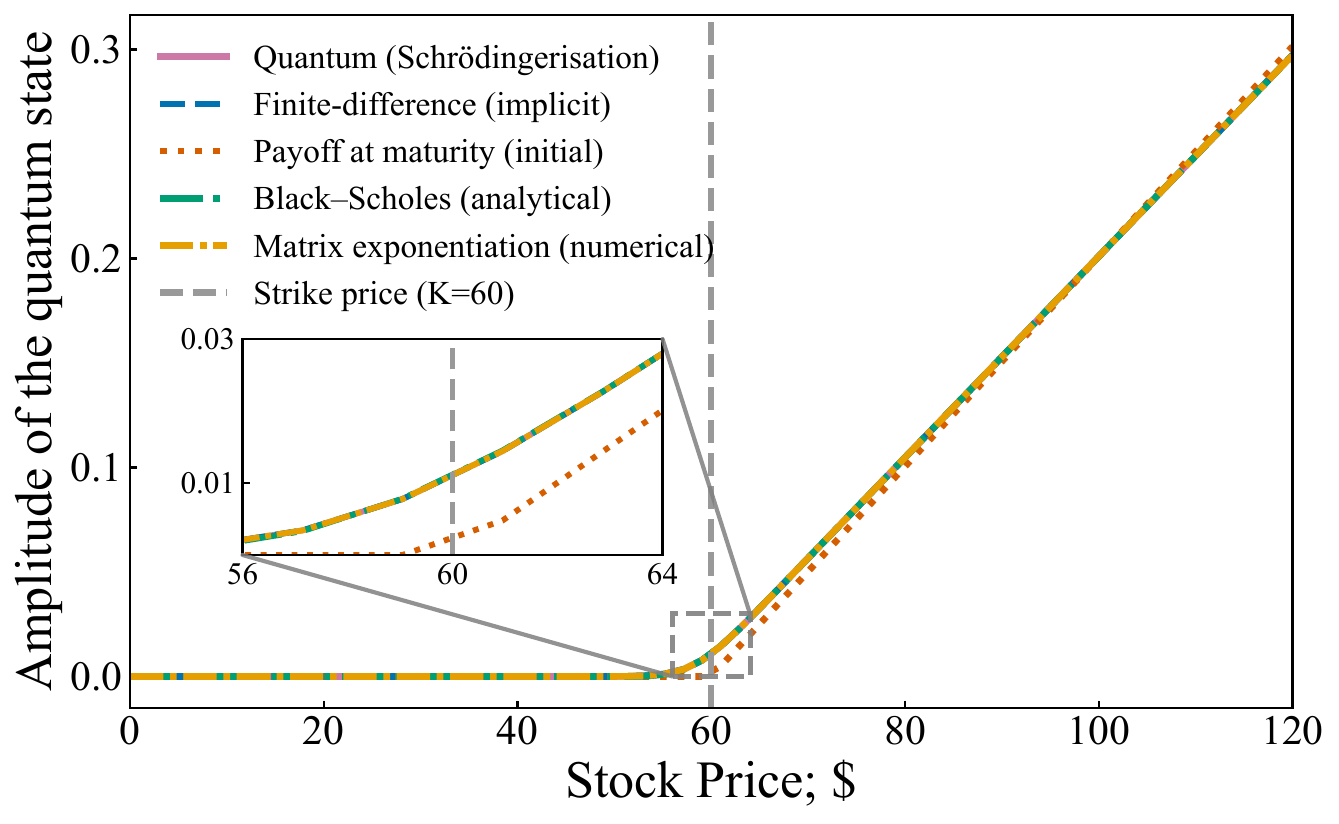}\\[2pt]
        \textbf{(B)}
    \end{minipage}

    \vspace{0.8em}

    \begin{minipage}[t]{0.48\linewidth}
        \centering
        \includegraphics[width=\linewidth]{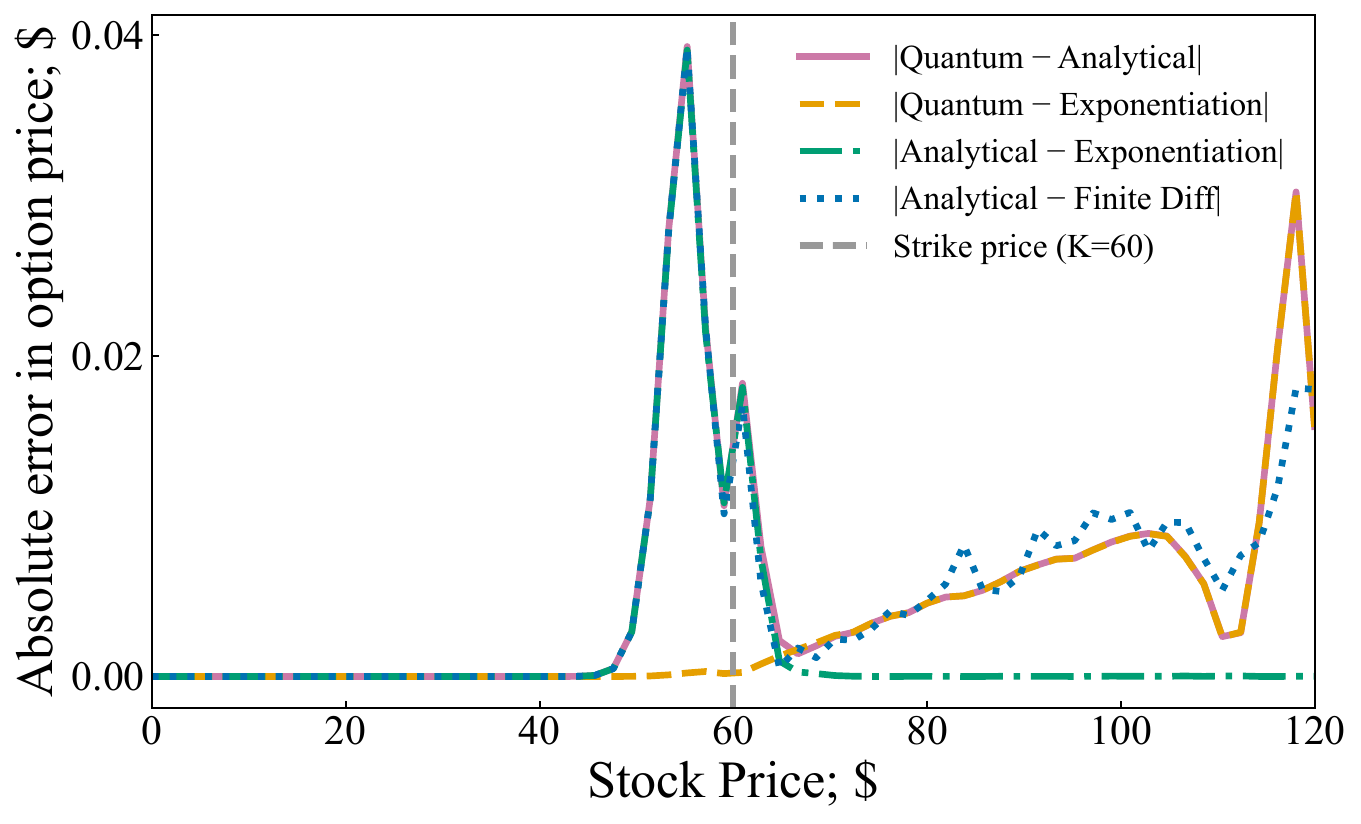}\\[2pt]
        \textbf{(C)}
    \end{minipage}

    \caption{
    Comparison of solvers for the 1D Black--Scholes PDE.
    \textbf{(A)} Solution $V(S,0)$ in dollars, together with the payoff at maturity
    $V(S,T)=\max(S-K,0)$ (shown only as an initial/terminal condition, not a PDE solution).
    \textbf{(B)} The same solution after normalization as a wave function, i.e.\ $\int |V(x)|^2\,dx=1$
    (continuous notation; discretized analog used numerically).
    \textbf{(C)} Pairwise absolute differences between solvers.
    We compare: Quantum (Schr\"odingerisation), Finite-difference (implicit),
    Matrix exponentiation (numerical), and the analytical Black--Scholes formula \cite{BlackScholes1973,Merton1973}.
    The classical baselines (matrix exponentiation and implicit finite differences) are described in Appendix~\ref{appendix:numerical methods}.
    The quantum curve corresponds to simulating only the evolution method described in Appendix~\ref{section:evoltuion non conservative};
    we do \emph{not} simulate state preparation or readout in this paper.
    All numerical parameters used in these simulations are listed in Table~\ref{tab:bs1d-params}.
    }
    \label{fig:bs-solver-comparison}
\end{figure}

\subsection{Numerical simulation}
\label{sec:bs1d-numerical-test}

In this subsection we consider numerical validation of
the framework from Appendix~\ref{section:evoltuion non conservative},
applying it to the one-dimensional Black--Scholes PDE in
Eq.~\eqref{eq:BS-PDE} with a European call terminal payoff and with constant volatility. The numerical implementation is available at \cite{guseynov_gitlab_placeholder}.

We solve the PDE backward in time on the truncated
domain $S\in[0,S_{\max}]$ and compare three numerical
approaches: a finite-difference method, a matrix-exponentiation
method, and the quantum PDE solver introduced earlier.
All classical numerical implementations are described in
Appendix~\ref{appendix:numerical methods}.

As an additional comparison, we evaluate the analytical
Black--Scholes European call price on the same grid, i.e.,
$V_{\mathrm{ref}}(S,t)=V_{\mathrm{BS}}^{\mathrm{call}}(S,t)$,
obtained from the closed-form risk-neutral valuation formula
for a lognormal underlying \cite{BlackScholes1973, Merton1973}.

\begin{equation}
\label{eqn:bs-call-analytic}
V_{\mathrm{BS}}^{\mathrm{call}}(S,t)
=
S\,\Phi(d_{1})
-
K\,e^{-r(T-t)}\Phi(d_{2}),
\quad
d_{1}=\frac{\ln(S/K)+(r+\sigma^{2}/2)(T-t)}{\sigma\sqrt{T-t}},
\quad
d_{2}=d_{1}-\sigma\sqrt{T-t}.
\end{equation}

All numerical parameters used in our test are summarized in
Table~\ref{tab:bs1d-params}, written exactly according to the
program settings, including the Schr\"odingerisation auxiliary
discretization $N_{\xi}=1024$ and its qubit count $n_{\xi}=10$.

\begin{table}[t!]
\centering
\renewcommand{\arraystretch}{1.15}
\begin{tabular}{l l l}
\hline\hline
Variable & Definition & Num.\ value \tabularnewline
\hline
$S_{\max}$ & Maximum stock price (domain truncation) & $120$ \tabularnewline
$T$ & Time to maturity & $1$ \tabularnewline
$r$ & Risk-free rate & $0.03$ \tabularnewline
$\sigma$ & Volatility & $0.05$ \tabularnewline
$K$ & Strike price & $60$ \tabularnewline
$S_{0}$ & Initial stock price & $50$ \tabularnewline
$N$ & Number of price steps (grid points in $S$) & $2^{6}$ \tabularnewline
$N_{t}$ & Number of time steps & $T\,N^{2}=4096$ \tabularnewline
$\Delta t$ & Time step size & $T/N_{t}=1/4096$ \tabularnewline
$\Delta S$ & Stock-price step size & $S_{\max}/(N-1)=120/63$ \tabularnewline
$N_{\xi}$ & Schr\"odingerisation auxiliary grid size & $1024$ \tabularnewline
$n_{\xi}$ & Qubits for auxiliary dimension ($N_{\xi}=2^{n_{\xi}}$) & $10$ \tabularnewline
\hline\hline
\end{tabular}
\caption{Numerical parameters for the one-dimensional Black--Scholes test
(as used in the implementation). The quantum parameters are clarified in Appendix~\ref{section:evoltuion non conservative}. The classical numerical methods are clarified in Appendix~\ref{appendix:numerical methods}.}
\label{tab:bs1d-params}
\end{table}

Let $V^{(m)}(S,0)$ denote the computed option value at time
$t=0$ for method
$m\in\{\mathrm{FD},\mathrm{EXP},\mathrm{Q}\}$, and let
$V_{\mathrm{ref}}(S,0)$ be the analytical reference.
We report the absolute pointwise error on the discrete grid
$\{S_{k}\}_{k=0}^{N_s-1}$.

\begin{equation}
\label{eqn:bs1d-absolute-error}
\varepsilon^{(m)}_{k}
:=
\bigl|V^{(m)}(S_{k},0)-V_{\mathrm{ref}}(S_{k},0)\bigr|,
\qquad
S_{k}=k\,\Delta S,
\qquad
k=0,1,\ldots,N_s-1.
\end{equation}

The numerical results are shown in Fig.~\ref{fig:bs-solver-comparison},
where we plot the pairwise absolute differences
$|\mathrm{Quantum}-\mathrm{Analytical}|$,
$|\mathrm{Quantum}-\mathrm{Exponentiation}|$,
$|\mathrm{Analytical}-\mathrm{Exponentiation}|$, and
$|\mathrm{Analytical}-\mathrm{Finite\ Diff}|$, together with
a vertical marker at the strike price $K=60$. The simulation uses $17$ qubits in total: (i) $6$ for main stock price register, (ii) $10$ for Schr\"odingerisation dimension (see Appendix~\ref{section:evoltuion non conservative}), (iii) $1$ for making ODE system homogeneous (see Eq.~\ref{eq:S-matrix})

\noindent
We conclude with the following observations:

\begin{enumerate}
\item
Near $S=S_{\max}$ the dominant error is induced by the
artificial Neumann boundary condition $\partial_{S}V=1$.
Nevertheless, the resulting deviation remains acceptable,
approximately four cents in our parameter regime, and
the quantum solution mimics the finite-difference error
behavior in this boundary-dominated region.
\item
Near $S\approx K$ the error increases due to the coarse
spatial grid, which poorly resolves the payoff kink.
We intentionally choose $N=2^{6}$ to demonstrate this effect.
\item
Overall, the quantum-solver errors are at a similar level
to those produced by the classical finite-difference and
exponentiation baselines under the same discretization, which validates the quantum method.
\end{enumerate}

\subsection{End-to-end complexity scaling}\label{subsec:1d BS scaling}
To highlight the dependence on the spatial grid size $N=2^{n}$ (and suppress lower-order factors including polylogarithmic factors in the target accuracy, e.g.\ $\log(1/\epsilon)$), we collect the dominant costs of the Black--Scholes pipeline for retrieving a single value $V(t=0,S_{0})$ at fixed $(K,\sigma)$. State preparation is described in Appendix~\ref{section:state prep}, the backward-in-time evolution $u(T,S)\mapsto u(0,S)$ is simulated as in Appendix~\ref{section:evoltuion non conservative}, and the measurement overhead is discussed in Appendix~\ref{section:readout}. The state-preparation and evolution costs add, whereas the readout overhead is multiplicative:
\begin{equation}
\label{eq:bs-end-to-end-scaling}
\Bigl(
\underbrace{n\log n + n_{\xi}\log n_{\xi}}_{\text{state prep.}}
+
\underbrace{2^{2n}T\bigl(Q_{\sigma}n\log n + n_{\xi}\log n_{\xi}\bigr)}_{\text{evolution}}
\Bigr)
\cdot
\underbrace{2^{n/2}}_{\text{readout}},
\end{equation}
where $Q_{\sigma}$ is the maximum degree of the polynomial approximating the local volatility $\sigma$; $n_\xi$ is the number of qubits in the auxiliary dimension controlling the introduced error (Appendix~\ref{section:evoltuion non conservative}). Note that $T$ is the simulation time.

In contrast, the classical baselines (implicit finite differences and matrix-exponential action) summarized in Appendix~\ref{appendix:numerical methods} scale as $\mathcal{O}(2^{3n})$ for this 1D discretization. Hence, in this single-point Black--Scholes setting, the quantum approach suggests a modest asymptotic improvement from $\mathcal{O}(2^{3n})$ to roughly $\mathcal{O}(2^{2.5n})$, up to the suppressed $poly(n)$ factors; we warn the reader that this advantage should be taken with care, see Remark~\ref{rem:universal-disclaimer-ft-topology}.

\FloatBarrier
\section{Multi-asset Black--Scholes PDE}
In a one-dimensional setting, Black--Scholes has an analytical formula. The probability density of the terminal stock price is log-normal, and the domain of integration for the payoff is a simple half-line (e.g., $S_T > K$). This symmetry allows the risk-neutral expectation to be solved analytically in terms of the cumulative normal distribution function, $\Phi(\cdot)$.

However, this analytical tractability rarely extends to the multidimensional case described in Eq.~\eqref{eqn:ICMultidimblack-scholes}. 

Let us state the multidimensional version, which is a linear, parabolic, second-order PDE \cite{alden2022model,guillaume2019multidimensional}. As with the one-dimensional case, the multidimensional terminal condition problem is
\begin{align}\label{eqn:ICMultidimblack-scholes}
&\frac{\partial V}{\partial t}
+ \sum_{i=1}^d r\,S_i \frac{\partial V}{\partial S_i}
+ \frac{1}{2}\sum_{i=1}^d \sigma_i^2 S_i^2 \frac{\partial^2 V}{\partial S_i^2}
+ \sum_{1 \le i < j \le d} \rho_{ij}\,\sigma_i \sigma_j\, S_i S_j \frac{\partial^2 V}{\partial S_i \partial S_j}
- r\,V \;=\; 0,\\
&V(S,T)=f(S)\nonumber
\end{align}
where \(f(S)\) denotes the payoff function. We examine some common payoff functions below. Underlying asset prices have to be positive,
$S_i \in (0,\infty)$ for $i=1,\dots,d$, $t \in [0,T)$, and
$(\rho_{ij})_{i,j=1}^d$ is a symmetric positive semidefinite correlation matrix with $\rho_{ii}=1$. 

\subsection{Typical Vanilla Payoffs}\label{sec:typical payoffs}
In the multidimensional Black--Scholes model with $d$ risky assets
$S_t = (S^1_t,\dots,S^d_t)$, the arbitrage-free price of a European derivative
with maturity $T$ and payoff $f(S_T)$ is given by the risk-neutral valuation formula
\cite{BlackScholes1973,Merton1973}
\begin{equation}
     V(t,S_t) = e^{-r(T-t)}\,\mathbb{E}^{\mathbb{Q}}\!\left[f(S_T)\,\middle|\,S_t\right],
\end{equation}

where $\mathbb{Q}$ is the risk-neutral measure. In general, this representation
does not yield a closed-form expression, except for certain special payoffs.

Below we collect several standard multivariate vanilla payoffs; see, e.g.,
\cite{BoyleNumerical,carmona2005generalizing,JohnsonMaxMinmulti}.

\begin{table}[ht]
\centering
\caption{Standard multivariate vanilla payoffs.}\label{tab:vanilla payoffs}
\renewcommand{\arraystretch}{2.2}
\begin{tabular}{@{}lll@{}}
\toprule
\textbf{Type} & \textbf{Payoff} & \textbf{Remarks} \\
\midrule
Basket Call
  & $\displaystyle\max\!\left(\sum_{i=1}^{d} w_i S^i_T - K,\, 0\right)$ 
  & $B_T = \sum_{i=1}^{d} w_i S^i_T$, with $w_i > 0$.\\
Basket Put
  & $\displaystyle\max\!\left(K - \sum_{i=1}^{d} w_i S^i_T,\, 0\right)$
  & \\
\midrule
Spread
  & $\displaystyle\max\!\left(S^1_T - S^2_T - K,\, 0\right)$
  & Two assets \\
Exchange
  & $\displaystyle\max\!\left(S^1_T - S^2_T,\, 0\right)$
  & Margrabe~(1978) \cite{Margrabe}; $K=0$ spread \\
\midrule
Best-of Call
  & $\displaystyle\max\!\left(\max_{1\le i\le d} S^i_T - K,\, 0\right)$
  & \multirow{2}{*}{\parbox{4.5cm}{Converge as $\rho_{ij}\!\uparrow\!1$ for all $i,j$}} \\
Worst-of Call
  & $\displaystyle\max\!\left(\min_{1\le i\le d} S^i_T - K,\, 0\right)$ 
  & \\
\bottomrule
\end{tabular}
\end{table}

The effective volatility for spread/exchange options under the bivariate
Black--Scholes model is
$\sigma_{12}^2 = \sigma_1^2 + \sigma_2^2 - 2\rho\,\sigma_1\sigma_2$.

While the governing Partial Differential Equation (PDE) remains linear and parabolic, fundamental obstacles prevent the derivation of a general closed-form formula for multi-asset derivatives:

\begin{enumerate}
    \item \textit{The Distributional Problem (Sums of Lognormals):} 
    While the individual assets $S_i$ follow log-normal diffusions, the weighted sum of these assets (as found in Basket options, $B_T = \sum w_i S_i$) does not. The sum of correlated log-normal variables follows no known parametric distribution. Without a closed-form probability density function for the sum, the risk-neutral expectation integral cannot be solved analytically. Therefore, the risk-neutral price
$\mathbb{E}\big[(B_T - K)_+\big]$
cannot be evaluated explicitly because the distribution of $B_T$ is not tractable.

    \item \textit{The Geometric Problem (Integration Boundaries):} 
    Even when the joint multivariate log-normal density is explicitly known, the payoff typically induces an integration region whose geometry renders analytic evaluation impossible. Several canonical examples illustrate this difficulty:
\begin{itemize}
    \item \emph{Basket options:}  
    The payoff depends on $\sum_{i=1}^d w_i S_i > K$, which corresponds to a curved hypersurface in the log-price coordinates 
    $(\log S_1, \ldots, \log S_d)$.

    \item \emph{Spread options:}  
    The exercise condition $S_1 - S_2 > K$ becomes a tilted, nonlinear boundary when expressed in log-space.

    \item \emph{Rainbow/max options:}  
    The exercise region is the union of polyhedral sectors with discontinuous boundaries.

    \item \emph{Asian multi-asset options:}  
    The situation is even more complex, involving high-dimensional, path-dependent, curved integration domains.
    \end{itemize}

\item \textit{Boundary and degeneracy problem (PDE domain truncation):}
Beyond the intractability of the valuation integral, high-dimensional PDE solvers
require boundary conditions on a truncated computational domain. This is delicate
for two distinct reasons: (i) the operator degenerates on the coordinate faces
$\{S:\ S_i=0\}$ (because the diffusion coefficients scale like $S_i^2$ and cross terms
like $S_i S_j$), and (ii) the outer boundaries $S_i=S_{i,\max}$ are artificial
truncations that must approximate the $S_i\to\infty$ asymptotics.
A common simplification is homogeneous Dirichlet data on $\{S:\ S_i=0\}$, but this is
only valid for payoffs that vanish whenever any underlying is worthless (e.g.,
minimum/worst-of structures).
For example, in a three-dimensional Black--Scholes solver for a minimum-driven
payoff, \cite{Kimkim3dBlackScholes} enforces
\begin{equation}
u(0,y,z,\tau)=0,\quad u(x,0,z,\tau)=0,\quad u(x,y,0,\tau)=0,
\end{equation}
implemented discretely as $u_{0jk}=u_{i0k}=u_{ij0}=0$, together with linear
(extrapolation/ghost-point) conditions at the upper truncation faces.

Achdou and Pironneau \cite{AchdouPironneau} instead emphasize that on $\{\,S\in\mathbb{R}_+^d:\ S_i=0\,\}$ the Black--Scholes operator becomes degenerate and the pricing problem naturally
reduces to a $(d-1)$-dimensional PDE with boundary payoff $f(0,S_{-i})$. This viewpoint avoids imposing artificial boundary values and yields the correct
boundary behavior for general payoffs (e.g., basket, spread, best-of) that do not
vanish when a single asset hits zero. Unlike the one-dimensional half-line, or the affine constraints that allow Gaussian integrals to be evaluated explicitly, these payoff boundaries are non-rectangular, non-separable, and do not align with the symmetries of the multivariate log-normal distribution. Consequently, the valuation integral 

\begin{equation}
 \int_{\mathcal{D}_+} f_{S_T}(s)\, ds,
\qquad s=(s_1,\ldots,s_d), \text{ and } \mathcal{D}_+
\;:=\;
\{\, s\in\mathbb{R}_+^d \mid f(S_T)>0 \,\}   
\end{equation}
admits no closed-form primitive.
\end{enumerate}

Thus, except for specific special cases such as the Geometric Basket option (where the product of log-normals remains log-normal) or the Exchange option (solved via Margrabe's formula), there are no exact analytical solutions for multidimensional Black--Scholes problems. To price these instruments and compute their sensitivities (Greeks) accurately, we must rely on numerical methods to solve the PDE \eqref{eq: operator form multi Black-Scholes} directly.

\begin{equation}
    \begin{gathered}
        \frac{\partial V}{\partial t} + \mathcal{L}V = 0,
\qquad
V(S,T)=f(S),
\qquad
S=(S_1,\ldots,S_d)^\top \in \mathbb{R}_{+}^{d}.\\
(\mathcal{L}V)(S,t)
:=
\sum_{i=1}^{d} r S_i \frac{\partial V}{\partial S_i}
+
\frac{1}{2} \sum_{i=1}^{d} \sigma_i^{2} S_i^{2}
\frac{\partial^{2} V}{\partial S_i^{2}}
+
\sum_{1 \le i < j \le d} \rho_{ij} \sigma_i \sigma_j
S_i S_j \frac{\partial^{2} V}{\partial S_i \partial S_j}
-
r V
    \end{gathered}\label{eq: operator form multi Black-Scholes}
\end{equation}
where $f(S)$ denotes the option payoff function,
determined by the specific derivative contract under study;
see Appendix~\ref{sec:typical payoffs} for typical payoff forms.

\subsection{Linear ODE system for multi-asset Black--Scholes PDE}
\label{sec:bs-multiasset-discretization}
We explain how the continuous operator form in Eq.~\eqref{eq: operator form multi Black-Scholes}
is reduced to the linear ODE form required by the quantum evolution framework of
Appendix~\ref{section:evoltuion non conservative}. The overall workflow closely mimics the
one-dimensional construction in Appendix~\ref{sec:discretization}; here we highlight only the
multi-dimensional features, in particular the treatment of mixed derivatives and boundary conditions.

Consider $d$ underlying assets with prices $S=(S_{1},\ldots,S_{d})$ and option value
$V(t,S)$ on a Cartesian grid. Let the uniform mesh sizes be $\Delta S_{i}$ along each
$S_{i}$-axis, and denote grid points by
$S^{(\boldsymbol{k})}=(k_{1}\Delta S_{1},\ldots,k_{d}\Delta S_{d})$ with multi-index
$\boldsymbol{k}=(k_{1},\ldots,k_{d})$. We collect the sampled values into a vector
$\vec{V}(t)\in\mathbb{R}^{N^{d}}$ (lexicographic ordering of $\boldsymbol{k}$).

We discretise each asset axis $S_i\in[0,S_{i,\max}]$ on a uniform grid with spacing
$\Delta S_i$, and write grid points as $S^{(\boldsymbol{k})}$ for a multi-index
$\boldsymbol{k}=(k_1,\ldots,k_d)$. Denote by $\boldsymbol{e}_i$ the unit multi-index in
direction $i$. Following the standard second-order central schemes (summarised in
Table~\ref{table:famous-schemes}), we approximate the drift and diffusion derivatives by
\begin{align}
\frac{\partial V}{\partial S_{i}}(t,S^{(\boldsymbol{k})})
&\approx
\frac{V\!\left(t,S^{(\boldsymbol{k}+\boldsymbol{e}_{i})}\right)
-
V\!\left(t,S^{(\boldsymbol{k}-\boldsymbol{e}_{i})}\right)}{2\Delta S_{i}},
\label{eq:fd-first-deriv-Si}
\\
\frac{\partial^{2} V}{\partial S_{i}^{2}}(t,S^{(\boldsymbol{k})})
&\approx
\frac{V\!\left(t,S^{(\boldsymbol{k}+\boldsymbol{e}_{i})}\right)
-
2V\!\left(t,S^{(\boldsymbol{k})}\right)
+
V\!\left(t,S^{(\boldsymbol{k}-\boldsymbol{e}_{i})}\right)}{\Delta S_{i}^{2}}.
\label{eq:fd-second-deriv-SiSi}
\end{align}
The multi-asset case differs from $d=1$ only through the correlation terms, which
introduce mixed derivatives $\partial^{2}V/(\partial S_i\partial S_j)$ for $i\neq j$.
Using the same central-difference philosophy, we employ the four-point stencil
\begin{equation}
\frac{\partial^{2} V}{\partial S_{i}\partial S_{j}}(t,S^{(\boldsymbol{k})})
\approx
\frac{
V\!\left(t,S^{(\boldsymbol{k}+\boldsymbol{e}_{i}+\boldsymbol{e}_{j})}\right)
-
V\!\left(t,S^{(\boldsymbol{k}+\boldsymbol{e}_{i}-\boldsymbol{e}_{j})}\right)
-
V\!\left(t,S^{(\boldsymbol{k}-\boldsymbol{e}_{i}+\boldsymbol{e}_{j})}\right)
+
V\!\left(t,S^{(\boldsymbol{k}-\boldsymbol{e}_{i}-\boldsymbol{e}_{j})}\right)
}{4\,\Delta S_{i}\Delta S_{j}},
\label{eq:fd-mixed-deriv-SiSj}
\end{equation}
which couples the four diagonal neighbours in the $(S_i,S_j)$ plane. Substituting
Eqs.~\eqref{eq:fd-first-deriv-Si}--\eqref{eq:fd-mixed-deriv-SiSj} into
Eq.~\eqref{eq: operator form multi Black-Scholes} and collecting all boundary-condition
contributions yields the semi-discrete affine system
$d\vec{V}(t)/dt = A\vec{V}(t)+\vec{b}$.

Boundary conditions (BC) are the most delicate part of the multi-asset discretization. We therefore
warn the reader that, while standard finite-difference solvers can accommodate fairly general (and
even time-dependent) boundary data, the quantum evolution framework adopted in
Appendix~\ref{section:evoltuion non conservative} is restricted to \emph{time-independent, constant-coefficient}
BC of Dirichlet/Neumann/Robin or periodic type. In particular, a Robin BC on a boundary face
$\Gamma$ takes the form
\begin{equation}
\label{eq:robin-bc}
\alpha\,V(t,S) + \beta\,\partial_{\boldsymbol{n}}V(t,S) = \gamma,
\qquad S\in \Gamma,
\end{equation}
where $\boldsymbol{n}$ is the outward normal and $\alpha,\beta,\gamma$ are constants. Throughout the remainder
we use ``BC'' for boundary conditions.

For the multi-asset Black--Scholes PDE, the correct BC depend strongly on the payoff function;
consult Appendix~\ref{sec:typical payoffs} for more details. In a finite-difference setting one must specify BC on each
truncation face $S_i=0$ and $S_i=S_{i,\max}$ for all $i=1,\ldots,d$ \cite{LeVeque2007FDM}. However,
for general payoffs it is \emph{not} always possible to express the correct boundary behavior in the
quantum-compatible form \eqref{eq:robin-bc} with constants. A particularly complicated situation occurs for basket and spread payoffs at the degenerate boundaries $S_i=0$, see Table~\ref{tab:vanilla payoffs}. The mathematically natural
treatment \cite{hirsa2024computational} is that, on the face $S_i=0$, the PDE reduces to a $(d-1)$-dimensional pricing PDE in the
remaining variables $S_{-i}:=(S_1,\ldots,S_{i-1},S_{i+1},\ldots,S_d)$, with terminal condition inherited
from the payoff restricted to that face:
\begin{equation}
\label{eq:reduced-pde-face}
\begin{aligned}
\frac{\partial V}{\partial t}
&+
\sum_{k\neq i} r S_k \frac{\partial V}{\partial S_k}
+
\frac{1}{2}\sum_{k\neq i}\sigma_k^2 S_k^2 \frac{\partial^2 V}{\partial S_k^2}
+
\sum_{\substack{k<\ell\\ k,\ell\neq i}}
\rho_{k\ell}\sigma_k\sigma_\ell S_k S_\ell
\frac{\partial^2 V}{\partial S_k\partial S_\ell}
- rV
=0,\\
&\qquad
V(S_{-i},T)=f(S_1,\ldots,0,\ldots,S_d).
\end{aligned}
\end{equation}

Implementing \eqref{eq:reduced-pde-face} means solving a hierarchy of lower-dimensional PDEs on
faces/edges/corners, which is outside the BC model supported by
Appendix~\ref{section:evoltuion non conservative}. In such cases one may adopt approximations that fit
the allowed BC class (e.g.\ imposing a simplified Neumann condition $\partial_{S_i}V=0$ at $S_i=0$),
but this introduces an additional modelling/discretization error that must be accounted for, which is out of scope of this paper.

By contrast, for payoffs that vanish whenever any underlying is worthless, the BC become much
simpler. For the Worst-of Call (see Table~\ref{tab:vanilla payoffs}), it is natural to impose
\emph{Dirichlet} conditions $V=0$ on each face $S_i=0$ \cite{Kimkim3dBlackScholes}. At the artificial upper truncation faces
$S_i=S_{i,\max}$, a common and stable choice is an asymptotic \emph{Neumann} condition
$\partial_{S_i}V=0$, which approximates saturation of the option value in the far field and is known
to control truncation artifacts in practice \cite{AchdouPironneau,duffy2006fdm,foulon2010adi, Kangrofarfield, TyskSpacetimefinitediff}. These Dirichlet/Neumann BC are
directly compatible with the quantum framework. Therefore, in the remainder of the multi-asset analysis, we restrict attention to the Worst-of Call.

With the interior stencils fixed by Eqs.~\eqref{eq:fd-first-deriv-Si}--\eqref{eq:fd-mixed-deriv-SiSj}
and with a chosen set of BC enforced by standard finite-difference mechanisms (Dirichlet
elimination, ghost points, or Robin relations), substituting the discrete derivatives into
Eq.~\eqref{eq: operator form multi Black-Scholes} produces the semi-discrete affine system 
\begin{equation}
\frac{d\vec{V}(t)}{dt}=A\,\vec{V}(t)+\vec{b}. 
\label{eq:bs-affine-ode}
\end{equation} 
Here $A\in\mathbb{R}^{N^{d}\times N^{d}}$ is sparse, with nonzero pattern determined by the local
stencils (including the mixed-derivative couplings), while $\vec{b}$ collects the contributions
arising from enforcing the BC on the truncated domain.

Finally, we convert the affine ODE~Eq.~\eqref{eq:bs-affine-ode} into a homogeneous linear system
via the standard augmentation trick (see~\cite{jin2024schr}). In our setting, it is convenient
to \emph{double} the state (add one qubit) by appending an auxiliary vector $\vec{r}(t)\in\mathbb{R}^{N^{d}}$ of
the same dimension as $\vec{b}$, which is held constant in time. Introduce a diagonal matrix
$B$ such that $B\vec{r}_0=\vec{b}$ entrywise. Define
\begin{eqnarray}
\label{eq:bs-homog-augmentation-vector}
\begin{gathered}
\frac{d}{dt}
\underbrace{
\left[
\begin{array}{c}
\vec{V}(t)\\
\vec{r}(t)
\end{array}
\right]
}_{\vec{w}(t)}
=
\underbrace{
\left[
\begin{array}{cc}
A & B\\
0 & 0
\end{array}
\right]
}_{S}
\left[
\begin{array}{c}
\vec{V}(t)\\
\vec{r}(t)
\end{array}
\right],
\qquad
\left[
\begin{array}{c}
\vec{V}(T)\\
\vec{r}(T)
\end{array}
\right]
=
\left[
\begin{array}{c}
\vec{V}_0\\
\vec{r}_0
\end{array}
\right],\\[6pt]
\vec{r}(t)=\vec{r}(0)=\vec{r}_0:=\frac{1}{\sqrt{2^{dn}}}(1,1,\ldots,1)^\top\in\mathbb{R}^{N^{d}},
\qquad
B:=\sqrt{2^{dn}}\times \mathrm{diag}\{b_0,b_1,\ldots,b_{N^{d}-1}\},\\[4pt]
\frac{d\vec{w}(t)}{dt}=S\,\vec{w}(t),\qquad \vec{w}(T)=\vec{w}_0,
\end{gathered}\label{eq: homogeneous BS after augmentation}
\end{eqnarray}
which is exactly of the homogeneous form~Eq.~\eqref{eq:MAINhomog-ode} and therefore suitable for the quantum evolution procedure of Appendix~\ref{section:evoltuion non conservative}.

\subsection{End-to-end complexity scaling}
\label{subsec:multiasset BS scaling}
We now state the end-to-end scaling for the \emph{multi-asset} Black--Scholes pipeline, in a form
consistent with the 1D summary in Appendix~\ref{subsec:1d BS scaling}. We use the Worst-of Call payoff state-preparation routine from Table~\ref{table:quantum state prep of payoffs}, with its success probability boosted to $\mathcal{O}(1)$ using the amplification technique in Appendix~\ref{app:amplification}. The backward
evolution cost follows from Appendix~\ref{section:evoltuion non conservative}, in particular
Proposition~\ref{theorem:QuantumHamiltonianSimulation}. Single-point retrieval $V(t=0,\vec S_0)$ uses the readout
procedure of Appendix~\ref{section:readout}.

Suppressing lower-order factors (e.g.\ polylogarithmic in $1/\epsilon$) and writing $N=2^{n}$,
the dominant end-to-end cost for extracting one value $V(t=0,\vec S_0)$ on an $N^{d}$ grid is
\begin{equation}
\label{eq:bs-multiasset-end-to-end-scaling}
\Bigl(
\underbrace{d^{3}n+d^{2}n\log n}_{\text{state prep.}}
+
\underbrace{dT\,2^{2n}\Bigl[dQ_{\sigma}n\log n+n_{\xi}\log n_{\xi}+d^{4} n\Bigr]}_{\text{evolution}}
\Bigr)
\cdot
\underbrace{2^{nd/2}}_{\text{readout}}.
\end{equation}
Here $Q_{\sigma}$ is the maximum polynomial degree used to approximate the local volatility,
$n_{\xi}$ is the number of qubits in the auxiliary dimension introduced by Schr\"odingerisation.

For comparison, the classical multi-asset methods described in Appendix~\ref{appendix:numerical methods}
scale as
\begin{equation}
\label{eq:bs-multiasset-classical-scaling}
\mathcal{C}_{\mathrm{classical}}=\mathcal{O}\!\left(d^{2}2^{n(d+2)}\right).
\end{equation}
Therefore, up to suppressed $poly(d,n)$ factors, the quantum pipeline scales as
$\mathcal{O}\!\left(2^{n(2+d/2)}\right)$ in its leading exponential dependence on $n$ (from evolution
and readout), versus $\mathcal{O}\!\left(2^{n(d+2)}\right)$ classically, suggesting an exponential-in-$n$
advantage in the multi-asset regime; we warn the reader that this advantage should be taken with care, see Remark~\ref{rem:universal-disclaimer-ft-topology}. We conclude that the larger the number of assets~$d$ in the contract, the more pronounced the quantum advantage becomes. However, we note that the so-called ``curse of dimensionality,'' namely the exponential growth of complexity with the dimension~$d$, is not eliminated by the proposed quantum algorithm. It is only mitigated.

\section{Heston}\label{section:Heston}
The Heston model augments the Black--Scholes setting by endogenizing the
instantaneous variance of the underlying through a mean-reverting square-root
process. Variance itself is stochastic; thus, the assumptions of standard Black--Scholes no longer apply.  Introducing stochastic variance can produce a fat-tailed terminal stock-price distribution. Consequently, this produces implied-volatility smiles and skews, while maintaining a
tractable Markovian two-factor structure. In this subsection we restrict
attention to the one-asset case.

Let $S_t$ denote the underlying price and let $v_t$ denote its instantaneous
variance. Under the risk-neutral measure $\mathbb{Q}$ the Heston dynamics are
\begin{align}
dS_t &= r S_t\,dt + \sqrt{v_t}\,S_t\, dW^{(S)}_t, \label{eq:HestonSDE_S}\\
dv_t &= \kappa(\theta - v_t)\,dt + \sigma\sqrt{v_t}\, dW^{(v)}_t, 
\label{eq:HestonSDE_v}
\end{align}
where the Brownian drivers satisfy
\[
dW^{(S)}_t\, dW^{(v)}_t = \rho\, dt,
\]
with $r$ the risk-free rate, $\kappa>0$ the mean-reversion speed of the variance,
$\theta>0$ its long-run mean, $\sigma>0$ the volatility of variance, and
$\rho\in[-1,1]$ the instantaneous correlation between price and volatility shocks.

Let $V(S,v,t)$ denote the arbitrage-free value of a European-style derivative
with maturity $T$ and payoff $V(S,v,T)=f(S)$. By risk-neutral valuation, the discounted process $e^{-rt}V(S_t,v_t,t)$ is a $\mathbb Q$-martingale under~\eqref{eq:HestonSDE_S}--\eqref{eq:HestonSDE_v},
which yields the parabolic pricing PDE
\begin{align}
\frac{\partial V}{\partial t}
&\;+\; \frac{1}{2} v S^{2}\, \frac{\partial^{2} V}{\partial S^{2}}
\;+\; \rho\, \sigma\, v S\, \frac{\partial^{2} V}{\partial S \partial v}
\;+\; \frac{1}{2}\sigma^{2} v\, \frac{\partial^{2} V}{\partial v^{2}}
\label{eq:HestonPDE}\\[2pt]
&\;+\; r S\, \frac{\partial V}{\partial S}
\;+\; \kappa(\theta - v)\, \frac{\partial V}{\partial v}
\;-\; r V
\;=\; 0, 
\qquad S>0,\; v\ge 0,\; 0\le t<T. \nonumber
\end{align}

At maturity the solution is fixed by the payoff:
\begin{equation}
V(S,v,T) = f(S).
\label{eq:HestonTerminal}
\end{equation}

The mixed derivative term $\partial^2_{S v} V$ and the degeneracy at $v=0$
distinguish the Heston operator from the Black--Scholes case. After spatial
discretization, \eqref{eq:HestonPDE} takes the homogeneous augmented ODE form
$d\vec{w}/dt=S\vec{w}$ considered in Appendix~\ref{section:evoltuion non conservative}, enabling Schr\"odingerisation and
subsequent quantum evolution.
\subsection{Implied-volatility slices in Heston: arbitrage-free regularization}
\label{subsec:ssvi-main}

To compare Heston-generated prices (classical or quantum) away from the at-the-money (ATM) strike,
we work with implied-volatility \emph{slices} at fixed maturity. In practice, quoted strikes are sparse
around ATM, and in our setting the inversion from prices to implied volatility is further corrupted by
quantum sampling / amplitude-estimation noise. A naïve interpolation in implied volatility (e.g.\ cubic
splines) can therefore introduce \emph{static} arbitrage even when the underlying model is arbitrage-free:
smoothness of $\sigma_{\mathrm{BS}}(K,T)$ does not imply convexity of the call price in $K$.

When convexity fails, $\partial_{KK}C(K,T)$ becomes negative on some region, and the
Breeden--Litzenberger identity interprets this as a negative ``density'', signalling butterfly arbitrage
\cite{BreedenLitzenberger}. Concretely, for a discounted call
\begin{equation}
C(K,T)=e^{-rT}\int_K^\infty (s-K)\,p_{S_T}(s)\,ds,
\end{equation}
differentiation in $K$ yields $\partial_{KK}C(K,T)=e^{-rT}p_{S_T}(K)$ in the classical/distributional sense.
Multiplying prices by a positive constant (e.g.\ passing to the undiscounted forward call) does not change
the sign of $\partial_{KK}C$, so the convexity/butterfly criterion is invariant under discounting. Equivalently, for the undiscounted/forward call $\widetilde C=e^{rT}C$ one has $\partial_{KK}\widetilde C(K,T)=p_{S_T}(K)$.

We therefore apply \emph{SSVI} solely as an arbitrage-preserving regularizer of the noisy, sparse implied-volatility
data used for plotting and comparison; the underlying option prices produced by the Heston PDE (classical or quantum)
are not altered beyond this parametric projection. SSVI stands for Surface Stochastic Volatility Inspired. Here we consider only a single maturity slice of the volatility surface.

We parameterize strike by forward log-moneyness $k=\log(K/F_T)$ with forward
$F_T=S_0e^{r\,T}$, and define total implied variance $w(k,T):=\sigma_{\mathrm{BS}}(k,T)^2\,T$.
At fixed maturity $T$, SSVI fits the slice $k\mapsto w(k,T)$ via \cite{Gatheral02012014}
\begin{equation}\label{eq:ssvi_slice_main}
w^{\mathrm{SSVI}}(k;\theta,\rho,\lambda)
=
\frac{\theta}{2}\left(
1+\rho\,\varphi(\theta;\lambda)\,k
+\sqrt{\big(\varphi(\theta;\lambda)\,k+\rho\big)^2+(1-\rho^2)}
\right),
\qquad |\rho|<1,
\end{equation}
with $\theta=w(0,T)$ the ATM total variance, $\rho$ controlling skew that is different from the correlation between Brownian dynamics, and $\varphi$ governing wing slopes.
The SSVI parameter $\rho_{\mathrm{SSVI}}$ or interchangeably for ease of notation $\rho$ here should not be confused with the Heston correlation parameter $\rho$. The former is a fitted smile-shape parameter, whereas the latter is the instantaneous correlation in the Heston dynamics. For notational convenience, we shall overload the meaning $\rho$, and make it explicit with the $\mathrm{SSVI}$ when necessary.

We use the Heston-like choice \cite[Example~4.1]{Gatheral02012014}
\begin{equation}\label{eq:phi_hestonlike_main}
\varphi(\theta;\lambda)
=\frac{1}{\lambda\theta}\left(1-\frac{1-e^{-\lambda\theta}}{\lambda\theta}\right),
\qquad \lambda>0.
\end{equation}
Gatheral--Jacquier provide closed-form \emph{sufficient} constraints (plus a standard tail condition)
that guarantee absence of butterfly arbitrage for the slice and enforce Lee-type wing behavior
\cite{Gatheral02012014,LeeMoment2004}.

We intentionally evaluate only a small number of strikes (three on each side of ATM) to mimic market quoting.
With this level of sparsity, and with stochastic errors in price estimates from quantum routines, the map
``prices $\rightarrow$ implied vols'' is ill-conditioned; direct interpolation can exhibit spurious arbitrage
that is purely numerical. Fig.~\ref{fig:ssvi-panels} illustrates this effect and the resulting implied densities.

\noindent\textbf{Constrained calibration.}
For each maturity we calibrate $(\theta,\rho,\lambda)$ by weighted least squares in total variance
subject to the Gatheral--Jacquier sufficient inequalities. The procedure is summarized in
Algorithm~\ref{alg:ssvi}.

\begin{algorithm}[h]
\caption{Constrained SSVI calibration with Heston-like $\varphi$}
\DontPrintSemicolon
\label{alg:ssvi}
{\color{blue}
\KwIn{Strikes $K_i$, implied vols $\sigma_{\mathrm{BS},i}$ at maturity $T$, spot $S_0$, rate $r$ (classical/quantum)}
}
\KwOut{SSVI slice parameters $(\theta^\ast,\rho_{\mathrm{SSVI}}^\ast,\lambda^\ast)$}
\BlankLine
\tcp{Step 1: Preprocessing (forward, log-moneyness, total variance)}
$F_T \gets S_0\,e^{rT}$\;
$k_i \gets \ln(K_i/F_T)$ for $i=1,\dots,n$\;
$w_i \gets \sigma_{\mathrm{BS},i}^2\,T$ \tcp*{total variance data}
\BlankLine
\tcp{Step 2: Initialize $\theta$ from the ATM total variance}
$i^\ast \gets \arg\min_i |k_i|$ \tcp*{closest-to-ATM}
$\theta_{\mathrm{init}} \gets w_{i^\ast}$ \tcp*{ATM proxy for initialization}
\BlankLine
\tcp{Step 3: Model specification (SSVI with Heston-like $\varphi$)}
Define
\begin{equation}
\varphi(\theta;\lambda)
=
\frac{1}{\lambda\theta}\!\left(1-\frac{1-e^{-\lambda\theta}}{\lambda\theta}\right),
\qquad \lambda>0.
\end{equation}
When $\lambda\theta$ is small, use a Taylor expansion for numerical stability\;
Define the SSVI total variance slice
\begin{equation}
w^{\mathrm{SSVI}}(k;\theta,\rho_{\mathrm{SSVI}},\lambda)
=
\frac{\theta}{2}\!\left(
1+\rho_{\mathrm{SSVI}}\,\varphi(\theta;\lambda)\,k
+\sqrt{\big(\varphi(\theta;\lambda)\,k+\rho_{\mathrm{SSVI}}\big)^2+(1-\rho_{\mathrm{SSVI}}^2)}
\right).
\end{equation}
\BlankLine
\tcp{Step 4: Initialization}
Choose $\theta_{\mathrm{init}}$ from Step 2, $\rho_{\mathrm{SSVI,init}}\in(-1,1)$, and $\lambda_{\mathrm{init}}>0$\;
\BlankLine
\tcp{Step 5: Constrained least-squares fit over $(\theta,\rho_{\mathrm{SSVI}},\lambda)$}
Choose weights $\omega_i>0$ (e.g.\ uniform or vega-based)\;
Solve
\begin{equation}
(\theta^\ast,\rho_{\mathrm{SSVI}}^\ast,\lambda^\ast)
\in
\arg\min_{\theta>0,\;|\rho_{\mathrm{SSVI}}|<1,\;\lambda>0}
\sum_{i=1}^n \omega_i\Big(
w^{\mathrm{SSVI}}(k_i;\theta,\rho_{\mathrm{SSVI}},\lambda)-w_i
\Big)^2
\end{equation}
subject to the sufficient no-butterfly constraints of Gatheral--Jacquier:
\begin{equation}
\theta\,\varphi(\theta;\lambda)\,(1+|\rho_{\mathrm{SSVI}}|)\le 4-\varepsilon,
\qquad
\theta\,\varphi(\theta;\lambda)^2\,(1+|\rho_{\mathrm{SSVI}}|)\le 4.
\end{equation}
\BlankLine
\Return{$(\theta^\ast,\rho_{SSVI}^\ast,\lambda^\ast)$}\;
\end{algorithm}

Once $(\theta^\ast,\rho_{\mathrm{SSVI}}^\ast,\lambda^\ast)$ are obtained, the fitted implied volatility is
$\sigma^{\mathrm{SSVI}}(k)=\sqrt{w^{\mathrm{SSVI}}(k;\theta^\ast,\rho_{\mathrm{SSVI}}^\ast,\lambda^\ast)/T}$.

\begin{figure}[h]
\centering

\begin{minipage}[t]{0.48\textwidth}
    \centering
    \includegraphics[width=\textwidth]{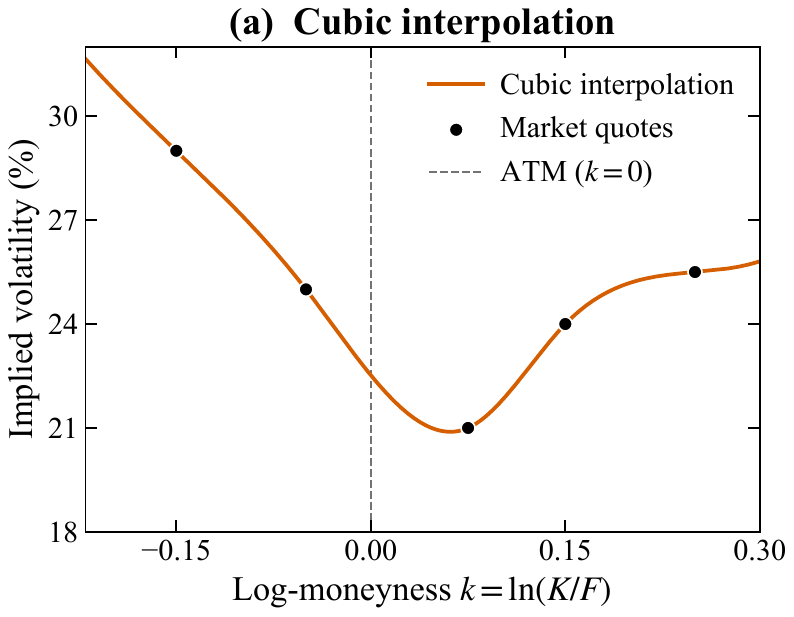}
\end{minipage}
\hfill
\begin{minipage}[t]{0.48\textwidth}
    \centering
    \includegraphics[width=\textwidth]{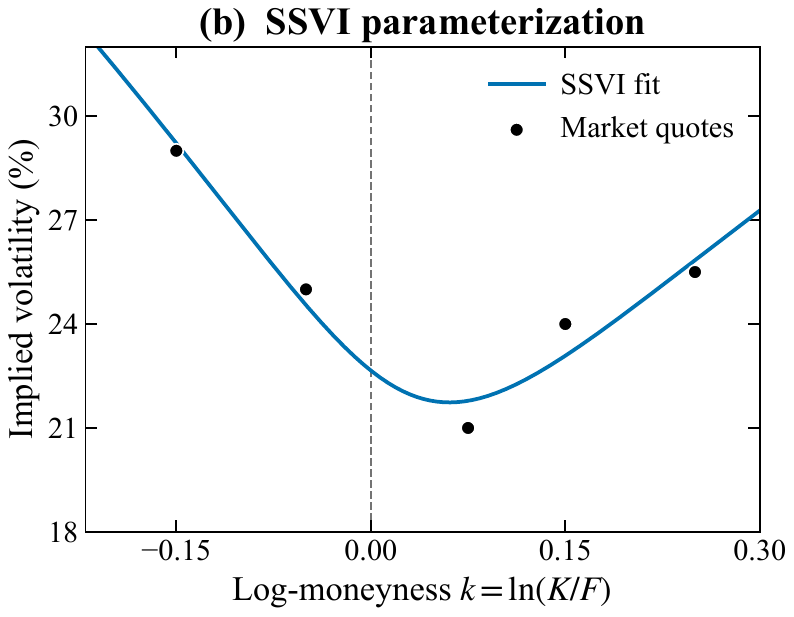}
\end{minipage}

\vspace{0.8em}

\begin{minipage}[t]{0.48\textwidth}
    \centering
    \includegraphics[width=\textwidth]{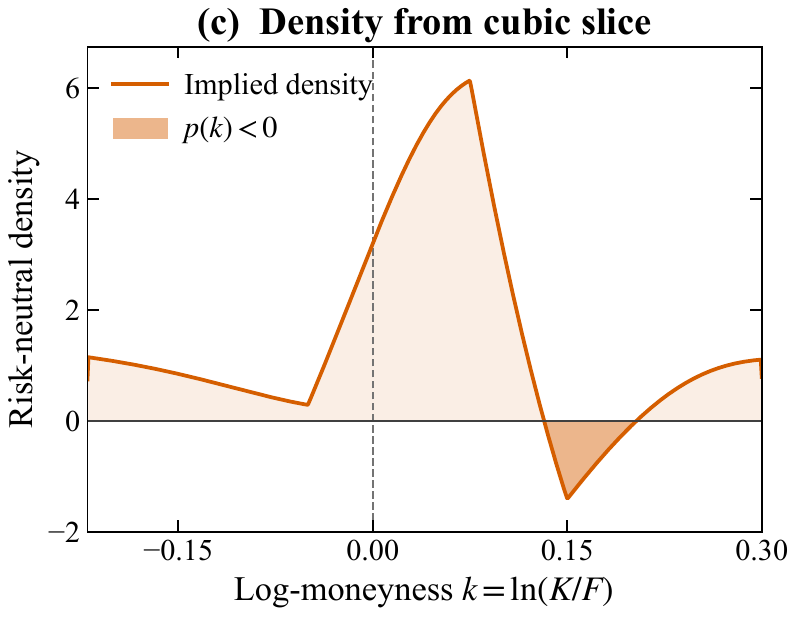}
\end{minipage}
\hfill
\begin{minipage}[t]{0.48\textwidth}
    \centering
    \includegraphics[width=\textwidth]{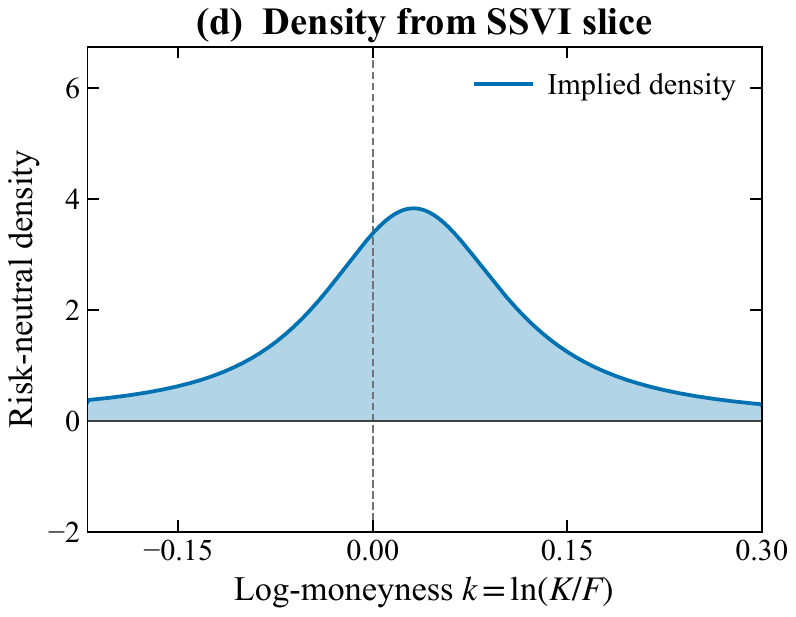}
\end{minipage}
\caption{\textbf{Arbitrage-free regularization of implied volatility slices.}
(a) Sparse quotes (red circles) with cubic interpolation of implied volatility.
(b) SSVI fit to the same quotes.
(c) Implied density from Breeden--Litzenberger applied to the cubic interpolation; negative regions indicate butterfly arbitrage.
(d) Implied density from the SSVI slice, which is non-negative under the enforced constraints.
SSVI provides an arbitrage-free benchmark for comparing classical and quantum Heston-PDE outputs.}
\label{fig:ssvi-panels}
\end{figure}
\medskip
\noindent To diagnose butterfly arbitrage we compute the implied risk-neutral density from each fitted slice.
For a $C^2$ total-variance smile $w(k)$ the density factorizes as a strictly positive prefactor times
Durrleman’s function $g(k)$; the full argument appears in Appendix~\ref{app:durrlemang} and \cite{durrleman2010implied,Gatheral02012014}.

\subsection{Linear ODE system for Heston model}
\label{sec: ODE heston 1d}
The overall discretization pipeline is the same as in
Appendix~\ref{sec:discretization} and
\ref{sec:bs-multiasset-discretization}; in particular, we employ the
standard second-order accurate centered finite-difference schemes
summarized in Table~\ref{table:famous-schemes} and
Eq.~\eqref{eq:fd-mixed-deriv-SiSj}. The terminal payoff provides the
initial condition for the backward-time evolution (or, equivalently, the
terminal condition for the PDE); we refer to
Appendix~\ref{sec:typical payoffs} for typical vanilla choices.

The only additional point that must be fixed to close the semi-discrete
system is the choice of boundary conditions. For each underlying price
coordinate~$S^i$, we impose the same boundary treatment as in
Appendix~\ref{sec:discretization}: a Dirichlet condition at the origin, $V(S^i=0,\,\cdot\,)=0,$ and a Neumann condition at the upper truncation,
\begin{equation}
 \partial_{S^i} V(S^i=S_{\max}^i,\,\cdot\,)=1.
\end{equation}

For each volatility coordinate~$v^i$, we truncate the variance interval
to $[v_{\min},v_{\max}]$, where $v_{\min}>0$ is chosen very close to~$0$
and $v_{\max}$ is taken sufficiently large relative to the region of
interest (typically $v_{\max}\approx 2$--$3$), and impose homogeneous
Neumann conditions at both artificial boundaries,
\begin{equation}
 \partial_{v^i} V(v^i=v_{\min},\,\cdot\,)=0,
\qquad
\partial_{v^i} V(v^i=v_{\max},\,\cdot\,)=0.   
\end{equation}

This is a standard numerical closure for artificial far-field boundaries
in finite-difference treatments of parabolic PDEs; for general
discussions of boundary treatments in financial PDE discretizations, see
\cite{duffy2006fdm, hirsa2024computational}. Since
$v_{\min}$ and $v_{\max}$ lie well outside the region where the option
price carries significant sensitivity to variance, the zero-flux
condition introduces negligible modelling error.

We emphasise that the Feller condition
$2\kappa_i\theta_i\geq\sigma_i^2$ is \emph{not} assumed. This condition
ensures strict positivity of each~$v^i$ at the SDE level but is
typically violated in calibrated models. At the boundary $\{v^i=0\}$,
the inward-pointing drift $\kappa_i\theta_i>0$ renders it a Fichera
boundary requiring no additional boundary data; the pricing PDE remains
well-posed without the Feller assumption. The truncation $v_{\min}>0$
sidesteps the degeneracy of the diffusion coefficient at $v^i=0$
entirely, which is convenient for the Schr\"{o}dingerisation step that
follows.

With these boundary conditions and the spatial discretization specified by Table~\ref{table:famous-schemes}, the Heston PDE reduces to an inhomogeneous linear ODE system of the form \(\dot{u}(t)=A u(t)+b\) (up to the same conventions as in Appendix~\ref{sec:discretization}). Finally, we make the system homogeneous via the standard augmentation, exactly as in Eq.~\eqref{eq: homogeneous BS after augmentation}, obtaining the homogeneous linear ODE
\begin{equation}
\frac{\mathrm{d}}{\mathrm{d}t}\,\tilde{u}(t)=\tilde{A}\,\tilde{u}(t).
\end{equation}
This form allows us to apply the quantum framework presented in Appendix~\ref{section:evoltuion non conservative}.

\begin{figure}[t!]
    \centering
    \begin{minipage}[t]{0.32\linewidth}
        \centering
        \includegraphics[width=\linewidth]{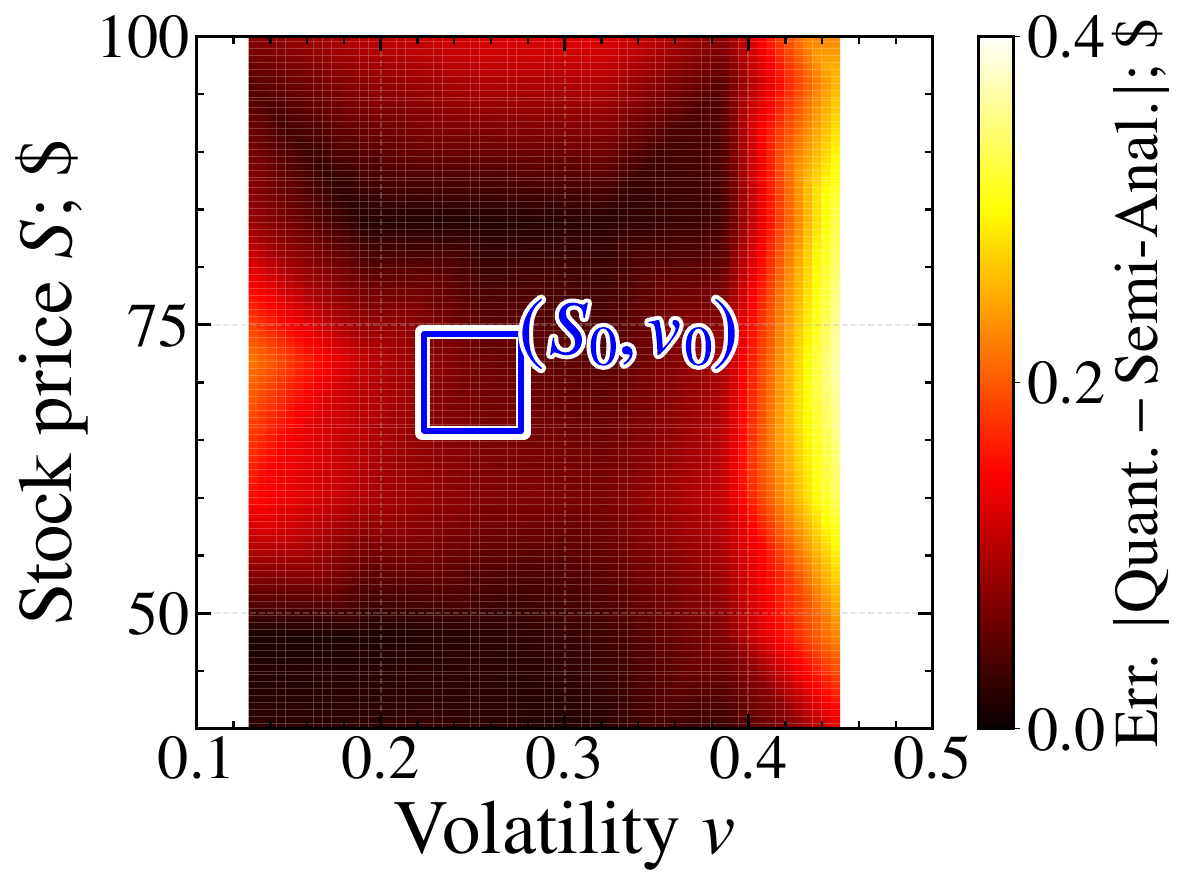}\\[2pt]
        \textbf{(A)}
    \end{minipage}\hfill
    \begin{minipage}[t]{0.32\linewidth}
        \centering
        \includegraphics[width=\linewidth]{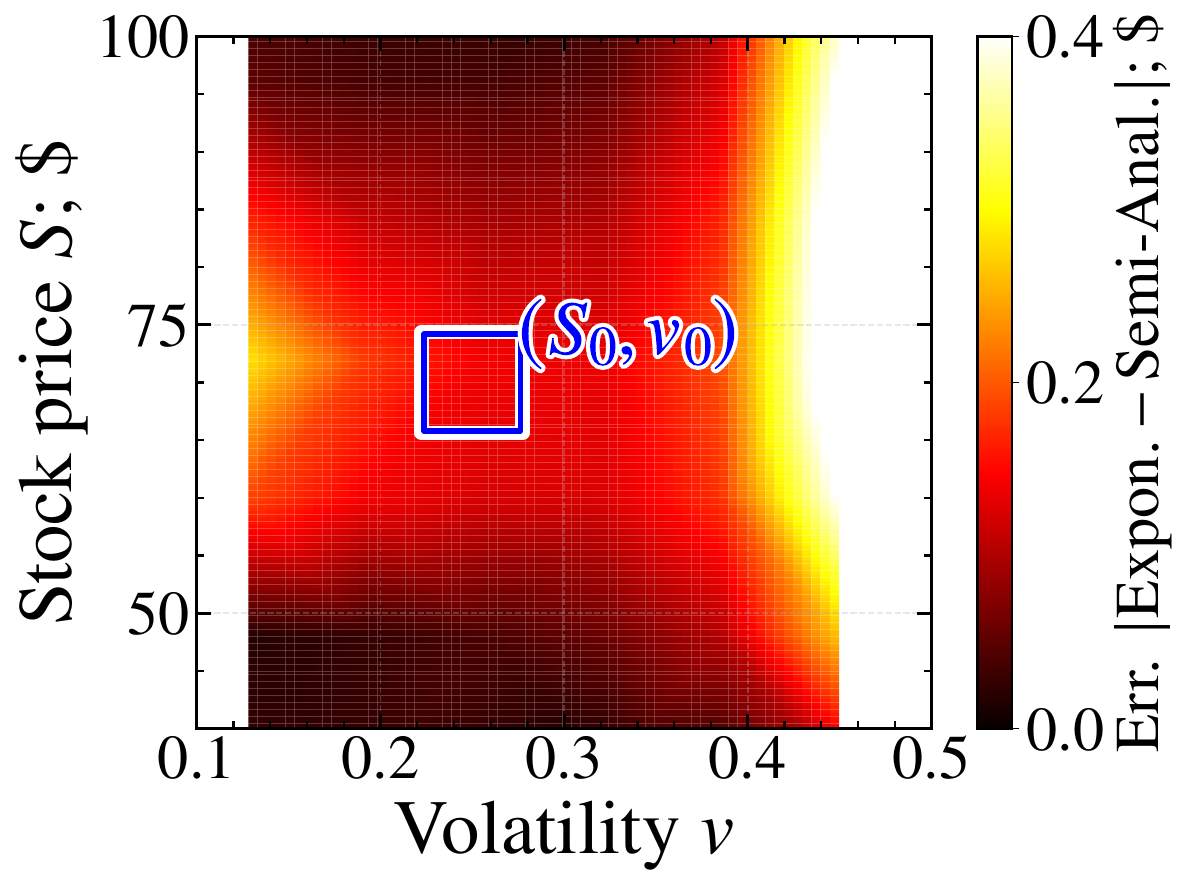}\\[2pt]
        \textbf{(B)}
    \end{minipage}\hfill
    \begin{minipage}[t]{0.32\linewidth}
        \centering
        \includegraphics[width=\linewidth]{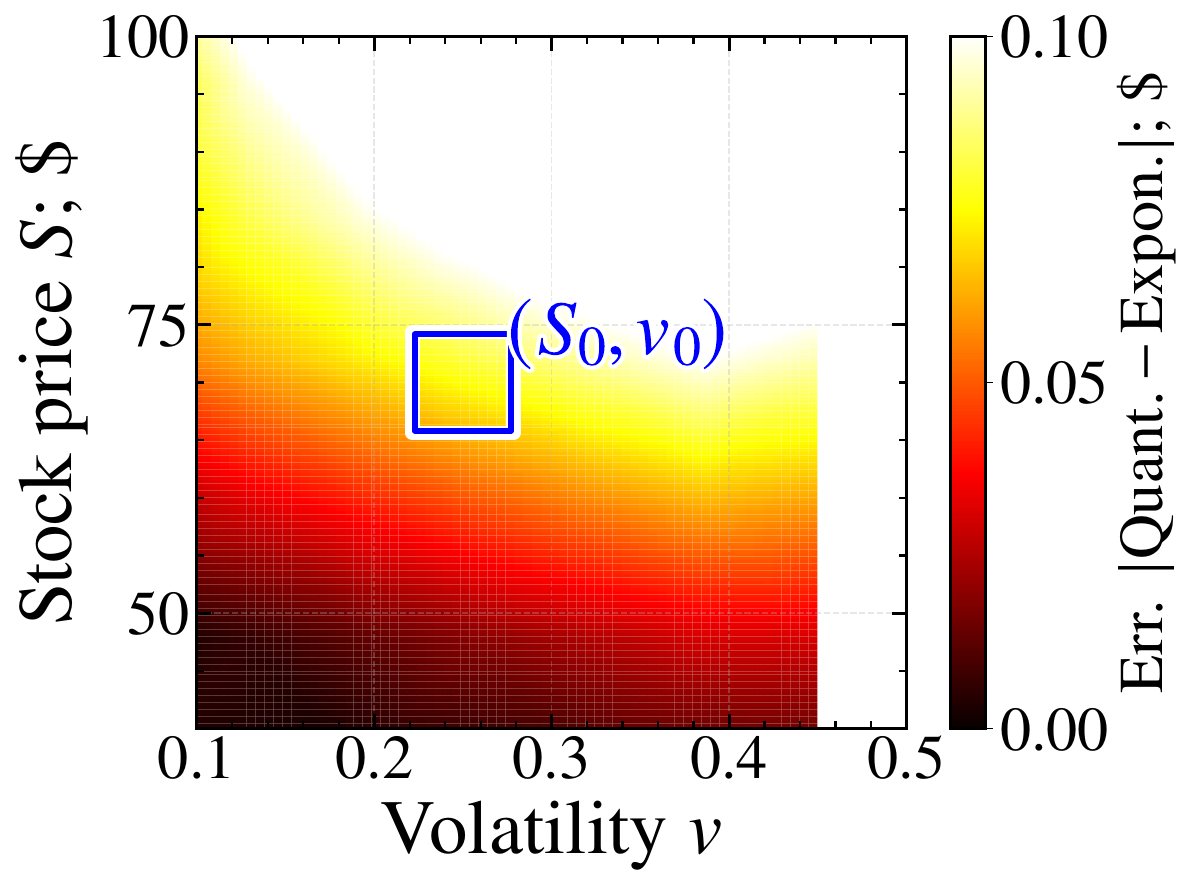}\\[2pt]
        \textbf{(C)}
    \end{minipage}

    \caption{
    \textbf{Absolute-error heat maps for option prices in the one-asset Heston pricing problem \eqref{eq:HestonPDE}.} 
    \textbf{(A)} Pointwise error $|V_{\mathrm{Q}}-V_{\mathrm{SA}}|$ between the quantum framework solution and the semi-analytical Heston benchmark. 
    \textbf{(B)} Similar pointwise error $|V_{\mathrm{EXP}}-V_{\mathrm{SA}}|$ for the matrix-exponential solver. 
    \textbf{(C)} Pointwise difference $|V_{\mathrm{Q}}-V_{\mathrm{EXP}}|$ between the quantum framework and matrix-exponential solutions. 
    In all three panels, the color bar gives the absolute option-price error (in dollars) on the truncated $(v,S)$ grid; the blue box marks the neighborhood of the region of interest $(S_0,v_0)$, the notation and parameters of the simulation are given in Table~\ref{tab:heston-shared-params}.
    }
    \label{fig:heston-pde-error-heatmaps}
\end{figure}

\begin{figure}[t!]
    \centering
    \begin{minipage}[t]{0.49\linewidth}
        \centering
        \includegraphics[width=\linewidth]{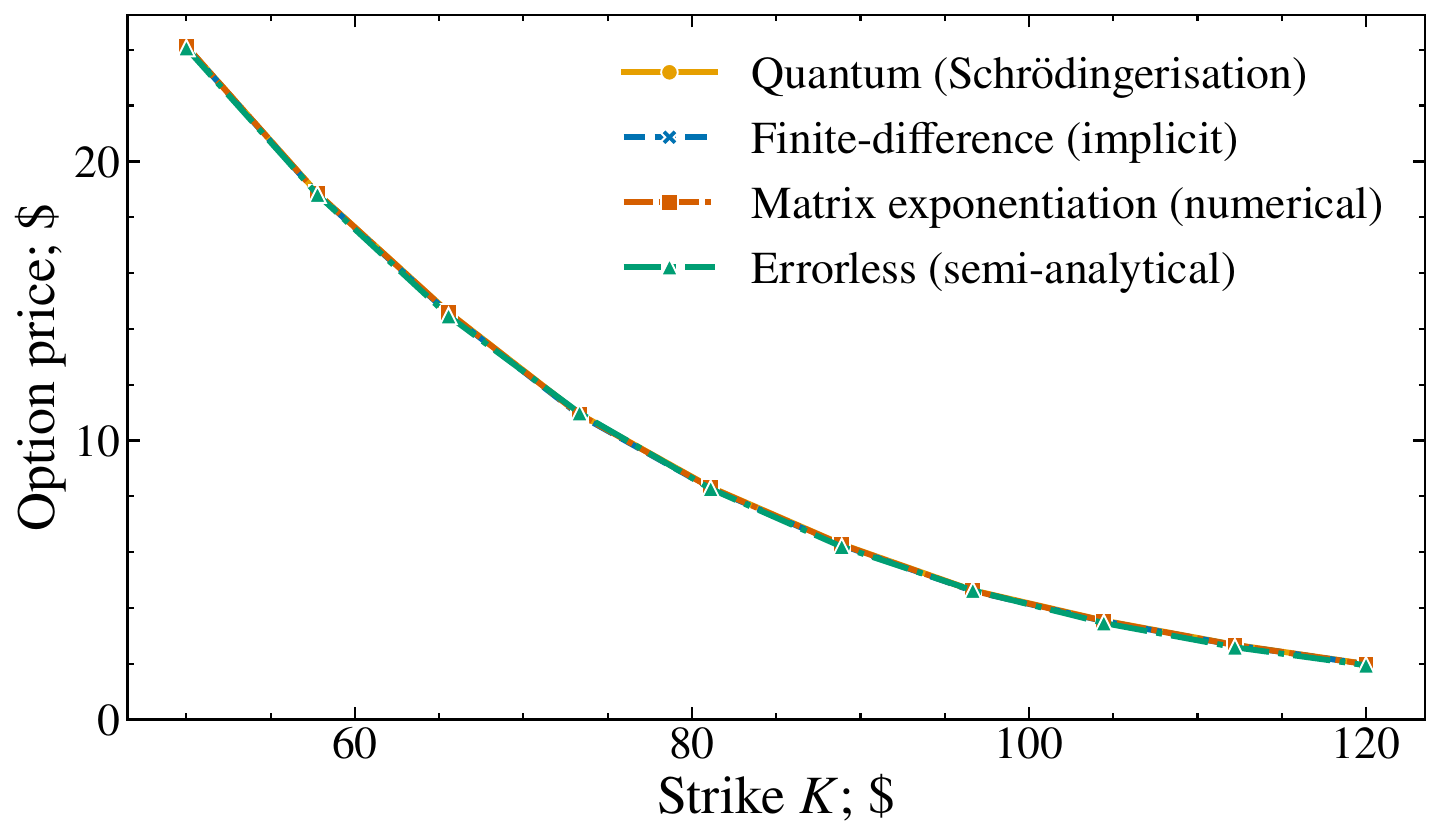}\\[2pt]
        \textbf{(A)}
    \end{minipage}\hfill
    \begin{minipage}[t]{0.49\linewidth}
        \centering
        \includegraphics[width=\linewidth]{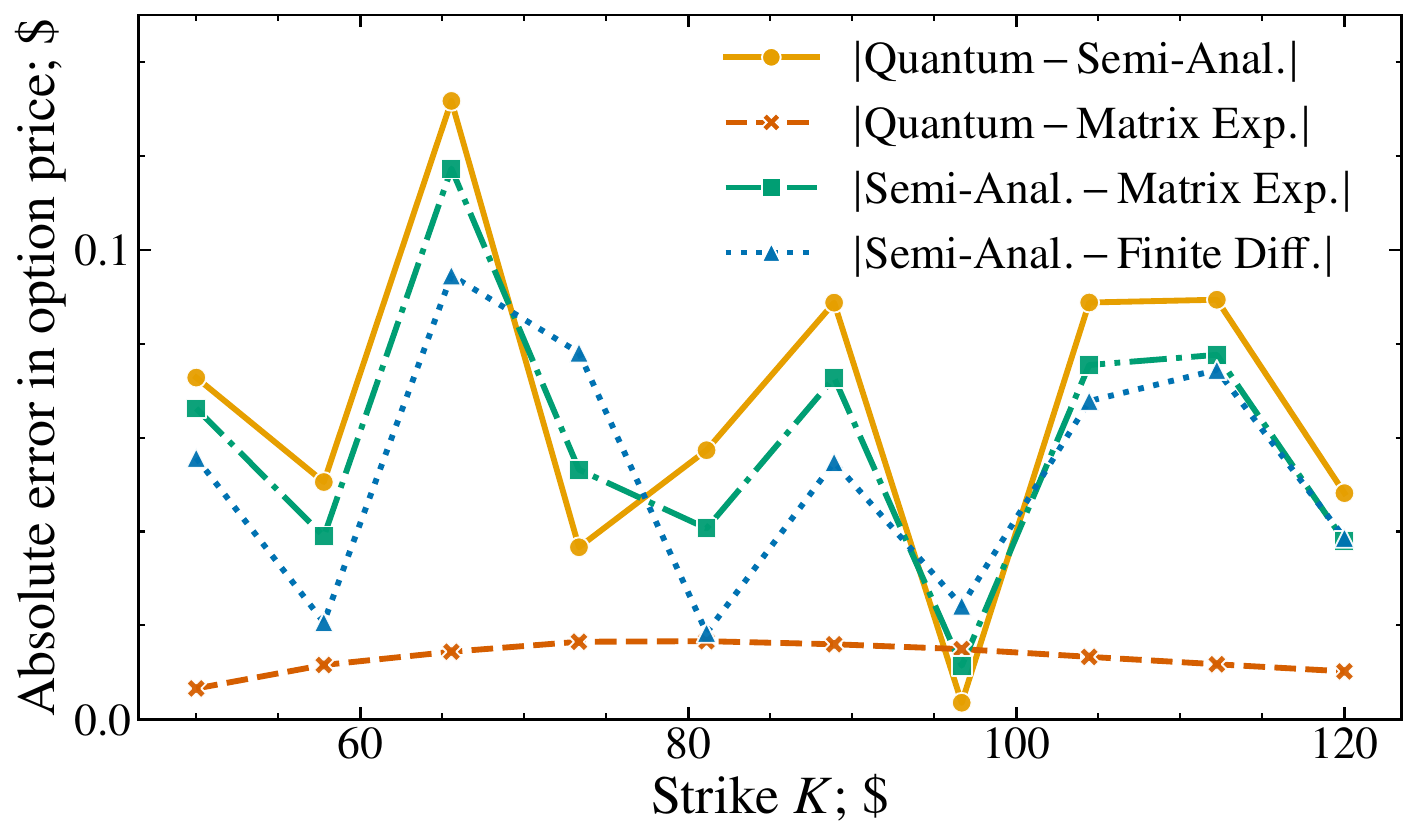}\\[2pt]
        \textbf{(B)}
    \end{minipage}
    \caption{
    \textbf{One-asset Heston option values in dollars at the query state $(S_0,v_0)$, plotted as functions of strike $K$.}
    \textbf{(A)} Prices $V(0,S_0,v_0;K)$ computed by the quantum Schr\"odingerisation workflow based on the semi-discrete ODE of Appendix~\ref{sec: ODE heston 1d} and the evolution method of Appendix~\ref{section:evoltuion non conservative}, together with the classical implicit finite-difference and matrix-exponential methods from Appendix~\ref{appendix:numerical methods}, and the semi-analytical Heston benchmark \cite{heston1993closed}. 
    \textbf{(B)} Absolute pairwise pricing differences among the methods. 
    The notation and simulation parameters are given in Table~\ref{tab:heston-shared-params}.}
    \label{fig:heston-strike-scan}
\end{figure}

\begin{table}[t!]
\centering
\renewcommand{\arraystretch}{1.15}
\begin{tabular}{l l l}
\hline\hline
Variable & Definition & Num.\ value \tabularnewline
\hline
$T$ & Time to maturity & $1$ \tabularnewline
$r$ & Risk-free interest rate & $0.03$ \tabularnewline
$\sigma$ & Volatility of volatility & $0.4$ \tabularnewline
$\rho$ & Spot--variance correlation & $-0.1$ \tabularnewline
$\kappa$ & Mean-reversion speed of variance & $2$ \tabularnewline
$\theta$ & Long-run variance mean & $0.12$ \tabularnewline
$S_{0}$ & Initial stock price & $70$ \tabularnewline
$v_{0}$ & Initial variance & $0.25$ \tabularnewline
$K_{\mathrm{surf}}$ & Strike price for option-price surface scan & $75$ \tabularnewline
$K_{\mathrm{smile}}$ & Strike range for implied-volatility smile scan & $[50,120]$ \tabularnewline
$S_{\max}$ & Maximum stock price (domain truncation) & $180$ \tabularnewline
$S_{\min}$ & Minimum stock price (domain truncation) & $0$ \tabularnewline
$v_{\min}$ & Minimum variance (domain truncation) & $0$ \tabularnewline
$v_{\max}$ & Maximum variance (domain truncation) & $0.45$ \tabularnewline
\hline\hline
\end{tabular}
\caption{\textbf{Model and domain parameters for the one-asset Heston numerical experiments.} Both the option-price surface and implied-volatility-smile programs use the same Heston parameter set. The surface experiment uses the fixed strike $K_{\mathrm{surf}}=75$, while the smile experiment scans strikes over $K_{\mathrm{smile}}\in[50,120]$.}
\label{tab:heston-shared-params}
\end{table}

\subsection{Numerical simulation}\label{subsec: Heston numerics}

This subsection demonstrates the application of the quantum framework developed in Appendices~\ref{section:evoltuion non conservative}--\ref{section:readout} to the one-asset Heston model \eqref{eq:HestonPDE}. The purpose of these numerical tests is to verify that the quantum workflow produces results consistent with classical benchmarks and semi-analytical solutions. Specifically, we compare the quantum results with classical numerical PDE solvers and the Heston characteristic-function-based benchmark \cite{Ball_Roma_1994,heston1993closed,mathworks_financial_instruments_toolbox_r2025a}. The code used for these simulations is available at \cite{guseynov_gitlab_placeholder}.

The first simulation studies the full price surface $V(0,S,v;K_{\mathrm{surf}})$ at fixed strike $K_{\mathrm{surf}}$. For the quantum computation, we discretize the stock-price variable by $N_S=2^4$, the variance variable by $N_v=2^3$, and the Schr\"odingerisation variable by $N_\xi=2^9$, and we include one augmentation qubit as in Eq.~\eqref{eq:bs-homog-augmentation-vector}; altogether this corresponds to an 18-qubit simulation. For the classical exponentiation and finite-difference solvers from Appendix~\ref{appendix:numerical methods}, we use the same $(S,v)$ discretization. The remaining Heston parameters are listed in Table~\ref{tab:heston-shared-params}. The resulting error heat maps are shown in Fig.~\ref{fig:heston-pde-error-heatmaps}. Because classical simulation of the quantum workflow becomes expensive rather quickly, we restrict the computational intervals in both $S$ and $v$ to relatively small segments. This truncation is the main reason for the larger errors near the artificial boundaries, most visibly close to $v_{\max}$. In the region of interest marked in Fig.~\ref{fig:heston-pde-error-heatmaps}, however, the agreement with the semi-analytical benchmark remains relatively good: in both \textbf{(A)} and \textbf{(B)}, the absolute error is approximately $0.1$ dollars. The key message comes from \textbf{(C)}. The dominant error is caused by the coarse discretization and domain truncation, while the additional purely quantum error introduced by Schr\"odingerisation is subleading. In principle, all these errors are caused by coarse discretization and can be reduced further by refining the discretization, although such tests are limited by the cost of classically simulating larger quantum systems.

The second simulation focuses on the single-point price $V(0,S_0,v_0;K)$ as a function of strike. This is the more realistic market setting, since in practice one prices options at the current state $(S_0,v_0)$ and then scans across strikes to build the implied-volatility smile. We use the same discretization as in the surface experiment above, namely $N_S=2^4$, $N_v=2^3$, $N_\xi=2^9$, and one augmentation qubit, for a total of $18$ qubits. The classical exponentiation and finite-difference solvers use the same $(S,v)$ discretization, and the remaining parameters are again listed in Table~\ref{tab:heston-shared-params}. The results are shown in Fig.~\ref{fig:heston-strike-scan}. Panel~\textbf{(A)} shows that all four methods are well aligned across the full strike range. The error analysis in panel~\textbf{(B)} shows that all numerical methods have the same main source of error, namely the coarse discretization, while the small residual difference between the exponentiation and quantum results is the additional error introduced by Schr\"odingerisation.

Starting from the strike scan in Fig.~\ref{fig:heston-strike-scan}, we invert each price curve to Black--Scholes implied volatilities and then regularize the resulting sparse slice by the arbitrage-free SSVI parametrization \cite{Gatheral02012014} introduced in Appendix~\ref{subsec:ssvi-main} and Appendix~\ref{app:ssvi}. This produces Fig.~\ref{fig:ssvi-three-panels}. The main message is panel \textbf{(Left)}: after SSVI fitting, the quantum, finite-difference, matrix-exponential, and semi-analytical Heston prices all generate the same non-flat smile/skew. In contrast to the flat Black--Scholes picture in Fig.~\ref{fig:volsmileskew}, the obtained Heston dynamics clearly reproduce a genuine smile structure. Panel \textbf{(Centre)} quantifies the fit discrepancy through
\begin{equation}\label{eq:heston-ssvi-relative-error}
\varepsilon_{\mathrm{rel}}^{(X,Y)}(K)
:=
\frac{\left|\sigma_{\mathrm{SSVI}}^{(X)}(K)-\sigma_{\mathrm{SSVI}}^{(Y)}(K)\right|}
{\sigma_{\mathrm{SSVI}}^{(\mathrm{SA})}(K)}\times 100\%,
\end{equation}
where $X,Y\in\{\mathrm{Q},\mathrm{EXP},\mathrm{SA}\}$. The relative errors remain below $1\%$, showing that the quantum workflow captures the market-relevant smile with high fidelity. Panel \textbf{(Right)} isolates the role of the SSVI skew parameter and shows that the fitted value $\rho_{\mathrm{SSVI}}\approx -0.12$ from Algorithm~\ref{alg:ssvi} reproduces the best skew curve for the data points. These residual discrepancies are primarily discretization effects and can be further suppressed by increasing the spatial and auxiliary resolutions, i.e.\ by adding qubits; this is comparatively natural for the target quantum device, but becomes rapidly expensive in classical full-wavefunction simulations.

\subsection{End-to-end complexity scaling}\label{subsec: Heston 1d caling}

We now collect the dominant gate complexities of the \emph{one-asset Heston} pipeline for estimating a single option value $V(0,S_{0},v_{0})$. The construction of the initial state and the smooth cut-off for Schr\"odingerisation are given in Appendix~\ref{section:state prep}. The auxiliary register size $n_{\xi}$ and the resulting non-conservative evolution are discussed in Appendix~\ref{section:evoltuion non conservative}, with the simulation cost governed by Proposition~\ref{theorem:QuantumHamiltonianSimulation}. The overhead for extracting a \emph{single} grid value from the final state is summarized in Appendix~\ref{section:readout}. Since state preparation and time evolution are performed once per run, while the pointwise readout dominates the sampling cost, the end-to-end complexity is
\begin{equation}
\label{eq:heston-1d-end-to-end-scaling}
\Bigl(
\underbrace{n\log n + n_{\xi}\log n_{\xi}}_{\text{state preparation}}
+
\underbrace{2^{2n}T\bigl(n\log n + n_{\xi}\log n_{\xi}\bigr)}_{\text{Hamiltonian simulation}}
\Bigr)
\cdot
\underbrace{2^{n}}_{\text{single-point readout}}.
\end{equation}
Here $n$ is the number of qubits per spatial axis (we take the same resolution in $S$ and $v$ for simplicity), $n_{\xi}$ is the number of qubits in the Schr\"odingerisation auxiliary register, and $T$ is the evolution time. The factor $2^{2n}$ arises from the scaling of the discretized generator (finite differences for second derivatives), i.e.\ $\|H\|_{\max}\sim 2^{2n}$.

For comparison, standard classical solvers for the same Heston PDE discretization are summarized in Appendix~\ref{appendix:numerical methods}; at the scaling level they typically exhibit leading dependence of order $\mathcal{O}(2^{4n})$ (up to tolerance- and iteration-dependent factors), while \eqref{eq:heston-1d-end-to-end-scaling} yields a leading dependence of order $\mathcal{O}(2^{3n})$ after suppressing $poly(n)$ contributions. As emphasized in Remark~\ref{rem:universal-disclaimer-ft-topology}, this comparison is intended only as a scaling-level benchmark.

\section{Multi-asset Heston PDE}
The multi-asset Heston model has received limited attention in PDE form. The more common approach is to work with the characteristic function or backward stochastic differential equation (BSDE) representation directly, bypassing the explicit PDE \cite{GnoattoDeepQuadraticHedge, GrasselliFXsmiles2013}. Our approach is complementary to these works and we work with the pricing PDE explicitly. To the best of our knowledge, working with this form has not been previously tackled.

Let $S(t)=(S^1(t),\dots,S^d(t))\in\mathbb R_+^d$ denote $d$ asset prices and
$v(t)=(v^1(t),\dots,v^d(t))\in\mathbb R_+^d$ their (idiosyncratic) variance factors.
Throughout we work under a fixed risk-neutral measure $\mathbb Q$ with constant risk-free rate $r$
and no dividends. The postulated dynamics are
\begin{align}
dS^i(t) &= r\,S^i(t)\,dt + S^i(t)\sqrt{v^i(t)}\,dW^{S,i}(t),\qquad i=1,\dots,d, \label{eq:multiHeston_S}\\
dv^i(t) &= \kappa_i(\theta_i-v^i(t))\,dt + \sigma_i\sqrt{v^i(t)}\,dW^{v,i}(t),\qquad i=1,\dots,d. \label{eq:multiHeston_v}
\end{align}
Collect the Brownian drivers into the $2d$-vector
\begin{equation}
    W(t)=\big(W^{S,1}(t),\dots,W^{S,d}(t),\;W^{v,1}(t),\dots,W^{v,d}(t)\big),
\end{equation}
with instantaneous correlation matrix $\Gamma\in\mathbb R^{2d\times 2d}$, i.e.
\begin{equation}
d\langle W^\alpha, W^\beta\rangle_t = \Gamma_{\alpha\beta}\,dt.
\label{eq:Gamma_cov}
\end{equation}
We view $\Gamma$ in block form
\begin{equation}
\Gamma=
\begin{pmatrix}
\Gamma^{SS} & \Gamma^{Sv}\\
(\Gamma^{Sv})^\top & \Gamma^{vv}
\end{pmatrix},
\qquad
\Gamma^{SS}_{ij}=\rho_{S_i,S_j},\;\;
\Gamma^{Sv}_{ij}=\rho_{S_i,v_j},\;\;
\Gamma^{vv}_{ij}=\rho_{v_i,v_j}.
\end{equation}

It is convenient to write \eqref{eq:multiHeston_S}--\eqref{eq:multiHeston_v} in vector form.
Define the $2d$-dimensional state $X(t)=(S(t),v(t))$ and the drift
\begin{equation}
b(S,v)=
\begin{pmatrix}
r\,S\\[1pt]
\kappa\odot(\theta-v)
\end{pmatrix},
\qquad
\kappa=(\kappa_1,\dots,\kappa_d),\;\theta=(\theta_1,\dots,\theta_d),
\end{equation}
where $\odot$ denotes component-wise multiplication.

\noindent
Define the $d\times d$ diagonal matrices $D_S(S,v):=\mathrm{diag}\!\big(S^1\sqrt{v^1},\dots,S^d\sqrt{v^d}\big),\qquad
D_v(v):=\mathrm{diag}\!\big(\sigma_1\sqrt{v^1},\dots,\sigma_d\sqrt{v^d}\big)$. Then the larger block diagonal matrix is set to
\begin{equation}
G(S,v):=
\begin{pmatrix}
D_S(S,v) & 0\\
0 & D_v(v)
\end{pmatrix}\in\mathbb R^{2d\times 2d}, 
\end{equation}

where $0$ denotes a $d\times d$ zero matrix.
Now the instantaneous covariance of $X$ is
\begin{equation}
a(S,v)=G(S,v)\,\Gamma\,G(S,v)^\top \in \mathbb R^{2d\times 2d}.
\label{eq:a_matrix}
\end{equation}

Let $f:\mathbb R_+^d\to\mathbb R$ be a European payoff depending only on terminal prices $S(T)$.
Risk-neutral valuation sets
\begin{equation}
  V(t,s,v)=\mathbb E^{\mathbb Q}\!\left[e^{-r(T-t)}\,f(S(T))\mid S(t)=s,\ v(t)=v\right].
\end{equation}

By the Feynman--Kac theorem (equivalently: the discounted process $e^{-rt}V(t,S(t),v(t))$ is a
$\mathbb Q$-martingale and It\^o's formula identifies its drift), $V$ satisfies the backward equation \cite{chan2024financial}
\begin{equation}
\partial_t V + \mathcal L^{\mathbb Q}V - rV = 0,
\qquad
V(T,s,v)=f(s),
\label{eq:multiHeston_generator_PDE}
\end{equation}
where the generator acts on smooth test functions $\varphi=\varphi(S,v)$ as
\begin{equation}
(\mathcal L^{\mathbb Q}\varphi)(S,v)
=
b(S,v)\cdot \nabla \varphi(S,v)
+\frac12\,\mathrm{Tr}\!\Big(a(S,v)\,\nabla^2\varphi(S,v)\Big)
\label{eq:generator_compact}
\end{equation}
or
\begin{align}
0 \;=\;& \frac{\partial V}{\partial t}
\;+\; \sum_{i=1}^{d} rS^i \frac{\partial V}{\partial S^i}
\;+\; \sum_{i=1}^{d} \kappa_i(\theta_i - v^i)\frac{\partial V}{\partial v^i}
\;-\; rV
\nonumber\\
&\;+\; \frac{1}{2}\sum_{i,j=1}^{d}\rho_{S_i,S_j}\,S^i S^j\sqrt{v^i v^j}\,
\frac{\partial^2 V}{\partial S^i \partial S^j}
\;+\; \sum_{i,j=1}^{d}\rho_{S_i,v_j}\,\sigma_j\,S^i\sqrt{v^i v^j}\,
\frac{\partial^2 V}{\partial S^i \partial v^j}
\;\nonumber \\ &+ \frac{1}{2}\sum_{i,j=1}^{d}\rho_{v_i,v_j}\,\sigma_i\sigma_j\sqrt{v^i v^j}\,
\frac{\partial^2 V}{\partial v^i \partial v^j}.
\label{eq:multiHeston_fullPDE}
\end{align}

Expanding \eqref{eq:generator_compact} reproduces the familiar componentwise multi-asset Heston PDE,
with the stock--stock, stock--variance, and variance--variance second-order couplings encoded
compactly by the blocks of $\Gamma$ in \eqref{eq:Gamma_cov} and the covariance map \eqref{eq:a_matrix}.

\textbf{Why $\Gamma$ must be positive semidefinite?}
For any $c\in\mathbb R^{2d}$, the scalar increment $c^\top(W_{t+\Delta t}-W_t)$ has variance
$\mathrm{Var}(c^\top\Delta W)=c^\top\Gamma c\,\Delta t$, which must be nonnegative.
Hence $c^\top\Gamma c\ge 0$ for all $c$, i.e.\ $\Gamma\succeq 0$.
Equivalently, $\Gamma\succeq 0$ guarantees the existence of a factorization $\Gamma=LL^\top$
(e.g.\ Cholesky), allowing one to represent correlated drivers as $dW=L\,dB$ for a standard Brownian
motion $B$.

\subsection{Linear ODE system for Multi-asset Heston model}
\label{sec: ODE heston multi}

Starting from the continuous multi-asset Heston PDE in Eq.~\eqref{eq:multiHeston_fullPDE}, we perform the same method-of-lines reduction as in Appendix~\ref{sec:bs-multiasset-discretization}. Consider $d$ underlying assets with prices $\mathbf{S}=(S_{1},\ldots,S_{d})$ and variance variables $\mathbf{v}=(v_{1},\ldots,v_{d})$, and let $V(t,\mathbf{S},\mathbf{v})$ denote the option value. We discretize each coordinate on a Cartesian grid, with mesh sizes $\Delta S_i$ and $\Delta v_i$, and collect the interior grid values into a vector $\vec V(t)\in\mathbb{R}^{N^{2d}}$ using lexicographic ordering. 

Pure first- and second-order derivatives are approximated by the centered second-order formulas in Eqs.~\eqref{eq:fd-first-deriv-Si} and \eqref{eq:fd-second-deriv-SiSi}, while all mixed derivatives are discretized by the four-point centered stencil in Eq.~\eqref{eq:fd-mixed-deriv-SiSj}. This applies to spot--spot, spot--variance, and variance--variance couplings alike. We emphasize that this discretization is not restricted to these particular second-order finite-difference schemes; higher-order alternatives may also be employed, depending on the desired accuracy and stability requirements \cite{langtangen2017finite}.

The terminal condition is given by the payoff at maturity; for example, for the Worst-of Call payoff from Table~\ref{tab:vanilla payoffs},
\begin{equation}
\label{eq:multi-heston-terminal}
V(T,\mathbf{S},\mathbf{v})
=
\max\!\left(\min_{1\leq i\leq d} S_i - K,\,0\right).
\end{equation}

As in Appendix~\ref{sec:bs-multiasset-discretization}, the main issue for applying the quantum framework is the boundary treatment. For general multi-asset payoffs, the mathematically natural boundary conditions on faces such as $S_i=0$ may require solving reduced lower-dimensional PDEs on faces, edges, and corners. Implementing such a hierarchy of auxiliary boundary PDEs lies outside the boundary-condition model supported by the quantum framework of Appendix~\ref{section:evoltuion non conservative}. One may instead impose artificial boundary conditions that fit the allowed constant-coefficient Dirichlet/Neumann/Robin class, but this introduces an additional modelling and discretization error that we do not analyze in this paper. We therefore focus on payoffs, such as the Worst-of Call, for which simple compatible conditions are available: on each face $S_i=0$ we impose Dirichlet data $V=0$, on each truncation face $S_i=S_{i,\max}$ we impose asymptotic Neumann data $\partial V/\partial S_i=0$, and for each variance coordinate we impose homogeneous Neumann conditions
\begin{equation}
\label{eq:multi-heston-v-bc}
\frac{\partial V}{\partial v_i}\Big|_{v_i=v_{i,\min}}=0,
\qquad
\frac{\partial V}{\partial v_i}\Big|_{v_i=v_{i,\max}}=0,
\qquad i=1,\ldots,d.
\end{equation}

After standard boundary elimination, the semi-discrete system takes the affine form
\begin{equation}
\label{eq:multi-heston-affine}
\frac{d\vec V(t)}{dt}
=
A\vec V(t)+\vec b,
\qquad
\vec V(T)=\vec V_T,
\end{equation}
where $\vec V_T$ is the sampled terminal payoff and $\vec b$ collects the constant contributions generated by the boundary closures. Finally, exactly as in Eq.~\eqref{eq: homogeneous BS after augmentation}, we convert Eq.~\eqref{eq:multi-heston-affine} into a homogeneous linear system,
\begin{equation}
\label{eq:multi-heston-homogeneous}
\frac{d\vec w(t)}{dt}
=
S\vec w(t),
\end{equation}
which is the linear ODE form required by the quantum evolution procedure of Appendix~\ref{section:evoltuion non conservative}.

\subsection{End-to-end complexity scaling}\label{subsec: Heston multi}

We summarize the dominant end-to-end cost for the multi-asset Heston model \eqref{eq:multiHeston_fullPDE}. As in Appendix~\ref{subsec:multiasset BS scaling}, the total cost is the sum of state preparation and backward evolution, multiplied by the single-point readout overhead. For this estimate, we use the Worst-of Call payoff in Eq.~\eqref{eq:multi-heston-terminal}. The main difference from the multi-asset Black--Scholes case is the replacement $d\to 2d$ in the evolution and readout costs, since the Heston. The payoff does not depend on the variance variables, so the payoff-state preparation cost remains the same and uses the amplified routine from Table~\ref{table:quantum state prep of payoffs} and Appendix~\ref{app:amplification}.

Suppressing lower-order factors, including polylogarithmic in $1/\epsilon$, and writing $N=2^{n}$, the dominant cost for extracting one value $V(0,\vec S_{0},\vec v_{0})$ is
\begin{equation}
\label{eq:heston-multi-end-to-end-scaling}
\Bigl(
\underbrace{d^{3}n+d^{2}n\log n}_{\text{state prep.}}
+
\underbrace{d\,T\,2^{2n}\Bigl[d\,n\log n+n_{\xi}\log n_{\xi}+d^{4} n\Bigr]}_{\text{evolution}}
\Bigr)
\cdot
\underbrace{2^{nd}}_{\text{readout}}.
\end{equation}
The evolution term follows from Appendix~\ref{section:evoltuion non conservative} and Proposition~\ref{theorem:QuantumHamiltonianSimulation} with $d\to 2d$, while the readout factor $2^{nd}$ follows from Appendix~\ref{subsec:readout_complexity}, since the Heston state occupies $W=2dn$ system qubits.

For comparison, the classical baselines of Appendix~\ref{appendix:numerical methods} applied to the $2d$-dimensional Heston discretization scale as
\begin{equation}
\label{eq:heston-multi-classical-scaling}
\mathcal{C}_{\mathrm{classical}}
=
\mathcal{O}\!\left(d^{2}2^{n(2d+2)}\right).
\end{equation}
Therefore, up to suppressed $poly(d,n)$ factors, the quantum pipeline has leading exponential dependence $\mathcal{O}\!\left(2^{n(d+2)}\right)$, compared with $\mathcal{O}\!\left(2^{n(2d+2)}\right)$ classically. Thus, we conclude that the quantum workflow is a promising approach for multidimensional option-pricing models: (i) even in the one-dimensional case $d=1$, it already suggests a polynomial advantage in the grid size $2^{n}$, and (ii) as the number of assets in the model \eqref{eq:multiHeston_fullPDE} increases, the relative advantage of the quantum framework becomes more pronounced. We note that this comparison should be interpreted only at the logical-scaling level; see Remark~\ref{rem:universal-disclaimer-ft-topology}.

\FloatBarrier
\section{Evolution of non-conservative ODE system using quantum computing}\label{section:evoltuion non conservative}

The quantum simulation evolves the option price backward in time, $\vec V(T)\rightarrow \vec V(0)$.  Here we consider only the evolution stage (see Fig.~\ref{fig:quantum_scheme_without_postselection}), deferring the state preparation and the readout for the next sections. After spatial discretization of PDEs considered in this paper, we always have a form
\begin{align}
\label{eq:MAINhomog-ode}
\frac{d\,\vec w}{dt}=S\,\vec w,
\qquad
\vec w(T)=\vec w_0,
\end{align}
where the matrix $S$ has sparsity $\mathfrak{s}$ and size $2^{W+1}\times2^{W+1}$, where $W$ is the number of system qubits and the additional qubit is used for the homogeneous augmentation.
 The first problem is that, as is well known, quantum computers implement unitary dynamics, meaning
the norm of the wave function is preserved. In its current form \eqref{eq:MAINhomog-ode} does not correspond to unitary dynamics
unless the matrix $S$ is anti-Hermitian, i.e., $S=iH$ with $H=H^\dagger$.

One of the most promising ways to overcome this problem is
to use Schr\"odingerisation \cite{PRL2024,PRA2023,analog,cao2023quantum,PRS2024,jin2024quantumsimulationfokkerplanckequation,JIN2025114138,jin2025quantumpreconditioningmethodlinear,doi:10.1137/23M1563451,JinLiuMa2024MaxwellSchrodingerisation,jin2025schrodingerizationbasedquantumalgorithms}. The idea is to enlarge
the dimension of the problem (Hilbert space) so that part
of the system can follow the dynamics of the ODE system
\eqref{eq:MAINhomog-ode}. The schematic of Schr\"odingerisation is as follows.
We first define
\begin{equation}
\label{eq:S1,S2}
S_1=\tfrac{1}{2}(S+S^{\dagger}),\qquad
S_2=\tfrac{1}{2i}(S-S^{\dagger}).
\end{equation}
Using these, we rewrite the ODE system as
\begin{equation}
    \frac{d\,\vec w}{dt}=S_1\vec w+iS_2\vec w.
    \label{eq:ode-herm-antiherm}
\end{equation}

The problematic part in Eq.~\eqref{eq:ode-herm-antiherm} is the absence
of the imaginary unit $i$ before $S_1$. To address this,
we introduce an auxiliary continuous variable $\xi$ and, using
the new variable $\vec \phi=\vec w e^{-\xi}$, we obtain
\begin{equation}
\begin{gathered}
    \frac{d\,\vec \phi}{dt}=S_1\vec w e^{-\xi}+iS_2\vec w e^{-\xi}=\\
    =-\frac{\partial}{\partial\xi}\!\big(S_1\vec w  e^{-\xi}\big)+iS_2\vec w e^{-\xi}
    =-\frac{\partial}{\partial\xi}\!\big(S_1\vec \phi\big)+iS_2\vec \phi,
\end{gathered}
\label{eq:schrodingerisation-transform}
\end{equation}
where the chosen $\xi$-dependence yields a derivative operator. We
consider the region $\xi>0$ as the aforementioned subsystem. The
retrieval process for Eq.~\eqref{eq:schrodingerisation-transform} is chosen as
\begin{equation}
    \vec{w}(T)=\int_{\lambda^{\min}_{S_1} T}^{+\infty}\vec\phi(T,\xi)\,d\xi,
    \label{eq:postselection Schrodingerisation}
\end{equation}
where $\lambda^{\min}_{S_1}$ is the absolute value of the smallest
negative eigenvalue of $S_1$; if there are no negative
eigenvalues, we use the default value $0$. This particular choice
of retrieval is explained in \cite{PRA2023}. Briefly, the form of
the differential system \eqref{eq:schrodingerisation-transform} in the $\xi$
dimension has an advection form; by considering the system in
the eigenbasis of $S$, one obtains
\begin{equation}
    \frac{d\,\vec{\tilde{\phi}}}{dt}=-\frac{\partial}{\partial\xi}D^{(S)}_1\vec{\tilde{\phi}}+iD^{(S)}_2\vec{\tilde{\phi}},
    \label{eq:advection-form}
\end{equation}
where $D^{(S)}_i$ are diagonal matrices. The analytical solution
of the advection form is a shift of the initial state
by $-\lambda_{S_1}T$. Consequently, the ``junk'' part from the region
$\xi<0$ can penetrate the target region $\xi>0$ by at most
$\lambda^{\min}_{S_1}T$ during the evolution.

\begin{remark}
    One can notice that the required postselection \eqref{eq:postselection Schrodingerisation} requires knowledge of $\lambda_{S_1}^{\min}$ which generally requires $\mathcal{O}(N_s^3)$ operations. However, we stress that the method still works if one uses a value greater than $\lambda_{S_1}^{\min}$. Thus, we only need to know the upper bound of $\lambda_{S_1}^{\min}$ which can be loose. 
\end{remark}

Another important aspect is choosing the initial condition
$\Phi(t=T,\xi)=\Phi_0(\xi)$ for the introduced dimension $\xi$. As
mentioned above, for the region $\xi>0$ we use $e^{-\xi}$; however,
we still have the opportunity to select the initial data in
$\xi<0$ arbitrarily. Recent advances \cite{schrodingerisation_optimal_queries}
show that one effective choice is a cut-off initial profile,
as illustrated in Fig.~\ref{fig:cutoff_theta}. Further, we discretize
the $\xi$ dimension similarly to the main stock-price dimension and
apply the Fourier transform. To achieve the optimal error
dependence $\mathcal{O}(\log(1/\epsilon_{\text{Schr}}))$, the initial function
$\Phi(t=T,\xi)$ should belong to the class $C^{\infty}(\mathbb{R})$,
meaning derivatives of every order exist and are continuous.

\begin{figure}[h!]
    \centering
    \includegraphics[width=0.43\textwidth]{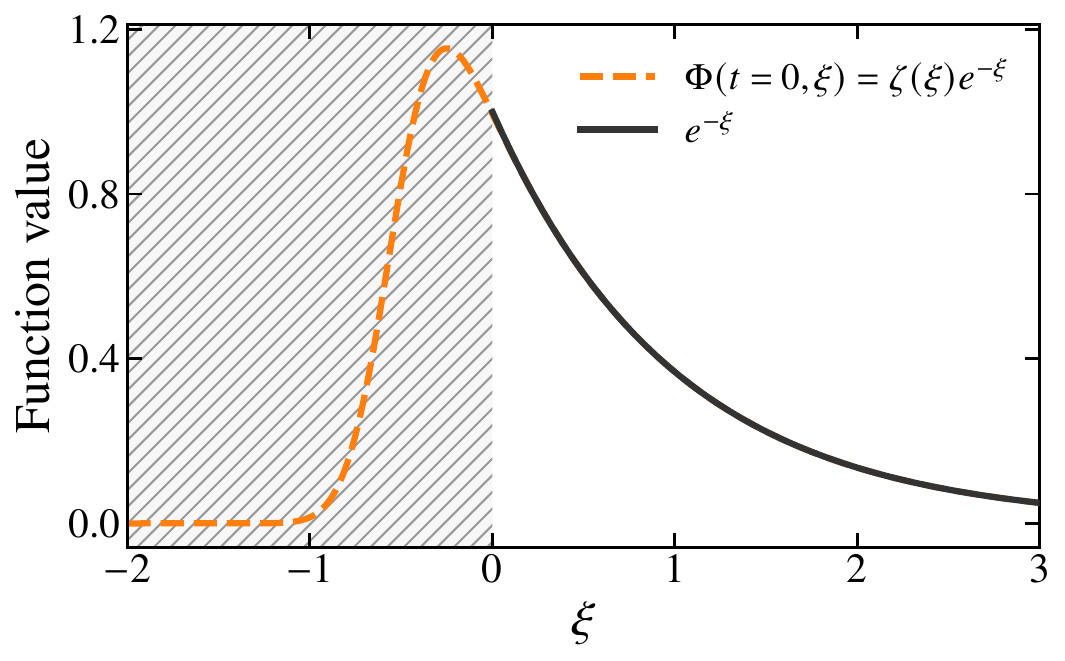}
    \caption{\textbf{Initial state $\Phi(t=T,\xi)$ as a smooth cut-off function.}
    The smooth initialization ensures optimal query complexity
    $\mathcal{O}(\log(1/\epsilon_{\text{Schr}}))$ for Schr\"odingerisation-based quantum
    simulation. The introduced smooth cut-off function coincides with
    $e^{-\xi}$ in the target region $\xi>0$.}
    \label{fig:cutoff_theta}
\end{figure}

To construct a smooth initialization for the Schr\"odingerisation register, we
take an exponentially decaying target and mask it with a $C^\infty$ cut-off.
Specifically, we set
\begin{equation}
\Phi(t=T,\xi)=\Phi_0(\xi):=\zeta(\xi)\,e^{-\xi},
\label{eq:initial schrodingerisation}
\end{equation}
where the window $\zeta$ is defined using the error function:
\begin{equation}
\zeta(\kappa)=\frac{C_\zeta}{2}+\frac{C_\zeta}{\pi}\int_0^{g(\kappa)}e^{-\vartheta^2}d\vartheta;\qquad g(\kappa)=(2\kappa+1)\ln^{1/2}\frac{1}{\epsilon_\text{Schr}}
\end{equation}
with $C_\zeta$ fixing the normalization. The error $\epsilon_\text{Schr}$ is the maximum absolute difference between $e^{-\xi}$ and $\Phi_0(\xi)$ for $\xi>0$. This construction yields a profile $\Phi_0$ that is infinitely differentiable for $\xi<0$,  whose Fourier coefficients decay rapidly, which is precisely the regularity required for Schr\"odingerisation to attain the optimal
query-complexity scaling $\mathcal{O}(\log(1/\epsilon_{\text{Schr}}))$.

Now, we finally discretize the Schr\"odingerisation dimension $\xi$ on the interval $\xi\in[-L_\xi,L_\xi-\Delta_\xi]$ using $N_\xi$ grid points $\Delta_\xi=\frac{2L_\xi}{N_\xi}$, which immediately raises a question regarding the boundary conditions imposed. The natural boundaries $\lim_{\xi\rightarrow \pm \infty}\Phi(t,\xi)=0$ are commonly turned to Dirichlet boundaries $\Phi(t,-L_\xi)=\Phi(t,L_\xi-\Delta)=0$ during discretization. However, in the Schr\"odingerisation framework, periodic boundaries are commonly used \cite{PRS2024}. The reason is the absence of dynamics near the boundaries if we select $-L_\xi,L_\xi$ large enough, which is clearly seen from Fig.~\ref{fig:cutoff_theta}. This allows us to impose periodic boundary
\begin{equation}
    \Phi(-L_\xi,t)=\Phi(L_\xi,t)
    \label{eq:periodic boundaries}
\end{equation}
which is the easiest to implement using the quantum framework \cite{guseynov2025quantumframeworksimulatinglinear}.

The benefits of introducing the new dimension $\xi$ are immediately seen if we switch \eqref{eq:schrodingerisation-transform} to the Fourier basis $\xi \rightarrow\eta$; with $\mathcal{F}_{\xi \rightarrow\eta}(\vec\phi)=\vec\psi$
\begin{equation}
    \frac{d\,\vec \psi}{dt}=-i\eta S_1\vec \psi+iS_2\vec \psi,
    \label{eq: ODE with coninious eta}
\end{equation}
where $\xi,\eta$ are still continuous. We also use the standard Fourier-transform identity $\mathcal{F}_{\xi \rightarrow\eta}(\partial/\partial\xi )=i\eta$. Now, we discretize Eq.~\eqref{eq: ODE with coninious eta}. The coordinate operator $\eta$ becomes
\begin{equation}
\begin{aligned}
    &X_\eta=diag(a_\eta+\frac{\Delta_\eta N_\xi}{2},a_\eta+\Delta_\eta[\frac{N_\xi}{2}+1],\dots,b_\eta-\Delta_\eta,a_\eta,a_\eta+\Delta_\eta,\dots,a_\eta+\Delta_\eta[\frac{N_\xi}{2}-1]);\\
    &a_\eta=-\frac{\pi  N_\xi }{ 2L_\xi};\qquad b_\eta=\frac{\pi N_\xi }{ 2L_\xi};\qquad \Delta_\eta=\frac{b_\eta-a_\eta}{N_\xi},
    \end{aligned}
    \label{eq:coordinate operator in eta}
\end{equation}
where we assumed that $N_\xi$ is an even number. One might expect that the coordinate operator would have monotonically increasing diagonal, $diag(a_\eta,a_\eta+\Delta_\eta,\dots,b_\eta-\Delta_\eta)$; however, as shown in Appendix~\ref{appendix:fourier to discrete fourier} the use of the discrete Fourier transform (DFT) instead of natural continuous Fourier transform causes these $N_\xi/2$ cyclic permutations.

Now, collecting everything together we get the final ODE system
\begin{equation}
\left\{
\begin{gathered}
\frac{d\,\vec \psi}{dt}
= i\big[-X_\eta\otimes S_1 + I_\eta\otimes S_2\big]\vec \psi
= iH\vec\psi,\\
\vec{\psi}(t=T)=
\vec \Psi_0\otimes\!\vec w_0,
\end{gathered}
\right.
\label{eq:ODE system SChrodingerised}
\end{equation}
where $\vec\Psi_0$ is obtained by applying discrete Fourier on $\vec \Phi_0$. From now on the $\vec \psi$ is fully discrete. We immediately see that the resulting equation is a Schr\"odinger equation with Hamiltonian $H$.

We now work in the discrete qubit Hilbert space. We introduce $n_\xi$ qubits (i.e. $2^{n_\xi}$ grid points) for the auxiliary Schr\"odingerisation grid $\xi_k$ or $\eta_k$, where $2^{W+1}=N_s$, and $2^{n_\xi}=N_\xi$. The resulting Hilbert space $\mathcal{H}_2^{\otimes(W+n_\xi+1)}$ is described by the normalized vector of length $N_sN_\xi$ (i.e. wave function):
\begin{equation}
    \ket{\psi}=\frac{\vec\psi}{\|\vec\psi\|_2};\qquad\|\vec\psi\|_2=\sqrt{\sum_{i=0}^{N_sN_\xi-1}\psi_i\psi_i^*},
    \label{eq:wave function definition}
\end{equation}
where $*$ is complex conjugation. We suppose that the norm $\mathcal{N}_\psi:=\|\vec\psi(t=0)\|_2$ is analytically computable from the representation in Eq.~\eqref{eq:ODE system SChrodingerised}; this will be important during the readout part.

Thus, after the mapping to the $W+n_\xi+1$-qubit space the Schr\"odinger equation \eqref{eq:ODE system SChrodingerised} becomes
\begin{equation}
\left\{
\begin{gathered}
\frac{d\,|\psi\rangle}{dt}
= i\big[-X_\eta\otimes S_1 + I_\eta\otimes S_2\big]| \psi\rangle
= iH| \psi\rangle,\\
\ket{\psi(T)}=\frac{1}{\mathcal{N}_\psi}
\vec \Psi_0\otimes\!\vec w_0,
\end{gathered}
\right.
\label{eq:final ODE system SChrodingerised}
\end{equation}
where $H$ is a $2^{W+n_\xi+1} \times 2^{W+n_\xi+1}$ Hermitian matrix. Hence, the dynamics of \eqref{eq:final ODE system SChrodingerised} is described by action of the evolution operator $e^{-iHT}$ 
\begin{equation}
    \ket{\psi(t=0)}=e^{-iHT}\ket{\psi(t=T)}=U_H(T)\ket{\psi(t=T)}.
    \label{eq:evolution operator definitio}
\end{equation}
From $\ket{\psi(0)}$, one can retrieve the normalized vector of option prices $\vec V(0)$ in two steps: (i) applying the inverse quantum Fourier transform (block 5 from Fig.~\ref{fig:quantum_scheme_without_postselection}) and (ii) applying the postselection procedure described in Eq.~\eqref{eq:postselection Schrodingerisation}, which in the discrete case means postselecting on bit strings $N_\xi/2+\lceil \frac{\lambda^{\min}_{S_1}}{\Delta_\xi}\rceil,N_\xi/2+\lceil \frac{\lambda^{\min}_{S_1}}{\Delta_\xi}\rceil+1,\dots,N_\xi-1$ (the 6th block from the Fig.~\ref{fig:quantum_scheme_without_postselection}). 

\subsection{Block-encoding as an input model}

In this subsection we address an important problem on how to input the model Hamiltonian 
\begin{equation}
    H=-X_\eta\otimes S_1+I_\eta\otimes S_2.
    \label{eq:Hamiltonian Black-Scholes}
\end{equation}
inside the quantum device. We focus on the block-encoding technique as it is known how to implement it on gate level efficiently \cite{guseynov2025quantumframeworksimulatinglinear,guseynov2024efficientPDE}. Block-encoding technique incorporates a non-unitary matrix of interest as a sub-block in an enlarged unitary matrix.

\begin{definition}[Block-encoding (\cite{gilyen2019quantum})]\label{def:block-encoding}
Let $A$ be an $n$-qubit operator. Given $\alpha, \epsilon \in \mathbb{R}_+$
and an integer $s \in \mathbb{Z}_+$, we say that a unitary $U$ acting on
$s + n$ qubits is an $(\alpha, s, \epsilon)$-block-encoding of $A$ if:
\begin{equation}
  U \ket{0}^s \ket{\psi} =  \ket{\phi}^{n+s},  
\end{equation}
with:
\begin{equation}
 (\bra{0}^s \otimes I) \ket{\phi}^{n+s} = \tilde{A} \ket{\psi},
\quad \text{and} \quad \| A - \alpha \tilde{A} \| \leq \epsilon.  
\end{equation}
\end{definition}

In general, creating a block-encoding for an arbitrary matrix requires exponential resources. However, for the specific task of sparse and structured matrices as those considered in this paper (derived using finite-difference) we have solutions with polynomial complexity.

\begin{proposition}[Multi-dimensional block-encoding (Theorem 6 from \cite{guseynov2025quantumframeworksimulatinglinear})]\label{theorem: block-encoding multidimensional}
Let $H$ be a $2^{W+n_\xi+1} \times 2^{W+n_\xi+1}$ Hermitian matrix with sparsity $\mathfrak s$. 
The matrix $H$ represents the Hamiltonian of a $W+n_\xi+1$-qubit system, governed by the equation:
\begin{equation}
\frac{d\psi}{dt} = iH\psi; 
\end{equation}

where the Hamiltonian takes the form derived from the Schr\"odingerisation 
technique applied to the $d$ dimensional PDE problems \eqref{eqn:ICMultidimblack-scholes}, \eqref{eq:multiHeston_fullPDE}:
\begin{equation}
H =  -X_\eta\otimes S_1+I_\eta\otimes S_2. 
\end{equation}

The matrices \( S_1 \) and \( S_2 \) are given by:
\begin{equation}
S_1 = \begin{pmatrix} \frac{A + {A}^\dagger}{2} & \frac{B}{2} \\ 
\frac{B}{2} & 0 \end{pmatrix}, \quad
S_2 = \begin{pmatrix} \frac{A - {A}^\dagger}{2i} & 
\frac{B}{2i} \\ -\frac{B}{2i} & 0 \end{pmatrix},    
\end{equation}

we can then create a 
$(\mathcal{O}(\mathfrak{s} \|H\|_{\max}), d \lceil \log_2 n \rceil+ 
\lceil \log_2 n_\xi \rceil + 
\lceil \log_2 \mathfrak{s} \rceil + \lceil \log_2 \eta \rceil + 4d + 5, 0)$-block-encoding of $H$. The total amount of resources scales as:
\begin{enumerate}
    \item $\mathcal{O}(dQ_\sigma  n \log n + n_\xi\log n_\xi+d \eta \mathfrak{s} n)$ quantum gates,
    \item $\mathcal{O}(n)$ ancilla qubits.
\end{enumerate}
The $Q_\sigma$ parameter describes the maximum degree of the polynomial approximation for the local volatility $\sigma_i(S_i)\approx\sum_{k=0}^{Q_\sigma-1}\varkappa_kS_i^k$ (can be generalized to piecewise local vol). The parameter $\eta\sim d^2$ is the number of terms in the original PDE.
\end{proposition}

We highlight that in Proposition~\ref{theorem: block-encoding multidimensional} the resulting block-encoding $U_H$ constant $\alpha$ necessarily contains $\|H\|_{\max}$ which is one of the most restrictive bottlenecks in this paper, the importance of this fact will be shown in Appendix~\ref{subsection: evolution operator creation}. From the view of the multidimensional PDEs in this paper \eqref{eqn:ICMultidimblack-scholes}, \eqref{eq:multiHeston_fullPDE} one can calculate that after imposing finite-difference approximation for the derivatives $\partial S_i\rightarrow\Delta S_i$ the norm for the ODE system matrix $\|A\|_{\max}$ scales as
\begin{equation}
    \|H\|_{\max}\sim\|A\|_{\max}=\mathcal{O}(d^22^{2n}),
\end{equation}
where $d^2$ corresponds to the number of terms in the analytical expressions of PDEs considered. The factor $2^{2n}$ comes from the maximum derivative degree contributing PDEs, for the case of multidimensional Heston and Black--Scholes the maximum degree is $2$.

\subsection{Evolution operator implementation}\label{subsection: evolution operator creation}

In this subsection we consider how several block-encodings $U_H$ (or conjugated version $U^\dagger_H$ ) can be used to implement the evolution operator $e^{iHt}$. The key technique here is called Alternating phase modulation sequence which allows one to polynomially transform the matrix under block-encoding.

\begin{definition}[Alternating phase modulation sequence \cite{gilyen2019quantum}]
    \label{def:Alternating phase modulation sequence}
    Let $U_H$ be a $(1,m,0)$-block-encoding of hermitian matrix $H$ such that

    \begin{equation}
       H=(\ket{0}^m\bra{0}^m\otimes I^{\otimes n})U_H(\ket{0}^m\bra{0}^m\otimes I^{\otimes n}); 
    \end{equation}
    let $\Phi\in\mathbb{R}^q$, then we define the $q$-phased alternating sequence $U_\Phi$ as follows
\begin{equation}
U_{\Phi} := 
\begin{cases} 
e^{i \phi_1 (2 \Pi - I)} U_H \prod_{j=1}^{(q-1)/2} \left( e^{i \phi_{2j} (2 \Pi - I)} U_H^\dagger e^{i \phi_{2j+1} (2 \Pi - I)} U_H \right) & \text{if } q \text{ is odd, and} \\
\prod_{j=1}^{q/2} \left( e^{i \phi_{2j-1} (2 \Pi - I)} U_H^\dagger e^{i \phi_{2j} (2 \Pi - I)} U_H \right) & \text{if } q \text{ is even,}
\end{cases}
\end{equation}
where $2\Pi-I=(2\ket{0}^m\bra{0}^m-I^{\otimes m})\otimes I^{\otimes n}$. The phase gate is given by
\begin{equation}
    C_\phi^{00\dots0}=e^{i\phi(2\Pi-I)}=e^{i\phi}\ketbra{00\dots0}+e^{-i\phi}(I-\ketbra{00\dots0}).
    \label{eq:phase gadget in alternating sequence}
\end{equation}
The circuit view is depicted in Fig.~\ref{fig:Alternating phase modulation sequence}.
\end{definition}

The following Proposition uses the combination of alternating phase modulation sequences to apply polynomial transformation of the block encoded matrix $H$.

\begin{figure}[h!]
    \centering
    \includegraphics[width=0.65\linewidth]{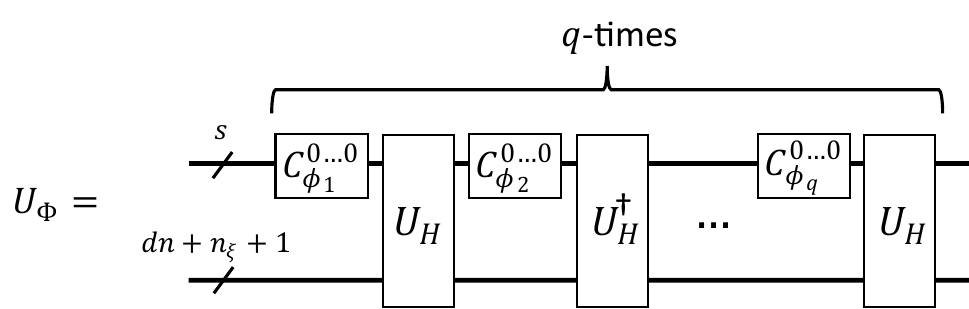}
    \caption{A schematic circuit view for the Alternating phase transformation sequence (for odd polynomial degree $q$). The phase gate $C_\phi^{00\dots 0}$ is given by \eqref{eq:phase gadget in alternating sequence}. The circuit allows to implement a block-encoding of a polynomial $poly(H)$ which obeys certain parity and symmetry.}
\label{fig:Alternating phase modulation sequence}
\end{figure}

\begin{proposition}[Singular value transformation by real polynomials (Corollary 18 from \cite{gilyen2019quantum})]\label{theorem:Singular value transformation by real polynomials}
    Let $P_{q}(x)$ be a degree-$q$ polynomial with the following properties:
\begin{enumerate}
    \item $P_{q}(x)$ has parity $(0$ or $1 \mod2)$, which means that it has only odd or even degrees of $x$;\label{cond1}
    \item $P_{q}(x)$ is a real polynomial, which means that all the coefficients are real;\label{cond2}
    \item For all $x\in[-1,1]$: $\abs{P_{q}(x)}<1$.\label{cond3}
\end{enumerate}
and $U_\Phi$ is the Alternating phase modulation sequence based on the $(1,\mathfrak{m},0)$-block-encoding of the Hamiltonian matrix $U_H$  as in Definition~\ref{def:Alternating phase modulation sequence}. Then there exist $\Phi\in\mathbb{R}^q$, such that
\begin{equation}
\begin{gathered}
    \bra{0}^{\mathfrak{m}+1}\bra{j}^n{U}_{P_q(H)}\ket{0}^{\mathfrak{m}+1}\ket{i}^n=(P_q(H))_{ji}\\
    {U}_{P_q(H)} = \left(H_W\otimes I^{\mathfrak m} \right) \left( |0\rangle \langle 0| \otimes U_{\Phi} + |1\rangle \langle 1| \otimes U_{-\Phi} \right)\left(H_W\otimes I^{\mathfrak{m}} \right).
    \end{gathered}
\end{equation}
    where $H_W = \frac{1}{\sqrt{2}},
\begin{pmatrix}
1 & 1 \\
1 & -1
\end{pmatrix}.$ Moreover, given $P_q(x)$ and $\delta \geq 0$ we can find a $P^\prime_q(x)$ and a corresponding angles $\Phi$, such that $|P^\prime_q(x) - P_q(x)| \leq \delta$ for all $x \in [-1,1]$, using a classical computer in time $\mathcal{O}(\text{poly}(q, \log(1/\delta)))$.
\end{proposition}

By selecting the polynomial $P_q$ to approximate $\sin x$ and $\cos x$ and combining them using linear combination of unitaries \cite{dalzell2023quantum} to create $e^{ix}$ and using the block-encoding $U_H$ from Proposition~\ref{theorem: block-encoding multidimensional} as an input one can prove the following.

\begin{proposition}[Quantum simulation of the option-pricing Hamiltonian (Theorems~8--9 of \cite{guseynov2025quantumframeworksimulatinglinear})]
\label{theorem:QuantumHamiltonianSimulation}
Let $H$ be a $\mathfrak{s}$-sparse Hamiltonian as in Eq.~\eqref{eq:Hamiltonian Black-Scholes}; and the corresponding $(\mathcal{O}(\mathfrak{s} \|H\|_{\max}), d \lceil \log_2 n \rceil+ 
\lceil \log_2 n_\xi \rceil + 
\lceil \log_2 \mathfrak{s} \rceil + \lceil \log_2 \eta \rceil + 4d + 5, 0)$-block-encoding as in Proposition~\ref{theorem: block-encoding multidimensional},
acting on $W+n_\xi+1$ qubits, where $W$ encodes the discretized
asset prices and volatilities, and $n_\xi$ is the auxiliary
register for Schr\"odingerisation, and the $1$ qubit is for making the ODE system homogeneous. Then, for any precision $\epsilon_\text{evol} > 0$ and evolution time $T > 0$, we can implement an $\epsilon_\text{evol}$-precise block-encoding of the evolution operator $e^{-iHT}$, with
block-encoding parameters
\begin{equation}
(2,\, c,\, \epsilon_\text{evol}), \quad
c =  d \lceil \log_2 n \rceil+ 
\lceil \log_2 n_\xi \rceil + 
\lceil \log_2 \mathfrak{s} \rceil + \lceil \log_2 \eta \rceil + 4d + 7.
\end{equation}

\noindent The total gate complexity (C-NOTs and one-qubit rotations) scales as
\begin{equation}
\mathcal{O}\left(\underbrace{\left[
\mathfrak{s}\, \lVert H \rVert_{\max}\, T
+ \frac{\log(1/\epsilon_\text{evol})}{\log\left( e + \frac{\log(1/\epsilon_\text{evol})}{\mathfrak{s}\,\lVert H \rVert_{\max}\, T} \right)}
\right]}_{\text{Ham. sim.}}\right.
\times \underbrace{\left[
dQ_\sigma  n \log n + n_\xi\log n_\xi+d \eta \mathfrak{s} n
\right]}_{\text{Block encoding of } H}\Biggl)  
\end{equation}
where $\mathfrak{s} = \mathcal{O}(d)$ for the central-difference discretization, the number of terms in the original PDE $\eta\sim d^2$ , and $\lVert H \rVert_{\max} \sim \sigma_{\max}^2 2^{2n}$,
with $\sigma_{\max} := \max_{S,i} \sigma_i(S, t)$. This is valid for $\mathfrak{s}\, \lVert H \rVert_{\max}\, T=\mathcal{O}\Bigl(\frac{\log(1/\epsilon_\text{evol})}{\log\left( e + (\mathfrak{s}\,\lVert H \rVert_{\max}\, T)^{-1}\log(1/\epsilon_\text{evol}) \right)}\Bigl)$.
\end{proposition}

\begin{remark}
One can observe now that an exponential advantage in $n$ is not achievable as the final complexity necessarily contains $2^{2n}$ terms induced by $\lVert H\rVert_{\max}$. This problem can be mitigated in several set-ups \cite{Rolando_low,bravyi2025quantumsimulationnoisyclassical}. However, we do not consider it in this paper.
\end{remark}

\subsection{Recovery of the ODE solution}

After Proposition~\ref{theorem:QuantumHamiltonianSimulation} we return to the
solution of the original system~\eqref{eq:MAINhomog-ode}.
The Schr\"odingerisation procedure embeds the dynamics into a register
decomposed (from left to right) as
\begin{equation}
\underbrace{c}_{\text{evol ancillas}},\qquad
\underbrace{W+1}_{\text{system}},\qquad
\underbrace{n_\xi}_{\text{Schr\"odingerisation}},\qquad
\underbrace{a_{\mathrm{prep}}}_{\text{prep ancillas}},
\end{equation}
so that $c$ qubits are used as block-encoding auxiliaries for
$U_{\mathrm{evol}}$, $W+1$ qubits encode the discretised components of
$\vec w$, the $n_\xi$ qubits carry the Schr\"odingerisation coordinate
$\xi$, and the $a_{\mathrm{prep}}$ qubits are used for state
preparation.

The cut-off state–preparation unitary
$U_{\mathrm{cut\text{-}off}}$ (see
Appendix~\ref{section:state prep}) acts only on the
rightmost $n_\xi + a_{\mathrm{prep}}$ qubits.
We then apply a quantum Fourier transform $QFT$ on the $n_\xi$
Schr\"odingerisation qubits, followed by the evolution unitary
$U_{\mathrm{evol}}$ from
Proposition~\ref{theorem:QuantumHamiltonianSimulation} on the joint
system–Schr\"odingerisation space of size $W+1+n_\xi$, and finally the
inverse transform $QFT^\dagger$ on the same $n_\xi$ register.
Since the physically meaningful support lies in $\xi>\lambda^{\min}_{S_1}$, the postselection projector on the $\xi$–register is
\begin{equation}
  \Pi_{\xi>0}
  =
  \sum_{k = k_\lambda}^{2^{n_\xi}-1}
  \ket{k}^{n_\xi}\bra{k}^{n_\xi},
  \qquad
  k_\lambda
  =
  \left\lceil
  \frac{\lambda_{S_1}^{\min}}{\Delta_\xi}
  \right\rceil .
\end{equation}
Putting all stages together (see
Fig.~\ref{fig:quantum_scheme_without_postselection}), the
branch of the computation is
\begin{align}
&
\bigl(
  \ket{0}^{c}\bra{0}^{c}
  \otimes
  I^{\otimes(W+1)}
  \otimes
  \Pi_{\xi>0}
  \otimes
  \ket{0}^{a_{\mathrm{prep}}}\bra{0}^{a_{\mathrm{prep}}}
\bigr)
\cdot
\bigl(
  I^{\otimes(c+W+1)}
  \otimes
  QFT^\dagger
  \otimes
  I^{\otimes a_{\mathrm{prep}}}
\bigr)
\nonumber\\[1mm]
&\quad\cdot
\bigl(
  U_{\mathrm{evol}}
  \otimes
  I^{\otimes a_{\mathrm{prep}}}
\bigr)
\cdot
\bigl(
  I^{\otimes(c+W+1)}
  \otimes
  QFT
  \otimes
  I^{\otimes a_{\mathrm{prep}}}
\bigr)
\nonumber\\[1mm]
&\quad\cdot
\bigl(
  I^{\otimes(c+W+1)}
  \otimes
  U_{\mathrm{cut\text{-}off}}
\bigr)
\ket{0}^{\otimes(c+W+1+n_\xi+a_{\mathrm{prep}})}
=
\ket{\psi}^{(W+1)}.
\end{align}
After discarding the $\xi$–register and all ancillas, the remaining
$(W+1)$–qubit state $\ket{\psi}$ encodes the normalised approximation
to the vector~$\vec w$ at the desired time.

The success probability satisfies
$p_{\mathrm{success}}
  = \mathcal{O}(\|\vec w(t)\|^{2}/\|\vec w(0)\|^{2})$
\cite{analog}.
We assume that the composite unitary
\begin{align}
U_{\mathrm{full}}
&=
\bigl(
  I^{\otimes(c+W+1)}
  \otimes
  QFT^\dagger
  \otimes
  I^{\otimes a_{\mathrm{prep}}}
\bigr)
\cdot
\bigl(
  U_{\mathrm{evol}}
  \otimes
  I^{\otimes a_{\mathrm{prep}}}
\bigr)
\nonumber\\[1mm]
&\quad\cdot
\bigl(
  I^{\otimes(c+W+1)}
  \otimes
  QFT
  \otimes
  I^{\otimes a_{\mathrm{prep}}}
\bigr)
\cdot
\bigl(
  I^{\otimes(c+W+1)}
  \otimes
  U_{\mathrm{cut\text{-}off}}
\bigr)
\end{align}
is used as the signal unitary in a standard amplitude–amplification
procedure (refer to Appendix~\ref{app:amplification}).
In particular, we coherently iterate $U_{\mathrm{full}}$ and its
adjoint $U_{\mathrm{full}}^\dagger$ within the usual Grover-type
iterate that reflects about the initial state
$\ket{0}^{\otimes(c+W+1+n_\xi+a_{\mathrm{prep}})}$ and about the
subspace defined by the projector
$\ket{0}^{c}\bra{0}^{c}
 \otimes I^{\otimes(W+1)}
 \otimes \Pi_{\xi>0}
 \otimes \ket{0}^{a_{\mathrm{prep}}}\bra{0}^{a_{\mathrm{prep}}}$.
After
$\mathcal{O}\!\bigl(\|\vec w(0)\|/\|\vec w(t)\|\bigr)$ such
alternating applications of $U_{\mathrm{full}}$ and $U_{\mathrm{full}}^\dagger$,
the overall success probability is boosted to $\mathcal{O}(1)$.

\FloatBarrier
\section{Quantum state preparation}\label{section:state prep}

In this paper, we consider quantum algorithms whose inputs and
outputs are classical data. To this end, we must
specify how to implement quantum state preparation (blocks
$1,2$ from Fig.~\ref{fig:quantum_scheme_without_postselection}) and how
to retrieve the option price $V(S_0,v_0,0)$ for the
target $S_0$ and $v_0$ from the quantum scheme
output $\vec w(0)/\lVert \vec w(0)\rVert$. We begin by
addressing the input problem. Most of the functions
considered are closely approximated by polynomials or piecewise
polynomials of moderate degree. The following Proposition provides
the basis for all necessary initial states in this
paper.

\begin{proposition}[Efficient preparation of piecewise polynomial functions~\cite{guseynov2024explicit_Quantum_state_prep}]
\label{theorem:piecewise_poly}
Let $f(x)$ be a piecewise continuous function
$f: \mathbb{R} \rightarrow \mathbb{C}$ that can be decomposed
into $G$ pieces, each described by a degree-$Q_g$ polynomial:
\begin{equation}
f(x) =
\begin{cases}
f_1(x) = \sum_{i=0}^{Q_1} \alpha^{(1)}_i x^i, & \text{if } K_1 \geq x \geq a \\
f_2(x) = \sum_{i=0}^{Q_2} \alpha^{(2)}_i x^i, & \text{if } K_2 \geq x > K_1 \\
\vdots \\
f_G(x) = \sum_{i=0}^{Q_G} \alpha^{(G)}_i x^i, & \text{if } b \geq x > K_{G-1}
\end{cases} 
\end{equation}

where $\alpha_i^{(g)} \in \mathbb{C}$.
Then, there exists a $n+a_\text{prep}$-qubit quantum circuit $U_f$ which is $(\mathcal{O}(1/\sqrt{\mathcal{F}}),\lceil \log_2 n\rceil+3,0)$ that
efficiently prepares a $2^n$-dimensional discretized quantum state
proportional to $f(x)$, using:
\begin{enumerate}
    \item $\mathcal{O}\left( Q_{\max} n \log n+Gn+Q_{\max}G\right)$ C-NOT and single-qubit gates,
    \item $n-1$ ancilla qubits.
\end{enumerate}
with success probability proportional to the filling ratio
\begin{equation}
\mathcal{F} := \|f\|_2^2/(2\|f\|^2_{\max}).
\end{equation}
\end{proposition}

The probability of success $p_\text{success}=\mathcal{O}(\mathcal{F})$ in Proposition~\ref{theorem:piecewise_poly} can be further boosted using the method described in Appendix~\ref{app:amplification}. Thus, one can implement $(\mathcal{O}(1),c,\epsilon_\text{evol})$ using $\approx 1/\sqrt{\mathcal{F}}$ queries to $U_f$ and $U_f^\dagger$.

The vanilla option payoff $\max(0,S-K)$ or
$\max(0,K-S)$ is straightforward to prepare using Proposition~\ref{theorem:piecewise_poly} by choosing the number of
pieces $G=2$ and the maximum polynomial degree
$Q_{\max}=1$, resulting in the total gate complexity
$\mathcal{O}(n\log n)$.

Next, we consider the
$\Phi(t=T,\xi)$ from Eq.~\eqref{eq:initial schrodingerisation}
and depicted in Fig.~\ref{fig:cutoff_theta}. From the
form of this function, we see that it is
smooth and can be well approximated by polynomials. Our
numerical tests \cite{guseynov_gitlab_placeholder} show that if we use
a single-piece $G=1$ polynomial to approximate
$\Phi(t=T,\xi)$, then achieving fidelity over $0.99$
requires $Q_{\max}=27$; however, if we use
a multi-piece function and minimize $Q_{\max}+G$
to achieve fidelity over $0.99$ the resulting parameters are $Q_{\max}=5,G=4$.

\subsection{Basket option payoffs}

To implement the basket option payoffs from
Table~\ref{tab:vanilla payoffs}, we rely on the
main technique underlying Proposition~\ref{theorem:piecewise_poly}, often
referred to as the \emph{indicator qubit} method~\cite{guseynov2024explicit_Quantum_state_prep}. 
In this approach, we reversibly (i.e., via a unitary operation) set an 
auxiliary qubit to the state $\ket{1}$ on a chosen region of interest:
\begin{equation}
    \forall \ket{\psi}=\sum_{i=0}^{2^{n}-1}\psi_i\ket{i},\qquad
    U^k_\text{indic}\left(\ket{\psi}^n\ket{0}\right)
    =\sum_{i=0}^{k-1}\psi_i\ket{i}^n\ket{0}
    +\sum_{i=k}^{2^n-1}\psi_i\ket{i}^n\ket{1}.
    \label{eq:indicator qubit}
\end{equation}
Ref.~\cite{guseynov2024explicit_Quantum_state_prep} explicitly demonstrates how to implement 
$U^k_\text{indic}$ for any $k$ with gate complexity $\mathcal{O}(n)$.

\begin{table}[h!]
    \centering
    \begin{tabular}{lccc}
        \toprule
        Function name & Formula & Gate cost & Probability of success $\mathcal{F}$ \\
        \midrule
        Vanilla 
        & $\max(0,S_T - K)$
        & $\mathcal{O}(n \log n)$ 
        & $\mathcal{O}(1)$ \\[2pt]
        Basket 
        & $\max\bigl(w^\top S_T - K, 0\bigr)$ 
        & $\mathcal{O}(d\,n \log n)$ 
        & $\mathcal{O}(1)$ \\[2pt]
        Spread 
        & $\max\bigl(w_\uparrow^\top S_T - w_\downarrow^\top S_T - K, 0\bigr)$ 
        & $\mathcal{O}(d\,n \log n)$ 
        & $\mathcal{O}(1/d)$ \\[2pt]
        Best 
        & $\max\bigl(\max_i S_{i,T} - K, 0\bigr)$ 
        & $\mathcal{O}(d^2 n + d\,n \log n)$ 
        & $\mathcal{O}(1)$ \\[2pt]
        Worst 
        & $\max\bigl(\min_i S_{i,T} - K, 0\bigr)$ 
        & $\mathcal{O}(d^2 n + d\,n \log n)$ 
        & $\mathcal{O}(1/d^2)$ \\[2pt]
        Smooth cut-off 
        & $\Phi(T,\xi) = \zeta(\xi)\,e^{-\xi}$ 
        & $\mathcal{O}(n \log n)$ 
        & $\mathcal{O}(1)$ \\
        \bottomrule
    \end{tabular}
    
    \caption{Complexity of creating quantum states corresponding to popular payoff 
    call option types (the result for put types is similar) and the smooth cut-off used for Schr\"odingerisation. Here $d$ denotes 
    the number of assets in the contract, and we assume that all weight vectors 
    are normalised so that $\sum_i w_i = 1$ (and similarly 
    $\sum_i w_{\uparrow,i} = \sum_i w_{\downarrow,i} = 1$ for the spread payoff). 
    The corresponding payoff definitions are given in 
    Table~\ref{tab:vanilla payoffs}. The third column 
    reports the gate complexity of a single state-preparation attempt, while the 
    fourth column shows the associated success probability (proportional to filling ratio $\mathcal{F}$). The 
    overall cost of preparing the desired quantum state is determined by both 
    quantities: it scales with the gate cost per attempt and the number of 
    repetitions required, which is inversely proportional to the success 
    probability. Using robust oblivious amplitude amplification%
    ~\cite{BCCKS14}, one can boost the success probability 
    close to $1$, effectively trading additional overhead in 
    gate complexity for near-deterministic state preparation.}
    \label{table:quantum state prep of payoffs}
\end{table}

This technique is essential for building piecewise functions. In
particular, using the indicator qubit we can mark regions where
$A(S_T) > B(S_T)$ for payoffs of the form
$f(S_T)=\max(A(S_T),B(S_T))$, as in
Table~\ref{tab:vanilla payoffs}, allowing us
to prepare $A(S_T)$ and $B(S_T)$ independently. In addition, for the
payoffs $f_\text{best}$ and $f_\text{worst}$ we can mark the regions
where $\lVert S_T \rVert_{\pm\infty} = S_T^{i}$ by applying the
indicator qubit technique $d$ times.

By applying a controlled version of Proposition~\ref{theorem:piecewise_poly}
to the marked regions, and using the LCU technique to construct
$w^\top S_T$, we obtain the gate complexities for all payoff types
considered in this paper, as summarized in
Table~\ref{table:quantum state prep of payoffs}.

\FloatBarrier
\section{Information retrieval}\label{section:readout}

In this paper we consider a quantum algorithm based on
PDE solving to price multi-asset Heston and Black–Scholes
models. The output of these algorithms is a unitary
operation $U$ which acts as a block encoding.
By applying $U$ to $\ket{0}^{W+a}$
we prepare a quantum state $\ket{\psi}$ that encodes
information about the option prices. This Appendix is
dedicated to retrieving the relevant data, namely the
option prices $V$, from the resulting final state.

For simplicity, in this Appendix we denote the number
of system qubits by $W=dn$ for Black--Scholes and
$W=2dn$ for Heston. Sections~\ref{section:evoltuion non conservative}
and~\ref{section:state prep} provide a systematic way
to prepare the block-encoding $U$ such that
\begin{equation}
    (\Pi\otimes I^{\otimes W})U\ket{0}^{W+a}=\mathcal{N}_{\varsigma}\ket{\varsigma}^a\otimes\ket{\psi}^W.
    \label{eq:readout schematic intro IX}
\end{equation}
Postselection is not a trace-preserving channel, hence the factor $\mathcal{N}_\varsigma$. We assume that state preparation and the evolution block-encodings are boosted using the oblivious amplitude amplification (see Appendix~\ref{app:amplification}), resulting in $\mathcal{N}_\varsigma=\mathcal{O}(1)$ which is assumed to be a known value. 

The projector $\Pi=\ket{0}^{a_\text{prep}}\bra{0}^{a_\text{prep}}\otimes
\ket{0}^{c}\bra{0}^{c}\otimes \Pi_{\xi>0}$ incorporates all
projectors corresponding to post-selections induced by:
\begin{enumerate}
    \item $\ket{0}^{a_\text{prep}}\bra{0}^{a_\text{prep}}$ with
    $a_\text{prep}=d\lceil \log_2 n\rceil + 3d$ arising from
    the quantum state preparation (Proposition~\ref{theorem:piecewise_poly}).
    \item $\ket{0}^{c}\bra{0}^{c}$ with $c =  d \lceil \log_2 n \rceil+
    \lceil \log_2 n_\xi \rceil +
    \lceil \log_2 \mathfrak{s} \rceil + \lceil \log_2 \eta \rceil + 4d + 7$
    coming from the evolution stage
    (Proposition~\ref{theorem:QuantumHamiltonianSimulation}), where
    $d\rightarrow2d$ for the Heston model.
    \item $\Pi_{\xi>0}=\sum_{k=k_\lambda}^{2^{n_\xi}-1}\ket{k}^{n_\xi}
    \bra{k}^{n_\xi}$ with $k_\lambda=\lceil\lambda_{S_1}^{\min}/\Delta_\xi\rceil$
    induced by Schr\"odingerisation (see
    Eq.~\eqref{eq:postselection Schrodingerisation}).
\end{enumerate}
The state $\ket{\varsigma}^a$ indicates a successful projection to the subspace generated by the $\Pi$ and up to $\mathcal{O}(\log\frac{1}{\epsilon_\text{Schr}})$ has the form:
\begin{equation}
    \ket{\varsigma}^a=\ket{0}^{a_\text{prep}}\otimes\ket{0}^c\otimes\frac{1}{\sqrt{\sum_{m=k_\lambda}^{2^{n_\xi}-1}e^{-2m\Delta \xi}}}\sum_{k=k_\lambda}^{2^{n_\xi}-1}e^{-k\Delta \xi}\ket{k}^{n_\xi}.
    \label{eq:varsigma state}
\end{equation}
We note that this state is easy to prepare using Proposition~\ref{theorem:piecewise_poly} on the $n_\xi$-qubit register, we denote such unitary as $S_\varsigma$.

The state $\ket{\psi}^W$ describes the option prices at
a specific time $t=0$,
\begin{equation}
    \ket{\psi}^W=\frac{1}{N_V(t=0)}\sum_{i=0}^{2^W-1} V(t=0,\vec S_i,\vec v_i)\ket{i}^W:=\sum_{i=0}^{2^W-1}\psi_i\ket{i}^W,
    \label{eq:state after postselection IX}
\end{equation}
where $N_V(t)=\sqrt{\sum_{i}\abs{V(t,\vec S_i,\vec v_i)}^2}$
is a normalization factor. In
Eq.~\eqref{eq:state after postselection IX} we use a
multi-index grid system: for a given index $i$ we fix
all $d$ stock prices $\vec S$ and all $d$ volatilities
$\vec v$ at a specific grid point, and the sum
$\sum_{i}$ runs over all possible grid points.
For the Black–Scholes model the volatilities
$\vec v$ are omitted.

Typically, in derivative pricing one is interested
in the contract value at a single specific choice
of stock prices $\vec S_q$ and volatilities
$\vec v_q$ because derivative prices are typically evaluated near the current state \((\vec S_0,\vec v_0)\), i.e. at initial time \(t=0\).
In our setting this corresponds to
$V(t=0,\vec S_q,\vec v_q)=\psi_qN_V(t=0)$. In this
Appendix we show how the quantities $\psi_q$ and
$N_V(t=0)$ can be estimated, and we provide the
corresponding error bounds and complexity analysis.

\subsection{Single amplitude estimation}
\label{subsec:single-amplitude-estimation}

In this subsection we provide an algorithm on how to estimate a fixed coefficient $\psi_q$ of the pricing state $\ket{\psi}^W=\sum_{i=0}^{2^W-1}\psi_i\ket{i}^W$. We recall that the block-encoding $U$ satisfies
\begin{equation}
(\Pi\otimes I^{\otimes W})U\ket{0}^{W+a}
=\mathcal{N}_{\varsigma}\ket{\varsigma}^a\otimes\ket{\psi}^W,
\label{eq:readout-recall1}
\end{equation}
where $\mathcal{N}_{\varsigma}=\mathcal{O}(1)$ is assumed known (computed from the Sections~\ref{section:evoltuion non conservative},~\ref{section:state prep}) and $\ket{\varsigma}^a$ lies in the image of $\Pi$. Next, we fix an index $q$ and let $P_q$ be a basis-preparation unitary on the $W$-qubit system register such that $P_q\ket{0}^{W}=\ket{q}^{W}$. We note that such a unitary can be prepared with up to $W$ Pauli $X$ operators.

Next, we assume access to a unitary $S_\varsigma$ acting on $(a+b_\xi)$ qubits, where $b_\xi=\left\lceil\log_2 n_\xi\right\rceil+3$ such that
\begin{equation}
\bigl(\ket{0}\!\bra{0}^{\,b_\xi}\otimes I^{\otimes a}\bigr)\,S_\varsigma\,\ket{0}^{a+b_\xi}
=
\mathcal{N}_{\varsigma,\xi}\,\ket{0}^{b_\xi}\ket{\varsigma}^{a},
\label{eq:Sv-block}
\end{equation}
with a known constant $\mathcal{N}_{\varsigma,\xi}=\mathcal{O}(1)$.

We define the reference preparation on the joint $(a+b_\xi+W)$-qubit work register by
\begin{equation}
V_q := S_\varsigma\otimes P_q,
\qquad
V_q\ket{0}^{a+b_\xi+W}=S_\varsigma\ket{0}^{a+b_\xi}\otimes\ket{q}^{W}.
\label{eq:Vq}
\end{equation}
We now introduce a single ancilla qubit $b$ and construct the multiplexed preparation operator
\begin{equation}
M_q:=\ket{0}\!\bra{0}\otimes V_q+\ket{1}\!\bra{1}\otimes\bigl(I^{\otimes b_\xi}\otimes U\bigr),
\label{eq:Mq}
\end{equation}
together with the Hadamard-test query unitary
\begin{equation}
\mathcal{A}_q:=(H\otimes I)\,M_q\,(H\otimes I),
\label{eq:Aq}
\end{equation}
which is applied to the initial state $\ket{0}\ket{0}^{a+b_\xi+W}$. A direct calculation yields
\begin{equation}
\mathcal{A}_q\ket{0}\ket{0}^{a+b_\xi+W}
=
\frac12\left[
\ket{0}\left(V_q\ket{0}^{a+b_\xi+W}+\bigl(I^{\otimes b_\xi}\otimes U\bigr)\ket{0}^{a+b_\xi+W}\right)
+
\ket{1}\left(V_q\ket{0}^{a+b_\xi+W}-\bigl(I^{\otimes b_\xi}\otimes U\bigr)\ket{0}^{a+b_\xi+W}\right)
\right].
\label{eq:Aq-action}
\end{equation}
Therefore,
\begin{align}
\Pr[b=0]
&=
\left\|\frac{V_q\ket{0}^{a+b_\xi+W}+\bigl(I^{\otimes b_\xi}\otimes U\bigr)\ket{0}^{a+b_\xi+W}}{2}\right\|^2
\nonumber\\
&=\frac{1+\Re\!\left(\bra{0}^{a+b_\xi+W}\,V_q^\dagger\bigl(I^{\otimes b_\xi}\otimes U\bigr)\ket{0}^{a+b_\xi+W}\right)}{2}.
\label{eq:prob-b0-general}
\end{align}
Using $V_q^\dagger=S_\varsigma^\dagger\otimes P_q^\dagger$, $P_q^\dagger\ket{0}^{W}=\ket{q}^{W}$, and the block-encoding property
\eqref{eq:Sv-block}, the overlap reduces to
\begin{equation}
\bra{0}^{a+b_\xi+W}\,V_q^\dagger\bigl(I^{\otimes b_\xi}\otimes U\bigr)\ket{0}^{a+b_\xi+W}
=
\mathcal{N}_{\varsigma,\xi}\,\bra{\varsigma}^{a}\bra{q}^{W}U\ket{0}^{W+a}.
\label{eq:overlap-step1}
\end{equation}
Since $\ket{\varsigma}^{a}$ lies in the image of $\Pi$, we have $\bra{\varsigma}^{a}=\bra{\varsigma}^{a}\Pi$, and hence
\begin{equation}
\bra{\varsigma}^{a}\bra{q}^{W}U\ket{0}^{W+a}
=
\bra{\varsigma}^{a}\bra{q}^{W}(\Pi\otimes I^{\otimes W})U\ket{0}^{W+a}
=\mathcal{N}_{\varsigma}\,\bra{q}^{W}\ket{\psi}^{W}
=\mathcal{N}_{\varsigma}\psi_q.
\label{eq:overlap-step2}
\end{equation}
Combining \eqref{eq:prob-b0-general}--\eqref{eq:overlap-step2} gives
\begin{equation}
\Pr[b=0]=\frac{1+\mathcal{N}_{\varsigma,\xi}\,\mathcal{N}_{\varsigma}\,\Re(\psi_q)}{2}.
\label{eq:prob-b0-final}
\end{equation}
In particular, when $\psi_q\in\mathbb{R}$ corresponding to our setting,
\begin{equation}
\psi_q=\frac{2\,\Pr[b=0]-1}{\mathcal{N}_{\varsigma,\xi}\,\mathcal{N}_{\varsigma}}.
\label{eq:psi-from-prob}
\end{equation}
Thus amplitude estimation applied to the query unitary $\mathcal{A}_q$ with the ``good'' event $b=0$ yields an estimator for
$\psi_q$ via \eqref{eq:psi-from-prob}, see Fig.~\ref{fig:quantum_scheme_for_A_q}.

\begin{figure}[h!]
    \centering
    \includegraphics[width=0.45\linewidth]{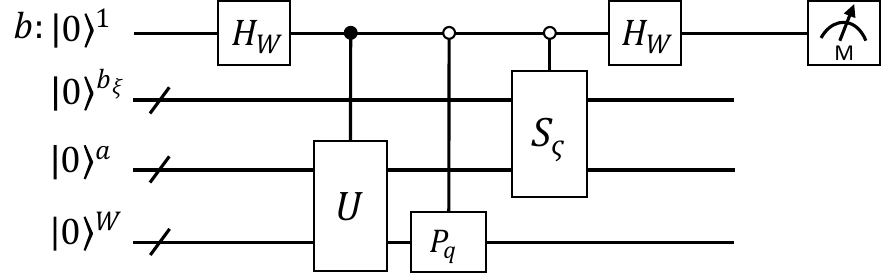}
    \caption{A circuit view for the Hadamard-test query unitary $\mathcal{A}_q$, see Eq.~\eqref{eq:Aq}. The block-encoding unitary $U$ is constructed using the methodology from the Sections~\ref{section:evoltuion non conservative},\ref{section:state prep}, and is given in Eq.~\ref{eq:readout-recall1}; the unitary $P_q$ is used to prepare a computational basis state $\ket{q}^W$ and consists of up to $W$ Pauli $X$ gates; the state preparation block-encoding $S_\varsigma$ is used to prepare the terminal state of the Schr\"odingerisation dimension, see Eq.~\eqref{eq:Sv-block}.}
\label{fig:quantum_scheme_for_A_q}
\end{figure}

\begin{proposition}[Modified iterative quantum amplitude estimation {\cite{doi:10.1137/1.9781611977561.ch12}}]
\label{thm:miqae-psi_q}
Fix $q$ and let $\mathcal{A}_q$ be the query unitary from \eqref{eq:Aq}, acting on the ancilla qubit $b$ and the work register of size $a+b_\xi+W$.
Define
\begin{equation}
a_q:=\Pr[b=0].
\end{equation}
For any $\epsilon>0$ and $\delta\in(0,1)$, the modified IQAE algorithm applied to $\mathcal{A}_q$ returns an estimate $\hat a_q$ such that
\begin{equation}
\Pr\!\left[\,|\hat a_q-a_q|<\epsilon\,\right]\ge 1-\delta,
\end{equation}
using $\mathcal{O}\!\bigl(\frac{1}{\epsilon}\log\frac{1}{\delta}\bigr)$ controlled applications of $\mathcal{A}_q$ (equivalently, of the associated Grover iterate).
\end{proposition}
In our setting,
\begin{equation}
a_q=\frac{1+\mathcal{N}_{\varsigma,\xi}\mathcal{N}_{\varsigma}\,\psi_q}{2},
\qquad
\psi_q=\frac{2a_q-1}{\mathcal{N}_{\varsigma,\xi}\mathcal{N}_{\varsigma}}
\end{equation}
then using Proposition~\ref{thm:miqae-psi_q} one can compute $\psi_q$ using $\mathcal{O}(\frac{1}{\epsilon}\log \frac{1}{\delta})$ queries to $\mathcal{A}_q$.

\subsection{Normalization factor propagation}

Let us now address the problem of finding
$N_V(t)=\sqrt{\sum_{i}\abs{V(t,\vec S_i,\vec v_i)}^2}$ for $0\leq t\leq T$.
In this paper we use this normalization factor to find the option values using
Eq.~\eqref{eq:state after postselection IX}. Proposition~\ref{thm:miqae-psi_q}
allows us to find $\psi_q$; which together with $N_V(t=0)$ allows us to find
the option prices.

We first remind the reader that the option value at maturity is the payoff
$V(t=T,\vec S,\vec v)=f(\vec S)$; hence the terminal normalization factor is
\begin{equation}
\label{eq:NV_terminal}
N_V(t=T)=\sqrt{\sum_i |f(\vec S_i)|^2}.
\end{equation}
On a homogeneous $W/n$-dimensional grid (with step $h=\Delta x$ in each coordinate) over the truncated cube
$\Omega := [0,\vec S_{\max}]\times[0,\vec v_{\max}]$, we use the approximation
\begin{equation}
\label{eq:NV_terminal_integral_approx}
\sum_i |f(\vec S_i)|^2 \;\approx\; \frac{1}{h^{W/n}}\int_{\Omega} |f(\vec S)|^2\, d\vec S\, d\vec v .
\end{equation}
Moreover, assuming $g(\vec S,\vec v):=|f(\vec S)|^2$ is Lipschitz on $\Omega$ with constant $L_g$, the corresponding
\emph{relative} error can be bounded as
\begin{equation}
\label{eq:NV_terminal_integral_relbound}
\frac{\left|\sum_i g(\vec S_i,\vec v_i)-\frac{1}{h^{W/n}}\int_{\Omega} g(\vec S,\vec v)\, d\vec S\, d\vec v\right|}
{\frac{1}{h^{W/n}}\int_{\Omega} g(\vec S,\vec v)\, d\vec S\, d\vec v}
\;\le\;
L_g\,\sqrt{W/n}\,h\,
\frac{\mathrm{Vol}(\Omega)}{\int_{\Omega} g(\vec S,\vec v)\, d\vec S\, d\vec v}
=\mathcal{O}(2^{-n}\sqrt{W/n}),
\end{equation}
where we used $\int_{\Omega} g(\vec S,\vec v)\, d\vec S\, d\vec v>0$, and assumed the discretization step satisfies
$h\sim 2^{-n}$, where $n$ is the number of associated qubits. Later, we assume that the value $N_V(t=T)$ is computed
analytically by the integral \eqref{eq:NV_terminal_integral_approx} using the explicit (and simply integrable) form of the
payoff functions in Table~\ref{tab:vanilla payoffs}.

Next, we propagate the $N_V(t)$ by using the result given in \cite{schrodingerisation_optimal_queries} (Section~2.3).
The probability of postselecting $\xi>\lambda_{S_1}^{\min}$ is
\begin{equation}
\label{eq:probability postselection Schrodingerisation}
\Pr[\xi>\lambda_{S_1}^{\min}]
=
\frac{\int_{\lambda^{\min}_{S_1} T}^{+\infty}\abs{\zeta(\xi)e^{-\xi}}^2\,d\xi}
{\int_{-\infty}^{+\infty}\abs{\zeta(\xi)e^{-\xi}}^2\,d\xi}\,
\frac{N_V(t=0)^2}{N_V(t=T)^2}
=
C_{\lambda_S}\frac{N_V(t=0)^2}{N_V(t=T)^2},
\end{equation}
where $\zeta(\xi)e^{-\xi}$ is the initial state in the Schr\"odingerisation dimension, see Eq.~\eqref{eq:initial schrodingerisation}.
On the other hand, the value $\Pr[\xi>\lambda_{S_1}^{\min}]$ can be estimated in two different ways:
\begin{enumerate}
    \item One of them is to directly record the number of successful postselections in the $n_\xi$-qubit register from the quantum scheme
    in Fig.~\ref{fig:quantum_scheme_without_postselection}. Denoting by $\hat p$ the empirical success probability, the corresponding
    Monte-Carlo (shot-noise) scaling is
    \begin{equation}
    \label{eq:MC_shots}
    N_{\text{shots}}=\mathcal{O}\!\left(\frac{1}{\epsilon_{\text{shot}}^2}\log \frac{1}{\delta}\right),
    \qquad
    \epsilon_{\text{shot}}:=\frac{|\hat p-\Pr[\xi>\lambda_{S_1}^{\min}]|}{\Pr[\xi>\lambda_{S_1}^{\min}]}.
    \end{equation}
    The given dependence is not optimal; however, we note that in our work the scheme in Fig.~\ref{fig:quantum_scheme_without_postselection}
    is used multiple times. Therefore, we can record many samples without specifically dedicated experiment.

    \item Another way is to use Proposition~\ref{thm:miqae-psi_q} with the whole quantum circuit as an observable $\mathcal{A}$.
    The resulting error scaling is
    \begin{equation}
    \label{eq:AE_shots}
    N_{\text{shots}}=\mathcal{O}\!\left(\frac{1}{\epsilon_{\text{shot}}}\log \frac{1}{\delta}\right),
    \end{equation}
    which is an optimal scaling.
\end{enumerate}

Finally, the normalization factor $N_V(t=0)$ is estimated by rearranging Eq.~\eqref{eq:probability postselection Schrodingerisation} as
\begin{equation}
\label{eq:NV0_from_postselection}
N_V(t=0)=N_V(t=T)\sqrt{\frac{\Pr[\xi>\lambda_{S_1}^{\min}]}{C_{\lambda_S}}}.
\end{equation}
Consequently, if $\widehat N_V(t=T)$ has relative error $\epsilon_{\mathrm{term}}=\mathcal{O}(2^{-n}\sqrt{W/n})$ (see Eq.~\eqref{eq:NV_terminal_integral_relbound}) and
$\Pr[\xi>\lambda_{S_1}^{\min}]$ is estimated with relative error $\epsilon_{\text{shot}}$, then the induced relative error in
\eqref{eq:NV0_from_postselection} satisfies,
\begin{equation}
\label{eq:NV0_error_propagation}
\frac{|\widehat N_V(t=0)-N_V(t=0)|}{N_V(t=0)}\;\approx\;\epsilon_{\mathrm{term}}+\frac{1}{2}\epsilon_{\text{shot}}.
\end{equation}

\subsection{Complexity scaling for readout}
\label{subsec:readout_complexity}

Now, we estimate how many applications of $U$ from Eq.~\eqref{eq:readout schematic intro IX} is necessary to estimate the single value $V(t=0,\vec S_q,\vec v_q)$ with error $\epsilon_V$. Fix a target grid point $q$ corresponding to $(\vec S_q,\vec v_q)$. The readout step is
\begin{equation}
\label{eq:readout-recall+norm}
V(t=0,\vec S_q,\vec v_q)=\psi_q\,N_V(t=0).
\end{equation}
As in the Hadamard-test construction, we estimate $\psi_q$ via mIQAE applied to the query unitary $\mathcal{A}_q$ with ``good'' event $b=0$:
\begin{equation}
a_q:=\Pr[b=0]=\frac{1+N_{\varsigma,\xi}N_{\varsigma}\,\psi_q}{2},
\qquad
\psi_q=\frac{2a_q-1}{N_{\varsigma,\xi}N_{\varsigma}}.
\end{equation}
By Proposition~\ref{thm:miqae-psi_q}, mIQAE outputs $\hat a_q$ such that
$|\hat a_q-a_q|\le \epsilon_{\mathrm{AE}}$ with probability at least $1-\delta$ using
$\mathcal{O}\!\left(\frac{1}{\epsilon_{\mathrm{AE}}}\log\frac{1}{\delta}\right)$ controlled calls to $\mathcal{A}_q$.
Therefore,
\begin{equation}
|\hat\psi_q-\psi_q|\le \frac{2\epsilon_{\mathrm{AE}}}{N_{\varsigma,\xi}N_{\varsigma}}
\qquad\Longrightarrow\qquad
|\hat V_q-V_q|
\le
N_V(t=0)\,\frac{2\epsilon_{\mathrm{AE}}}{N_{\varsigma,\xi}N_{\varsigma}},
\end{equation}
where $V_q:=V(t=0,\vec S_q,\vec v_q)$. Imposing $|\hat V_q-V_q|\le \epsilon_V$ yields the choice
\begin{equation}
\epsilon_{\mathrm{AE}}=\frac{N_{\varsigma,\xi}N_{\varsigma}}{2}\,\frac{\epsilon_V}{N_V(t=0)}.
\end{equation}
Since each controlled application of $\mathcal{A}_q$ contains a single controlled use of the block-encoding unitary $U$,
the number of applications of $U$ required to achieve error $\epsilon_V$ (with failure probability at most $\delta$) is
\begin{equation}
\label{eq:U_complexity_readout}
N_U=\mathcal{O}\!\left(N_V(t=0)\,
\frac{1}{\epsilon_V}\,\log\frac{1}{\delta}\right).
\end{equation}

To make the dependence on $W$ explicit, assume a homogeneous $W/n$-dimensional grid ($d$ for Black--Scholes and $2d$ for Heston) with $n$ qubits per dimension, hence
\begin{equation}
N_{\mathrm{grid}} = 2^{W}.
\end{equation}
If $|V(t=0,\vec S_i,\vec v_i)|=\mathcal{O}(1)$ on the grid (in particular, bounded away from scaling with $n$ and $W$), then
\begin{equation}
N_V(t=0)=\left(\sum_i |V(t=0,\vec S_i,\vec v_i)|^2\right)^{1/2}
=\mathcal{O}\!\left(\sqrt{N_{\mathrm{grid}}}\right)
=\mathcal{O}\!\left(2^{\frac{W}{2}}\right),
\end{equation}
and therefore
\begin{equation}
\label{eq:U_complexity_readout_nW}
N_U=\mathcal{O}\!\left(2^{\frac{W}{2}}\,
\frac{1}{\epsilon_V}\,\log\frac{1}{\delta}\right).
\end{equation}

In Appendix~\ref{section:evoltuion non conservative} the evolution stage attains an exponential improvement in the dimension $d$
(while remaining polynomial in $n$). Although the PDE evolution can be simulated efficiently, the retrieval in Eq.~\eqref{eq:readout-recall+norm} requires estimating a single amplitude and multiplying by $N_V(t=0)=\Theta(2^{W/2})$. This limits the end-to-end advantage, but still
permits a meaningful speedup in $d$ and $n$.

\section{Numerical methods and classical baselines}
\label{appendix:numerical methods}
Throughout this Appendix we assume the reduced model is already
homogeneous and time independent, as stated in Eq.~\eqref{eq:MAINhomog-ode}.
After spatial semi-discretization in $d$ dimensions, we obtain the
linear ODE system on $N^{d}$ degrees of freedom,
\begin{equation}
\label{eqn:appendix-homog-ode-multidim}
\frac{d u(t)}{d t}=A\,u(t),
\qquad
u(0)=u_{0},
\qquad
A\in\mathbb{R}^{N^{d}\times N^{d}}.
\end{equation}
Here $N=2^{n}$ is the number of grid points per spatial dimension.
The matrix $A$ is sparse due to local finite-difference stencils;
we denote by $\mathfrak{s}$ the row sparsity per dimension, so the
full $d$-dimensional row sparsity scales as $\mathcal{O}(d\mathfrak{s})$.

\paragraph{(i) Finite-difference time stepping \cite{LeVeque2007FDM} (implicit Euler, $p=1$).}
We discretise time by $t_k=k\Delta t$, $\Delta t=T/N_{t}$, and compute
$u_k\approx u(t_k)$ using the implicit Euler update
\begin{equation}
\label{eqn:appendix-implicit-euler-p1}
\bigl(I-\Delta t\,A\bigr)u_{k}=u_{k-1},
\qquad
k=1,2,\ldots,N_{t}.
\end{equation}
For diffusive/parabolic discretizations one typically enforces
$\Delta t=\Theta(\Delta x^{2})$ (matching first-order temporal error with second-order spatial error), which implies $N_{t}=\Theta(N^{2})$.
The dominant per-step work is the sparse linear solve / banded algebra,
scaling as $\mathcal{O}(d^2\mathfrak{s}N^{d})$ with memory
$\mathcal{O}(d^2\mathfrak{s}N^{d})$, where $\mathfrak{s}$ is per dimensional sparsity and is defined by the degree of accuracy of the finite difference scheme.
Therefore the overall arithmetic complexity decomposes as
\begin{equation}
\label{eqn:appendix-fd-total-multidim}
\underbrace{\mathcal{O}(d^2\mathfrak{s}N^{d})}_{\text{spatial operator}}
\times
\underbrace{\mathcal{O}(N^{2})}_{\text{time steps}}
=
\mathcal{O}(d^2\mathfrak{s}N^{d+2}).
\end{equation}

\paragraph{(ii) Exponential integration (first-order exponential action) \cite{Higham2008FunctionsOfMatrices}.}
A second classical baseline applies the matrix exponential propagator
$e^{AT}$ to the initial vector, using matrix--exponential action
methods or a product of short-time factors.
A first-order approximation over a short interval is
\begin{equation}
\label{eqn:appendix-exp-short-step-p1}
e^{A\Delta t}
=
I+A\Delta t
+
\mathcal{O}\!\bigl(\|A\|_{\max}^{2}\Delta t^{2}\bigr),
\qquad
\Delta t=\frac{T}{N_{t}}.
\end{equation}
Choosing $\|A\|_{\max}\Delta t=\Theta(1)$ yields
$N_{t}=\Theta(\|A\|_{\max}T)$.
For the second-order finite-difference discretizations considered here,
$\|A\|_{\max}=\Theta(\Delta x^{-2})=\Theta(N^{2})$, where
$\Delta x=\Theta(1/N)$. Hence again $N_{t}=\Theta(N^{2})$.
Each sparse exponential-action step costs
$\mathcal{O}(d^2\mathfrak{s}N^{d})$, giving
\begin{equation}
\label{eqn:appendix-exp-total-multidim}
\underbrace{\mathcal{O}(d^2\mathfrak{s}N^{d})}_{\text{spatial operator}}
\times
\underbrace{\mathcal{O}(N^{2})}_{\text{time steps}}
=
\mathcal{O}(d^2\mathfrak{s}N^{d+2}),
\end{equation}
with memory $\mathcal{O}(d^2\mathfrak{s}N^{d})$.

In summary, for first-order temporal accuracy ($p=1$) and a
$d$-dimensional spatial discretization, both classical baselines
exhibit the same leading scaling: $\mathcal{O}(d^2\mathfrak{s}N^{d})$
for the spatial operator times $\mathcal{O}(N^{2})$ time steps,
i.e., $\mathcal{O}(d^2\mathfrak{s}N^{d+2})$ total work.

\subsection{SDE-based Monte Carlo baselines}
\label{subsec:sde-mc-classical-baselines}

The estimates above are grid-based PDE estimates.  Now we introduce
complexity analysis for SDE Monte Carlo methods.  This is a different
classical benchmark.  A PDE method discretizes the state space, while an
SDE Monte Carlo method simulates paths and averages the resulting payoffs.
Thus the full $N^d$ spatial grid does not appear in the Monte Carlo cost.

To compare the accuracy dependence with the PDE estimates, we use the same
accuracy scale as before.  For a second-order finite-difference
discretization, the spatial error satisfies
\begin{equation}
\label{eq:sde-baseline-eps-n}
\varepsilon \asymp N^{-2}=2^{-2n},
\end{equation}
where $N=2^n$ is the number of grid points per spatial direction
\cite{LeVeque2007FDM}; $n$ is the number of qubits.  This relation is only an accuracy
parametrization. 

Consider first a general SDE simulated by a weak order-one
scheme.  If the time step is $h$, the weak bias is $\mathcal{O}(h)$.
Therefore, bias $\mathcal{O}(\varepsilon)$ requires
$h=\mathcal{O}(\varepsilon)$, which gives
$\mathcal{O}(\varepsilon^{-1})$ time steps per path.  In addition, plain
Monte Carlo requires $\mathcal{O}(\varepsilon^{-2})$ independent paths to
obtain sampling error $\mathcal{O}(\varepsilon)$.  Together these two
factors give the standard Euler path-simulation scaling
$\mathcal{O}(\varepsilon^{-3})$ \cite{10.1287/opre.1070.0496}.

The remaining dimension-dependent factor comes from generating correlated
Brownian increments.  For dense Brownian correlation, each increment requires
a dense $d\times d$ loading operation, giving $\mathcal{O}(d^2)$.  For a
rank-\(q\) factor correlation representation, the loading is $d\times q$,
giving $\mathcal{O}(dq)$.  For diagonal correlation, only componentwise
scaling is needed, giving $\mathcal{O}(d)$
\cite{RebonatoJaeckel2000Correlation,Brigo2002CorrelationRankReduction}.
The resulting general SDE Monte Carlo costs are shown in
Table~\ref{tab:sde-mc-costs}.

\begin{table}[h!]
\centering
\renewcommand{\arraystretch}{1.25}
\begin{tabular}{lccc}
\toprule
\textbf{Brownian corr. structure}
&
\textbf{Cost per Brownian step}
&
\textbf{Cost in $\varepsilon$}
&
\textbf{Cost using $\varepsilon\asymp 2^{-2n}$}
\\
\midrule
Dense correlation
&
$\mathcal{O}(d^2)$
&
$\mathcal{O}(d^2\varepsilon^{-3})$
&
$\mathcal{O}(d^2 2^{6n})$
\\
$q$-factor correlation
&
$\mathcal{O}(dq)$
&
$\mathcal{O}(dq\,\varepsilon^{-3})$
&
$\mathcal{O}(dq\,2^{6n})$
\\
Diagonal correlation
&
$\mathcal{O}(d)$
&
$\mathcal{O}(d\varepsilon^{-3})$
&
$\mathcal{O}(d\,2^{6n})$
\\
\bottomrule
\end{tabular}
\caption{General SDE case: plain pathwise Monte Carlo complexity for a weak-order-one time discretization.  The factor $\varepsilon^{-3}$ comes from $\varepsilon^{-1}$ time steps per path and $\varepsilon^{-2}$ Monte Carlo samples.}
\label{tab:sde-mc-costs}
\end{table}

The multivariate Black--Scholes model is a special case.  In this model the
terminal asset vector is multivariate lognormal, so terminal samples can be
generated directly without Euler time stepping \cite{BlackScholes1973}.
Therefore the $\varepsilon^{-1}$ time-discretization factor from
Table~\ref{tab:sde-mc-costs} disappears.  The cost is then the usual Monte
Carlo sampling cost $\mathcal{O}(\varepsilon^{-2})$, multiplied by the cost
of sampling the correlated terminal vector.  This gives the exact-terminal
Black--Scholes costs \cite{BOYLE1977323} shown in Table~\ref{tab:sde-exact-bs-costs}.

\begin{table}[h!]
\centering
\renewcommand{\arraystretch}{1.25}
\begin{tabular}{lccc}
\toprule
\textbf{Brownian corr. structure}
&
\textbf{Cost per Brownian step}
&
\textbf{Cost in $\varepsilon$}
&
\textbf{Cost using $\varepsilon\asymp 2^{-2n}$}
\\
\midrule
Dense correlation
&
$\mathcal{O}(d^2)$
&
$\mathcal{O}(d^2\varepsilon^{-2})$
&
$\mathcal{O}(d^2 2^{4n})$
\\
$q$-factor correlation
&
$\mathcal{O}(dq)$
&
$\mathcal{O}(dq\,\varepsilon^{-2})$
&
$\mathcal{O}(dq\,2^{4n})$
\\
Diagonal correlation
&
$\mathcal{O}(d)$
&
$\mathcal{O}(d\varepsilon^{-2})$
&
$\mathcal{O}(d\,2^{4n})$
\\
\bottomrule
\end{tabular}
\caption{Multivariate Black--Scholes case: Monte Carlo complexity when the terminal lognormal distribution is sampled exactly.  The time-discretization factor is absent, so the accuracy dependence is $\varepsilon^{-2}$ rather than $\varepsilon^{-3}$.}
\label{tab:sde-exact-bs-costs}
\end{table}

Finally, multilevel Monte Carlo (MLMC) gives a sharper SDE benchmark when the
standard MLMC rate assumptions hold.  In the Euler-based setting, these are:
weak bias $\mathcal{O}(h)$, variance of level corrections
$\mathcal{O}(h)$, and level-sample cost proportional to $h^{-1}$ times the
Brownian loading cost.  Under these assumptions, MLMC reduces the plain Euler
Monte Carlo complexity from $\mathcal{O}(\varepsilon^{-3})$ to
$\mathcal{O}(\varepsilon^{-2}\operatorname{polylog}(1/\varepsilon))$
\cite{10.1287/opre.1070.0496}.  Using
\eqref{eq:sde-baseline-eps-n} gives Table~\ref{tab:sde-mlmc-costs}.

\begin{table}[h!]
\centering
\renewcommand{\arraystretch}{1.25}
\begin{tabular}{lccc}
\toprule
\textbf{Brownian corr. structure}
&
\textbf{Cost per Brownian step}
&
\textbf{Cost in $\varepsilon$}
&
\textbf{Cost using $\varepsilon\asymp 2^{-2n}$}
\\
\midrule
Dense correlation
&
$\mathcal{O}(d^2)$
&
$\mathcal{O}\!\left(d^2\varepsilon^{-2}\operatorname{polylog}(1/\varepsilon)\right)$
&
$\mathcal{O}\!\left(d^2 2^{4n}\operatorname{poly}(n)\right)$
\\
$q$-factor correlation
&
$\mathcal{O}(dq)$
&
$\mathcal{O}\!\left(dq\,\varepsilon^{-2}\operatorname{polylog}(1/\varepsilon)\right)$
&
$\mathcal{O}\!\left(dq\,2^{4n}\operatorname{poly}(n)\right)$
\\
Diagonal correlation
&
$\mathcal{O}(d)$
&
$\mathcal{O}\!\left(d\,\varepsilon^{-2}\operatorname{polylog}(1/\varepsilon)\right)$
&
$\mathcal{O}\!\left(d\,2^{4n}\operatorname{poly}(n)\right)$
\\
\bottomrule
\end{tabular}
\caption{MLMC case: complexity under the standard Euler-based MLMC assumptions.  These assumptions include weak bias $\mathcal{O}(h)$, level-correction variance $\mathcal{O}(h)$, and level cost proportional to $h^{-1}$ times the Brownian-loading cost.}
\label{tab:sde-mlmc-costs}
\end{table}

The MLMC rates are standard for sufficiently regular SDEs and payoff
functionals.  They may require modification for discontinuous payoffs,
barrier events, hitting-time functionals, or degenerate stochastic-volatility
dynamics.

\section{Coordinate operator under discrete Fourier transform}\label{appendix:fourier to discrete fourier}

We pass from the continuous Fourier transform to the discrete one and show the induced half-grid permutation. We firstly define the standard continuous Fourier transformation
\begin{equation}
F(\eta)=\int_{\mathbb{R}} f(p)\,e^{-i\eta p}\,dp,
\qquad
f(p)\;=\;\frac{1}{2\pi}\int_{\mathbb{R}} F(\eta)\,e^{i\eta p}\,d\eta.
\end{equation}
We use uniform grids
\begin{equation}
\begin{gathered}
p_n=-\pi R+n\Delta p,\qquad \Delta p=\frac{2\pi R}{N},\qquad n=0,\dots,N-1,\\
\eta_k=\frac{k-\frac{N}{2}}{R},\qquad \Delta\eta=\frac{1}{R},\qquad
\Delta p\,\Delta\eta=\frac{2\pi}{N},\\
p_n\,\eta_k
= \frac{2\pi}{N}
\left(n-\frac{N}{2}\right)
\left(k-\frac{N}{2}\right).
\end{gathered}
\end{equation}

The discrete Fourier transform reads
\begin{equation}
F_k=\sum_{n=0}^{N-1} f_n\,e^{-i\frac{2\pi}{N}kn},
\qquad
f_n\;=\;\frac{1}{N}\sum_{k=0}^{N-1} F_k\,e^{+i\frac{2\pi}{N}kn},
\end{equation}
with $f_n:=f(p_n)$. We now derive the half-grid permutation starting from the continuous transform sampled at $\eta_{k+N/2}$ and approximating the integral by the Riemann sum on the grid:
\begin{equation}
F(\eta_{k+N/2})
=\int_{\mathbb{R}} f(p)\,e^{-i\eta_{k+N/2} p}\,dp
\;\approx\; \Delta p\sum_{n=0}^{N-1} f(p_n)\,e^{-i\eta_{k+N/2} p_n}.
\end{equation}
Insert the definitions $\eta_{k+N/2}=k/R$ and $p_n=-\pi R+n(2\pi R/N)$:
\begin{equation}
\eta_{k+N/2}\,p_n
=\frac{k}{R}\!\left(-\pi R+n\frac{2\pi R}{N}\right)
\;=\;-k\pi+\frac{2\pi}{N}kn.
\end{equation}
Hence the exponential factor becomes
\begin{equation}
e^{-i\,\eta_{k+N/2}p_n}
=e^{-i(-k\pi+\frac{2\pi}{N}kn)}
\;=\;e^{ik\pi}\,e^{-i\frac{2\pi}{N}kn}
\;=\;(-1)^k\,e^{-i\frac{2\pi}{N}kn}.
\end{equation}
Substituting back and recognizing the DFT sum gives
\begin{equation}
F(\eta_{k+N/2})
\approx(-1)^k\,\Delta p \sum_{n=0}^{N-1} f_n\,e^{-i\frac{2\pi}{N}kn}
\;=\;(-1)^k\,\Delta p\,F_k.
\end{equation}
Finally, using $\Delta p=2\pi R/N$,
\begin{equation}
\begin{gathered}
F(\eta_{k+N/2})
\;\approx\;(-1)^k\,\Delta p\,F_k
\;=\;(-1)^k\,\frac{2\pi R}{N}\,F_k.
\end{gathered}
\end{equation}

\FloatBarrier
\section{Oblivious amplitude amplification for block-encodings}
\label{app:amplification}

In many quantum algorithms, one has access to a unitary $U_A$ which, when
applied with an ancilla projector $\Pi_{\mathrm{succ}}$, implements a scaled
version of a target operator $A$.  Concretely, $U_A$ is an $(\alpha,s,0)$
block-encoding of~$A$ if
\begin{equation}
    (\Pi_{\mathrm{succ}}\otimes I)\,
    U_A\,
    (\Pi_{\mathrm{succ}}\otimes I)
    = \frac{A}{\alpha}.
\end{equation}
Applying $U_A$ to $(\Pi_{\mathrm{succ}}\otimes I)|\psi\rangle$ yields a
superposition of a ``success'' component in $\Pi_{\mathrm{succ}}\otimes I$ and
a failure component orthogonal to it.  The success probability is
\begin{equation}
    p
    =
    \bigl\|
        (\Pi_{\mathrm{succ}}\otimes I)\,
        U_A\,
        (\Pi_{\mathrm{succ}}\otimes I)|\psi\rangle
    \bigr\|^{2}
    =
    \frac{\|A|\psi\rangle\|^{2}}{\alpha^{2}}.
\end{equation}

\emph{Oblivious amplitude amplification} (OAA) is a coherent,
measurement-free method that boosts this success amplitude using only $U_A$,
$U_A^\dagger$, and reflections on the ancilla.
This framework was introduced in~\cite{BCCKS14} and is standard in the
block-encoding and LCU literature.

\begin{proposition}[Oblivious amplitude amplification for block-encodings {\cite{BCCKS14}}]
Let $U_A$ be an $(\alpha,s,0)$ block-encoding of $A$ with respect to a
projector $\Pi_{\mathrm{succ}}$.  Define

\begin{equation}
    R := I - 2(\Pi_{\mathrm{succ}}\otimes I),
    \qquad
    Q := U_A\,R\,U_A^\dagger\,R.
\end{equation}

Then $Q$ acts as a rotation in the two-dimensional space spanned by the
success and failure components.  Writing $p=\sin^{2}\theta$, we have
\begin{equation}
    Q^{k} U_A (\Pi_{\mathrm{succ}}\otimes I)|\psi\rangle
    =
    \sin((2k+1)\theta)\,
    (\Pi_{\mathrm{succ}}\otimes I)\frac{A}{\alpha}|\psi\rangle
    +
    \cos((2k+1)\theta)\,|\mathrm{fail}_k\rangle.
\end{equation}
Thus after
\begin{equation}
    k=\Theta(1/\sqrt{p})   
\end{equation}

iterations, the success amplitude becomes $\Theta(1)$, yielding an amplified
block-encoding of $A$ with normalization parameter
$\alpha_{\mathrm{eff}}=\Theta(\|A\|)$.
Each iteration uses one application each of $U_A$ and $U_A^\dagger$ and a
reflection about the ancilla success subspace.
\end{proposition}

In summary, OAA converts a block-encoding with small success probability~$p$
into one with constant success probability using $O(1/\sqrt{p})$ evaluations of
$U_A$ and $U_A^\dagger$.  Since $\alpha$ incorporates both the magnitude of
$A$ and the initial normalization, the resulting amplified block-encoding has
effective scale $\alpha_{\mathrm{eff}}=\Theta(\|A\|)$.

\FloatBarrier
\section{Static arbitrage-free SSVI regularization used in our experiments}
\label{app:ssvi}
This Appendix summarizes the arbitrage-free SSVI construction of Gatheral--Jacquier. 
In this work SSVI is used solely as a static, arbitrage-preserving regularizer for sparse strike grids and quantum-induced noise in implied-volatility data. 
We do not interpret SSVI as stochastic volatility dynamics.

Option strikes are parameterised using forward log-moneyness $k$. Let $\sigma_{\mathrm{BS}}(k,T)$ denote the Black--Scholes implied volatility and define the
\textit{total implied variance} $w(k,T) := \sigma_{\mathrm{BS}}^2(k,T)\,T $. For fixed $T$, the map $k \mapsto w(k,T)$ is a \emph{slice} of the implied volatility surface. Following \cite{Gatheral02012014}, each maturity slice is regularized using the SVI (SSVI) form

\begin{equation}\label{eq:ssvi_app}
w(k,T)=\frac{\theta_T}{2}\left[1+\rho_{\mathrm{SSVI}}\,\varphi(\theta_T)\,k
+ \sqrt{\big(\varphi(\theta_T)\,k+\rho_{\mathrm{SSVI}}\big)^2 + \big(1-\rho_{\mathrm{SSVI}}^2\big)}\right].
\end{equation}

where $\theta_T := w(0,T)=\sigma_{\mathrm{BS}}^2(0,T)\,T$ is the at-the-money (ATM) total variance,
$\rho_{\mathrm{SSVI}} \in (-1,1)$ controls skew, and $\varphi:\mathbb{R}_+\to\mathbb{R}_+$ governs wing slopes. 
This representation corresponds to the SVI ``natural'' slice with parameters depending on $\theta_T$, thereby coupling slices across maturities through the ATM variance curve. Choosing $\varphi$ from a simple parametric family provides the standard SSVI mechanism for maturity consistency and static no-arbitrage constraints.

A surface is free of static arbitrage if it is free of (i) calendar spread arbitrage and
(ii) butterfly arbitrage. In the SSVI setting, these are ensured by simple constraints on
$\theta_T$ and $\varphi$ \cite{Gatheral02012014}. Since our focus is on a single slice or maturity, we will only consider butterfly arbitrage. Gatheral and Jacquier also give explicit sufficient conditions guaranteeing non-negativity of the
risk-neutral density (equivalently, convexity of call prices in strike). A convenient closed-form set of sufficient conditions rules out butterfly arbitrage and enforces Lee's wing bound \cite{Gatheral02012014,LeeMoment2004}.

In this paper, SSVI is used purely as an \emph{arbitrage-free regularization} of sparse market/model quotes. It provides a smooth, convex benchmark implied volatility surface against which we compare the Heston-PDE-generated implied volatilities. We do not use SSVI as a dynamical model. For Heston-like dynamics, we use the following functional form. For the smoothing of the volatility curves, consider the function $\varphi$ defined by $\varphi(\theta)
:=
\frac{1}{\lambda \theta}
\left(
1 - \frac{1 - e^{-\lambda \theta}}{\lambda \theta}
\right),
\qquad \lambda>0,$ for $\theta>0$. A direct computation yields
\begin{equation}
\partial_\theta\big(\theta \varphi(\theta)\big)
=
\frac{e^{-\lambda \theta}\big(e^{\lambda \theta}-1-\lambda \theta\big)}
{\lambda^2 \theta^2}
>0,
\qquad \theta>0,  
\end{equation}

so that the map $\theta \mapsto \theta \varphi(\theta)$ is strictly increasing on $(0,\infty)$.

Moreover, the normalized derivative admits the explicit representation

\begin{equation}
\frac{\partial_\theta\big(\theta \varphi(\theta)\big)}{\varphi(\theta)}
=
\frac{1-(1+\lambda\theta)e^{-\lambda\theta}}
{e^{-\lambda\theta}+\lambda\theta-1}.
\end{equation}
For any $\lambda>0$, this function is strictly decreasing on $(0,\infty)$ and satisfies $\lim_{\theta\downarrow 0}
\frac{\partial_\theta\big(\theta \varphi(\theta)\big)}{\varphi(\theta)}
=
1.$ These properties ensure that the parameterisation exhibits the monotonicity and short-maturity behavior required for an admissible Heston-like implied volatility smile.


\subsection{Durrleman's $g$--function and butterfly arbitrage}
\label{app:durrlemang}
To make the butterfly-arbitrage check explicit, we recall the standard
Durrleman / Gatheral--Jacquier factorization for a single implied-volatility slice. We write $N$ and $\phi$ for the standard normal CDF and pdf:
\begin{equation}
N(x)=\int_{-\infty}^x \frac{1}{\sqrt{2\pi}}e^{-u^2/2}\,du,
\qquad
\phi(x)=\frac{1}{\sqrt{2\pi}}e^{-x^2/2}. 
\end{equation}

Use $\phi$ for the standard normal density and reserve $\varphi(\theta)$
for the SSVI wing function. For $T>0$ and forward price $F>0$, consider an (undiscounted) forward call slice written in implied-total-variance coordinates
\begin{equation}
C(k)=C_{BS}(k,w(k)),\qquad w:\mathbb R\to(0,\infty),
\end{equation}
where $C_{BS}(k,w):=F\big(N(d_+(k,w))-e^kN(d_-(k,w))\big)$ and
\begin{equation}
d_\pm(k,w):=-\frac{k}{\sqrt w}\pm\frac{\sqrt w}{2}.
\end{equation}
Assume $w\in C^2(\mathbb R)$ and the mild tail condition $\lim_{k\to+\infty}d_+(k,w(k))=-\infty$
(to exclude missing mass at infinity \cite{Gatheral02012014, MartiniSVI}). See also \cite{BabakDearbitrageweaksmile} for a practitioner-oriented discussion.

The function $g(k)$ and the equivalence between density nonnegativity and $g(k)\ge 0$ (with a mild tail condition) are standard. For completeness we include a short self-contained derivation.

\begin{proposition}[Durrleman / Gatheral--Jacquier butterfly criterion]
Define
\begin{equation}\label{eq:gdef-app}
g(k):=\Big(1-\frac{k w'(k)}{2w(k)}\Big)^2
-\frac{(w'(k))^2}{4}\Big(\frac{1}{w(k)}+\frac14\Big)
+\frac{w''(k)}{2}.
\end{equation}
Then, for $K=Fe^k$,
\begin{equation}\label{eq:CKK-factor-app}
\frac{\partial^2 C}{\partial K^2}(K)
=\frac{F\,\phi(d_+(k,w(k)))}{K^2\sqrt{w(k)}}\,g(k),
\end{equation}
so $\partial_{KK}C(K)\ge 0$ for all $K>0$ whenever $g(k)\ge 0$
for all $k\in\mathbb R$, together with the tail condition above.
\end{proposition}

\begin{proof}
\emph{(i) Change-of-variable identities.} With $k=\log(K/F)$ and $F$ fixed,
\begin{equation}
\frac{d}{dK}=\frac{1}{K}\frac{d}{dk},
\qquad
\frac{d^2}{dK^2}=\frac{1}{K^2}\Big(\frac{d^2}{dk^2}-\frac{d}{dk}\Big),
\end{equation}
hence
\begin{equation}\label{eq:CKK-bridge-app}
\frac{\partial^2 C}{\partial K^2}(K)=\frac{1}{K^2}\big(C''(k)-C'(k)\big).
\end{equation}

\emph{(ii) Two Black--Scholes identities.} First,
\begin{equation}\label{eq:pdfid-app}
\phi(d_+(k,w))=e^k\phi(d_-(k,w)),
\end{equation}
since $d_+=d_-+\sqrt w$ and $\phi(x)\propto e^{-x^2/2}$.
Second, the following partial derivatives (holding the other variable fixed) are standard
and obtained by direct differentiation using \eqref{eq:pdfid-app}:
\begin{align}
C_k(k,w)&=-Fe^k N(d_-),\label{eq:Ck-app}\\
C_{kk}(k,w)&=-Fe^k N(d_-)+\frac{F}{\sqrt w}\phi(d_+),\label{eq:Ckk-app}\\
C_w(k,w)&=\frac{F}{2\sqrt w}\phi(d_+),\label{eq:Cw-app}\\
C_{kw}(k,w)&=-F\phi(d_+)\Big(\frac{k}{2w^{3/2}}-\frac{1}{4\sqrt w}\Big),\label{eq:Ckw-app}\\
C_{ww}(k,w)&=\frac{F\phi(d_+)}{2}\Big(\frac{k^2}{2w^{5/2}}-\frac{1}{2w^{3/2}}-\frac{1}{8w^{1/2}}\Big).\label{eq:Cww-app}
\end{align}

\emph{(iii) Compose with the smile $w=w(k)$.} By the chain rule,
\begin{equation}
C'(k)=C_k+C_w w',\qquad
C''(k)=C_{kk}+2C_{kw}w'+C_{ww}(w')^2+C_w w'',   
\end{equation}

hence
\begin{equation}
C''(k)-C'(k)=(C_{kk}-C_k)+(2C_{kw}-C_w)w'+C_{ww}(w')^2+C_w w''.   
\end{equation}

Substituting \eqref{eq:Ck-app}--\eqref{eq:Cww-app} and simplifying yields the factorization
\begin{equation}
C''(k)-C'(k)=\frac{F\,\phi(d_+(k,w(k)))}{\sqrt{w(k)}}\,g(k),
\end{equation}

where $g$ is exactly \eqref{eq:gdef-app}. Combining with \eqref{eq:CKK-bridge-app} gives
\eqref{eq:CKK-factor-app}. The prefactor in \eqref{eq:CKK-factor-app} is strictly positive
for $K>0$, so $\partial_{KK}C\ge0$ is equivalent to $g\ge0$, with the stated tail condition
ensuring no probability mass escapes at infinity (as in Gatheral--Jacquier).
\end{proof}


\end{document}